\documentclass[twoside,leqno]{article}

\usepackage[letterpaper]{geometry}

\usepackage{ltexpprt}
\usepackage{hyperref}

\usepackage{appendix}
\usepackage{ifpdf}
\usepackage{algorithmic}
\usepackage{subfigure}
\usepackage{wrapfig}
\usepackage{cite}
\usepackage{url}
\usepackage{hyperref}
\usepackage{graphicx}

\usepackage[textsize=scriptsize]{todonotes}

\usepackage{commands-tam}

\newcommand{\calS}{\mathcal{S}}

\newcommand{\proddom}[1]{\mathcal{A}^{dom}[\mathcal{#1}]}
\newcommand{\termdom}[1]{\mathcal{A}_{\Box}^{dom}[\mathcal{#1}]}

\newtheorem{observation}{Observation}[section]
\newtheorem{definition}{Definition}[section]
\newtheorem{claim}{Claim}[section]

\begin{document}

\title{\Large The Need for Seed (in the abstract Tile Assembly Model)}
\author{
Andrew Alseth%
  \thanks{Department of Computer Science and Computer Engineering, University of Arkansas,
    \protect\url{andrew.alseth@gmail.com}
    Supported in part by National Science Foundation Grant CAREER-1553166.}
\and
  Matthew J. Patitz%
    \thanks{Department of Computer Science and Computer Engineering, University of Arkansas,
      \protect\url{patitz@uark.edu}
      Supported in part by National Science Foundation Grant CAREER-1553166.}
}

\date{}

\maketitle


\fancyfoot[R]{\scriptsize{Copyright \textcopyright\ 2023 by SIAM\\
Unauthorized reproduction of this article is prohibited}}





\begin{abstract} \small\baselineskip=9pt In the abstract Tile Assembly Model (aTAM) square tiles self-assemble, autonomously binding via glues on their edges, to form structures. Algorithmic aTAM systems can be designed in which the patterns of tile attachments are forced to follow the execution of targeted algorithms. Such systems have been proven to be computationally universal as well as intrinsically universal (IU), a notion borrowed and adapted from cellular automata showing that a single tile set exists which is capable of simulating all aTAM systems (FOCS 2012). The input to an algorithmic aTAM system can be provided in a variety of ways, with a common method being via the ``seed'' assembly, which is a pre-formed assembly from which all growth propagates. Arbitrary amounts of information can be encoded into seed assemblies by both (1) the types and patterns of glues exposed on their exteriors, and (2) their shapes. Since a common metric by which aTAM systems are measured is their tile complexity (i.e. the number of unique types of tiles they utilize), in order to provide a fair basis for comparison, systems are often designed with seed assemblies consisting of only a single seed tile, a.k.a. single-tile seeds. (For instance, in STOC 2000 and 2001 information theoretically optimal tile complexity was shown possible for the self-assembly of squares.) This requires the transferring of any information that may be encoded in a multi-tile seed assembly into tile complexity. In this paper, we explore this process to show when and how such transformations are possible while ensuring that a derived system with a single-tile seed faithfully replicates the behaviors of the original system.

We first show that a trivial transformation, in which the locations of a multi-tile seed are tiled by ``hard-coded'' tiles that can grow to represent that seed from a single tile, can succeed only if (1) there are not tile locations in the seed such that there exist growth sequences where those locations could block future growth, or (2) an ordering of growth can be enforced for the growth of the seed from a single tile to ensure that such blocking locations are tiled before collisions are possible. However, we show that knowing if this is the case is uncomputable. Therefore, we examine what is possible if the scale factor of the original system is increased and show that all systems with multi-tile seeds can be transformed into systems with single-tile seeds at scale factor 3 (i.e. each tile of the original system is replaced by a $3 \times 3$ square of tiles), such that the transformed systems faithfully replicate the dynamics of the original systems. We also prove that this scale factor is optimal, and that in fact there exist systems with multi-tile seeds for which no systems at scale factors 1 or 2 (or scale factor 3 when a more restrictive form of simulation is required) with single-tile seeds exist that can even produce the same sets of terminal output shapes. Since the scale 3 transformation results in a tile complexity which is proportional to the size of the original tile set plus the size of the multi-tile seed multiplied by the scale factor, we then also provide a transformation that yields an asymptotically optimal tile complexity proportional to the Kolmogorov complexity of the original system and which is based on the IU construction from FOCS 2012. Additionally, we are able to make simple modifications to that construction to provide a single aTAM system which simultaneously and in parallel simulates \emph{all} aTAM systems, and provide a connection between that system and the existence of systems within models other than the aTAM which are IU for the aTAM. This set of results provides a full characterization of the tradeoffs between systems with multi-tile seeds and those with single-tile seeds, which is fundamental to the measure of complexity of aTAM systems.
\end{abstract}

\section{Introduction}

Mathematical models of self-assembly provide high-level abstractions of the self-assembling behaviors of systems that typically function on a molecular level. While many natural systems harness the power of self-assembly, scientists and engineers are rapidly increasing their abilities to design self-assembling systems. One of the most versatile molecules that we have discovered is DNA. While biology utilizes DNA primarily as a storage medium, it is now being recognized for its impressive capabilities for structure-building, augmented by its ease of synthesis, durability, and programmability. Engineers are designing impressive DNA-based self-assembling systems whose complexities are increasing at a nearly exponential rate (e.g.  \cite{woods2019diverse,zhang2020programming,MonaLisa,AndersenBox2009}). Not to be outdone, theoreticians are rapidly expanding the mathematical toolkit for modeling, predicting, and analyzing such systems. The level of granularity in modeling can vary from highly realistic atomic-level modeling to purely mathematical abstractions, and computer science is impacting all levels. In this paper, we utilize computational theory, complexity theory, and algorithm design and analysis to explore a fundamental aspect of mathematical modeling: the ways in which so-called ``seed'' assemblies, i.e. preexisting assemblies that serve as nucleation points for growth, can influence and impact the behaviors of self-assembling systems.

A popular and widely studied mathematical model of self-assembly called the abstract Tile Assembly Model (aTAM) was introduced by Winfree \cite{Winf98} and immediately shown to be capable of complex algorithmic behavior and computationally universal. In the aTAM, the basic components are square tiles that have glues on their edges that allow them to bind together, and as a ``seeded'' model growth begins with tile attachments to an input seed assembly and proceeds by growing that assembly. This is in contrast to ``seedless'' models such as the 2-Handed Assembly Model (2HAM) \cite{AGKS05g,2HAMIU}, a.k.a. Hierarchical Assembly Model, in which growth can begin with the combination of any pair of base components and proceeds by the combination of pairs of base or previously formed components. Prior work has directly compared these two models \cite{Versus,jVersus1}, and although the power of hierarchical assembly grants the 2HAM many advantages, the few results showing advantages for the aTAM resulted from aspects related to the seed structures, particularly due to the geometric hindrance (a.k.a. blocking) they can provide. 

One of the major metrics used to analyze self-assembling systems is called \emph{tile complexity}, which is the number of unique types of tiles that are required to self-assemble a target structure. Various results have exemplified the power of algorithmic self-assembly by demonstrating theoretical constructions that meet information theoretic lower bounds on tile complexity. For instance, the self-assembly of $n \times n$ squares in the aTAM has been shown to be possible with $\Theta(\frac{\log n}{\log \log n})$ tile types\cite{RotWin00,AdChGoHu01}, and general shapes (at large scale factors) with $\Theta(\frac{K(S)}{\log K(S)})$ tile types (where $S$ is the definition of the target shape and $K(S)$ is the Kolmogorov complexity of $S$)\cite{SolWin07}.

The aTAM allows seed assemblies to be arbitrarily complex (as long as they are finite), but if the goal is to measure the complexities of systems, the information contained in the seed assemblies should be factored in along with the tile complexity. However, it isn't immediately clear how to completely quantify the amount of information contained within a seed assembly. Clearly, the amount of information encoded in the glues exposed by an assembly is important. However, information can also be implicitly encoded in the shapes of tiles and assemblies (e.g. \cite{GeoTiles,jDuples,Polygons,Polyominoes,OneTile}), and variations in shapes can allow systems to utilize another tool: geometric hindrance. Since quantifying the information embedded in a seed assembly can be difficult, a standard basis for these complexity results is the requirement that the system have a seed consisting of a single tile (i.e. a single-tile seed). In this way, all information is transferred into the definitions of the tiles, providing a more even foundation for comparison. Such single-tile seeds are the basis of results such as the self-assembly of $n \times n$ squares and also of generals shapes previously mentioned. Of course, given the ability of seed assemblies to encode arbitrary amounts of information, it is also possible to use a constant tile set and vary the seed assemblies to instead shift all of the information to the seed. Several results have been based on this approach, including those related to \emph{intrinsic universality} (IU), a notion of simulation in which one tile assembly system is used to simulate another in a way that attempts to capture the full dynamics of the simulated system, modulo only scaling in time and space. In \cite{IUSA} it was shown that the aTAM is intrinsically universal, which means that there is a single universal aTAM tile set (and functions that specify how to arrange those tiles into seed assemblies and interpret blocks of them as individual tiles in the simulated systems) that can be used to simulate \emph{any} aTAM system. This is a powerful closure property of the model, and the notion of intrinsic simulation has been utilized to compare and contrast the powers of various models and classes of systems (e.g. \cite{DirectedNotIU,j2HAMIU,temp1notIU,WoodsMeunierSTOC}). When designing and utilizing IU tile sets, the tile set is constant across differing simulations, and therefore the information serving as input to the system comes completely from the structure and glues of the seed assembly.
Slight variations to the aTAM have also resulted in models in which there are alternative methods for providing information to systems to seed their growth. These include \emph{temperature programming} in which the temperature parameter of a system can be programmed to change following a prescribed, and arbitrarily complex, series of values that causes a growing structure to periodically grow and/or partially break apart \cite{AGKS05g,KS07,SummersTemp}, \emph{concentration programming} in which a very precise concentration value can be specified for each tile type in order to influence the probabilities of attachment of specific tile types \cite{KaoSchS08,Dot10}, and others.

In this paper, we present a wide array of results that serve to quantify the types and magnitudes of impacts that seed structures can have on aTAM systems. We first provide a few simple results and observations to show (1) systems that can encode arbitrary amounts of information in seeds and utilize ``Garden of Eden'' assemblies (which are assemblies that have no pathway for growth and can only stably exist if they completely form instantaneously) in Section \ref{sec:seed-encoding-GoE}, (2) ``blocking'' (which occurs when potential placement of later tiles is prevented by the prior placement of earlier tiles) by tiles in a seed assembly can impact the dynamics of systems in important ways but that it is uncomputable, in general, to determine if one or any seed tile locations have the potential to block growth (Observation \ref{obs:uncomp-blocking}), (3) the same sets of shapes of the assemblies produced by even some relatively simple multi-tile seed systems cannot be produced by systems with single-tile seeds if they are not allowed to be scaled up in size (i.e. if each single tile is replaced by an $n \times n$ square of tiles, we say that the system is scaled by factor $n$) in Section \ref{sec:imp-scale-1}, and (4) that, since the aTAM does not allow for seeds of infinite size, there is an infinite set of shapes that cannot self-assemble in any aTAM system (Observation \ref{obs:infinite-seed}).

Given the utility of transferring information encoded in seed assemblies into tile complexity, by trading multi-tile seeds for single-tile seeds, we explore when and how it is possible to do so.  In order to fully explore the capabilities and limitations of converting systems with multi-tile seeds to those with single-tile seeds, we provide definitions for new notions of simulation (Section \ref{sec:simulation_def}).
In standard definitions, both the simulated system and the simulating system start from the same assembly (modulo scale factor and interpreted through a mapping function). In these new definitions, the simulator is allowed to begin with seed assemblies as small as a single tile, and then are allowed to grow assemblies representing the seed structures. In the most permissive definition (which we call ``shape-simulation''), we only stipulate that both systems have to produce terminal assemblies of the same shapes (modulo scale factor and mapping). In the most restrictive definition (which we call ``seed-first-simulation''), we require that the simulating system completely grow an assembly mapping to the seed structure of the original system before allowing any of the rest of the growth permitted by that system to occur.
Using these notions of simulation, we prove the tightest results possible.
We prove that even with the most relaxed version, shape-simulation, there are aTAM systems with multi-tile seeds that no systems with single-tile seeds can correctly shape-simulate at scale factors 1 or 2, or also at scale factor 3 with a technical restriction applied (i.e. ``cheating fuzz'', as defined in Section \ref{sec:simulation_def}, is not allowed) (Theorem \ref{thm:multi-imposs}). However, using even the most restrictive notion of simulation, seed-first-simulation, we prove that any aTAM system with a multi-tile seed can be seed-first-simulated by a system with a single-tile seed at scale factor 3 with the technical restriction removed (i.e. cheating fuzz is allowed) (Theorem \ref{thm:scale3-sim}) or at scale factor 4 with the technical restriction once again applied (i.e. cheating fuzz is not allowed) (Theorem \ref{thm:scale4-sim}). Using the restrictive notion of seed-first-simulation makes it possible to guarantee the full seed-representing assembly is grown in the simulator before any outward growth can occur, making it behave identically to the simulated system after an initial seed-building phase.

While we prove that the scale 3 and 4 simulations achieve the lower bound for scale factor, they pay the price in terms of tile complexity, which is $O(|\sigma|+|T|)$ for the simulation of a system with tile set $T$ and seed assembly $\sigma$. In order to achieve optimality for the tile complexity metric, our next result (Theorem \ref{thm:scaled-simulation}) proves that the tile complexity can be reduced to $O(\frac{K(\calT)}{\log(K(\calT)})$ (which, by an information theoretic argument, is the lower bound) for the seed-first-simulation of an arbitrary aTAM system $\calT$ at the trade-off of a (massive) increase in scale factor (which is proportional to the running time of a relatively complex Turing machine).

Although the tradeoff in scale factor is immense for the previous construction, beyond approaching optimal tile complexity, it provides a basis for additional theoretically interesting results. As a module of that construction, we use a (minimally modified) version of the intrinsically universal aTAM tile set from \cite{IUSA}. This machinery allows us to extend the construction to first present a new construction that simultaneously and in parallel simulates all aTAM systems (Theorem \ref{thm:simult-sim}). This, of course, requires a modified definition of simulation to take into account the fact that no fixed value for a scale factor could suffice to simulate arbitrary aTAM systems, so we define a type of ``mixed-scale-simulation''. Then, for our final result (Theorem \ref{clm:cross-model-IU}), we show how we can make use of the ability of systems in a different model to simulate systems of the aTAM by providing a construction that utilizes ``nested simulations'', allowing us to make connections between models that can simulate arbitrary systems in the aTAM and intrinsic universality for the aTAM (notions we also show are not necessarily correlated).

In summary, we prove a wide assortment of results that expose the impacts of seed assemblies on the dynamics and complexities of aTAM systems and provide tight bounds that show how it is possible to replace multi-tile seeds with single-tile seeds and what the tradeoffs are.

The paper is laid out as follows. Section \ref{sec:definitions} contains the definitions of, and related to, the aTAM, as well as definitions for the various types of simulations used throughout the paper. Section \ref{sec:basic} contains a set of simple results and observations that lay the foundation for the more complex results to follow. Section \ref{sec:single-tile-limits} sketches the result showing the lower bound for the scale factor required to simulate systems with multi-tile seeds by those with single-tile seeds.
Section \ref{sec:min-scaled-sims} gives overviews of the constructions that provide a tight upper bound for the scale factor of such simulations. In Section \ref{sec:scaled-sim} we present the results showing simulation by systems with single-tile seeds using optimal tile complexity, as well as simultaneous simulation of all aTAM systems, and one relating to intrinsic universality and the aTAM. Then, the Technical Appendix, Section \ref{sec:append}, contains proofs and technical details omitted from the other sections.

\section{Definitions}\label{sec:definitions}
In this section we provide definitions of the aTAM and also related to the simulation of one tile assembly system by another.

\subsection{The abstract Tile Assembly Model}
\label{sec:tam-informal}

This section gives a brief informal sketch of the abstract Tile Assembly Model (aTAM) \cite{Winf98} and uses notation from \cite{RotWin00} and \cite{jSSADST}. For more formal definitions and additional notation, see \cite{RotWin00} and \cite{jSSADST}.

A \emph{tile type} is a unit square with four sides, each consisting of a \emph{glue label} which is often represented as a finite string.
An aTAM system has a finite set $T$ of tile types, but an infinite number of copies of each tile type, with each copy being referred to as a \emph{tile}.
A \emph{glue function} is a symmetric mapping from pairs of glue labels to a non-negative integer value which represents the strength of binding between those glues.
An \emph{assembly}
is a positioning of tiles on the integer lattice $\Z^2$, described  formally as a partial function $\alpha:\Z^2 \dashrightarrow T$. 
Let $\mathcal{A}^T$ denote the set of all assemblies of tiles from $T$, and let $\mathcal{A}^T_{< \infty}$ denote the set of finite assemblies of tiles from $T$.
We write $\alpha \sqsubseteq \beta$ to denote that $\alpha$ is a \emph{subassembly} of $\beta$, which means that $\dom\alpha \subseteq \dom\beta$ and $\alpha(p)=\beta(p)$ for all points $p\in\dom\alpha$.
We write $\alpha \setminus \beta$ to denote the assembly $\alpha$ without any of the tiles in locations in $\beta$, i.e. the result of starting with $\alpha$ and removing tiles from any locations which are in both $\alpha$ and $\beta$.
Two adjacent tiles in an assembly \emph{interact}, or are \emph{attached}, if the glue labels on their abutting sides match. 
Each assembly induces a \emph{binding graph}, a grid graph whose vertices are tiles, with an edge between two tiles if they interact.
The assembly is \emph{$\tau$-stable} if every cut of its binding graph has strength at least~$\tau$, where the strength   of a cut is the sum of all of the individual glue strengths in the cut.


A \emph{tile assembly system} (TAS) is a 3-tuple $\calT = (T,\sigma,\tau)$, where $T$ is a finite set of tile types, $\sigma:\Z^2 \dashrightarrow T$ is a finite $\tau$-stable \emph{seed assembly},
and $\tau$ is the \emph{temperature} parameter (a.k.a. \emph{binding threshold}).
Given an assembly $\alpha$, the \emph{frontier}, $\frontiert{\alpha}$, is the set of locations to which tiles can $\tau$-stably attach.
An assembly $\alpha$ is \emph{producible} if either $\alpha = \sigma$ or if $\beta$ is a producible assembly and $\alpha$ can be obtained from $\beta$ by the stable binding of a single tile to a location in $\frontiert{\beta}$. In this case we write $\beta\to_1^\calT \alpha$ (to mean $\alpha$ is producible from $\beta$ by the attachment of one tile), and we write $\beta\to^\calT \alpha$ if $\beta \to_1^{\calT*} \alpha$ (to mean $\alpha$ is producible from $\beta$ by the attachment of zero or more tiles).
An \emph{assembly sequence} in a TAS $\calT$ is a (finite or infinite) sequence $\vec{\alpha} = (\alpha_0,\alpha_1,...)$ of assemblies in which each $\alpha_{i+1}$ is obtained from $\alpha_i$ by the addition of one tile, i.e. $\alpha_i \to_1^{\calT} \alpha_{i+1}$.

We let $\prodasm{\calT}$ denote the set of producible assemblies of $\calT$.
An assembly $\alpha$ is \emph{terminal} if no tile can be $\tau$-stably attached to it, i.e. $|\frontiert{\alpha}| = 0$.
We let   $\termasm{\calT} \subseteq \prodasm{\calT}$ denote  the set of producible, terminal assemblies of $\calT$.
A TAS $\calT$ is \emph{directed} if $|\termasm{\calT}| = 1$. Hence, although a directed system may be nondeterministic in terms of the order of tile placements, it is deterministic in the sense that exactly one terminal assembly can be produced.

We define the \emph{producible shapes} of an aTAM system $\calT$, denoted $\proddom{T}$, as the set of shapes (i.e. domains) of all producible assemblies of $\calT$. More formally, $\proddom{T} = \{\dom \alpha \mid \alpha \in \prodasm{T}\}$. Similarly, we define the \emph{terminal shapes} of an aTAM system $\calT$, denoted $\termdom{T}$, as the set of shapes (i.e. domains) of all terminal assemblies of $\calT$. More formally, $\termdom{T} = \{\dom \alpha \mid \alpha \in \termasm{T}\}$. We say that aTAM system $\calT$ \emph{self-assembles shape $S$} if and only if $|\termdom{T}| = 1$ and $S \in \termdom{T}$, that is, $\calT$ produces terminal assemblies of only a single shape, which is $S$.

\begin{definition}[TAS equivalence]\label{def:tas-equiv}
Given two aTAM systems $\calT$ and $\mathcal{S}$, and assembly $\gamma$, we say that $\calT$ and $\mathcal{S}$ are \emph{equivalent modulo $\gamma$} if and only if for every producible assembly $\alpha \in \prodasm{\calT}$ there exists a producible assembly $\beta \in \prodasm{\mathcal{S}}$ such that $(\alpha \setminus \gamma) = (\beta \setminus \gamma)$, and vice versa (i.e. for every producible assembly $\beta \in \prodasm{\mathcal{S}}$ there exists a producible assembly $\alpha \in \prodasm{\mathcal{T}}$ such that $(\beta \setminus \gamma) = (\alpha \setminus \gamma)$).
\end{definition}

Note that the notion of equivalence between aTAM systems is quite strict, requiring all of the same tile types to be used outside of the region of $\gamma$, and not allowing for any scaling factor. For more general notions of equivalence between systems, see Section \ref{sec:simulation_def}.

\begin{definition}[Strict dependence]\label{def:strict-dep}
Given an aTAM system $\calT$ and sets of locations $L_1,L_2 \subset \mathbb{Z}^2$, we say that $L_2$ \emph{strictly depends upon} $L_1$ if, for every valid assembly sequence in $\calT$, if a tile is placed in some location in $L_2$, a tile was previously placed in some location in $L_1$.
\end{definition}

We define a \emph{path} $p$ as an ordered list of distinct locations in $\mathbb{Z}^2$, and refer to the $i$th location in $p$ as $p[i]$, so that each location $p[i]$, for $0 \le i < |p|-1$, is adjacent to $p[i+1]$.

\begin{definition}[Dependence path]\label{def:dep-path}
Given an aTAM system $\calT$, a producible assembly $\alpha \in \prodasm{\calT}$, assembly sequence $\vec{\alpha}$ which produces $\alpha$, and locations $l_1,l_2 \in \dom{\alpha}$, we say that there is a \emph{dependence path from $l_1$ to $l_2$} if and only if there exists a path $p$ from $1_1$ to $l_2$ in $\dom{\alpha}$ such that, in $\vec{\alpha}$, for each location $p[i]$, for $0 < i < |p|-1$, (1) the tile at location $p[i]$ was placed before the tile in $p[i+1]$, and (2) the tile at location $p[i+1]$ was required to form a bond with the tile at location $p[i]$ in order to attach.
\end{definition}

Intuitively, a dependence path from $l_1$ to $l_2$ is a path formed by sequential tile attachments that leads from $l_1$ to $l_2$. Note that, by definition, a dependence path is directional and therefore a dependence path from $l_1$ to $l_2$ cannot be the same path as a dependence path from $l_2$ to $l_1$.

\begin{lemma}[Dependence paths]\label{lem:depend-path}
Given an aTAM system $\calT$ whose seed consists of a single tile, a producible assembly $\alpha \in \prodasm{\calT}$, and sets of locations $L_1,L_2 \subset \dom{\alpha}$, if $L_2$ strictly depends upon $L_1$, then in each valid assembly sequence of $\calT$ there must be a dependence path from some  $l_1 \in L_1$ to some $l_2 \in L_2$.
\end{lemma}

The proof of Lemma \ref{lem:depend-path} can be found in Section \ref{sec:depend-path-proof}.

\begin{figure}
    \centering
    \includegraphics[width=6.0in]{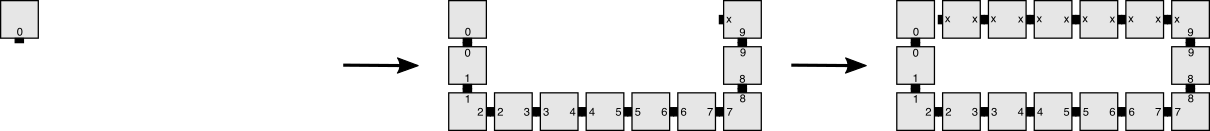}
    \caption{A basic example of blocking by a seed location. (Left) The seed tile, (Middle) Growth from the seed with tiles specific to each location, (Right) Growth back toward the seed using a single, repeating tile type. This row ``crashes'' into the seed (i.e. it would place a tile in the location of the seed tile if that tile weren't there), so the seed location blocks placement of a tile (of the type with $x$ glues on its east and west).}
    \label{fig:blocking}
\end{figure}

\begin{definition}[Blocking]\label{def:blocking}
Given a TAS $\calT$ and producible assembly $\alpha \in \prodasm{\calT}$, a tile $t$ at location $\vec{l}$ in $\alpha$ \emph{blocks} a tile if there exists a valid assembly sequence which begins with $\alpha$ and results in one or more tiles adjacent to $t$ which did not require glues of $t$ to bind to the assembly and to which a tile of a type different than $t$ could bind with $\tau$ strength in location $\vec{l}$ if $t$ was removed.
\end{definition}

See Figure \ref{fig:blocking} for an example of blocking. Note that we still say blocking occurs even if the tile adjacent to $t$ strictly depends upon $t$, meaning it would technically be impossible to remove $t$ and replace it with a tile of another type.

\subsection{Simulation of tile assembly systems}
\label{sec:simulation_def}

First, we give a very brief intuitive definition of what it means for one tile assembly system to simulate another, and then provide more technically detailed definitions related to simulation, especially as it relates to scale factors greater than 1. We define new notions of simulation that are necessary to allow, and capture, the dynamics of simulating systems that don't begin from seed assemblies that represent the full seed assemblies of the systems they are simulating, and thus they have to grow those missing portions which the standard definition of simulation assumes exist at the beginning of the simulation. For the technical definitions related to the standard model of simulation, see \cite{DirectedNotIU}.

Intuitively, simulation of a system $\calT$ by a system $\calS$ requires that there is some scale factor $m \in \Z^+$ such that $m \times m$ squares of tiles in $\calS$ represent individual tiles in $\calT$, and there is a ``representation function'' capable of inspecting assemblies in $\calS$ and mapping them to assemblies in $\calT$. A representation function $R$ takes as input an assembly in $\calS$ and returns an assembly in $\calT$ to which it maps. In order for $\calS$ to correctly simulate $\calT$, it must be the case that for any producible assembly $\alpha \in \prodasm{\calT}$ that there is a corresponding assembly $\beta \in \prodasm{\calS}$ such that $R(\beta) = \alpha$. (Note that there may be more than one such $\beta$.) Furthermore, for any $\alpha' \in \prodasm{\calT}$ which can result from a tile addition to $\alpha$, there exists $\beta' \in \prodasm{\calS}$, where $R(\beta') = \alpha'$, which can result from the addition of one or more tiles to $\beta$, and conversely, $\beta$ can only grow into assemblies which can be mapped into valid assemblies of $\calT$ into which $\alpha$ can grow.

We now present a formal, rigorous definition of what it means for one tile assembly system to ``simulate'' another. Our definitions are based on those of \cite{temp1notIU}, but are adapted to account for the simulating system to grow a representation of the original system's seed rather than beginning from a seed assembly which already represents it. Also, note that a great amount of the complexity required for the definitions arises due to the possible dynamics of simulations with scale factors $> 1$, and that otherwise the mapping of assemblies and equivalence of production and dynamics are much more straightforward.


From this point on, let $T$ be a tile set, and let $m\in\Z^+$.
An \emph{$m$-block supertile} over $T$ is a partial function $\alpha : \Z_m^2 \dashrightarrow T$, where $\Z_m = \{0,1,\ldots,m-1\}$.
Let $B^T_m$ be the set of all $m$-block supertiles over $T$.
The $m$-block with no domain is said to be $\emph{empty}$.
For a general assembly $\alpha:\Z^2 \dashrightarrow T$ and $(x,y)\in\Z^2$, define $\alpha^m_{x,y}$ to be the $m$-block supertile defined by $\alpha^m_{x,y}(i_x, i_y) = \alpha(mx+i_x, my+i_y)$ for $0 \leq i_x,i_y< m$.
For some tile set $S$, a partial function $R: B^{S}_m \dashrightarrow T$ is said to be a \emph{valid $m$-block supertile representation} from $S$ to $T$ if for any $\alpha,\beta \in B^{S}_m$ such that $\alpha \sqsubseteq \beta$ and $\alpha \in \dom R$, then $R(\alpha) = R(\beta)$.
Note that we use the term \emph{macrotile} interchangeably with supertile, to mean the same thing.

For a given valid $m$-block supertile representation function $R$ from tile set~$S$ to tile set $T$, define the \emph{assembly representation function}\footnote{Note that $R^*$ is a total function since every assembly of $S$ represents \emph{some} assembly of~$T$; the functions $R$ and $\alpha$ are partial to allow undefined points to represent empty space.}  $R^*: \mathcal{A}^{S} \rightarrow \mathcal{A}^T$ such that $R^*(\alpha') = \alpha$ if and only if $\alpha(x,y) = R\left(\alpha'^m_{x,y}\right)$ for all $(x,y) \in \Z^2$.
For an assembly $\alpha' \in \mathcal{A}^{S}$ such that $R(\alpha') = \alpha$, $\alpha'$ is said to map \emph{cleanly} to $\alpha \in \mathcal{A}^T$ under $R^*$ if for all non empty blocks $\alpha'^m_{x,y}$, $(x,y)+(u_x,u_y) \in \dom \alpha$ for some $u_x,u_y \in \{-1,0,1\}$ such that $u_x^2 + u_y^2 \leq 1$. In other words, $\alpha'$ may have tiles on supertile blocks representing empty space in $\alpha$, but only if that position is adjacent to a tile in $\alpha$.  We call such growth ``around the edges'' of $\alpha'$ \emph{fuzz} and thus restrict it to be adjacent to only valid supertiles, but not diagonally adjacent (i.e.\ we do not permit \emph{diagonal fuzz}). Additionally, if tiles grow as fuzz into a supertile region which maps to a location in the simulated system in which there is never an incident glue (other than the null glue) from any adjacent tile, we call this \emph{cheating fuzz}. The justification for this distinction is that the goal of an intrinsic simulation is for the simulator to only utilize the dedicated macrotile space to simulate a tile or to compute if a tile may grow into a location, but cheating fuzz would be tile growth into a macrotile location that never has any input glues and therefore no chance of ever becoming a tile. See Figure \ref{fig:S3-10_cheating_fuzz_locations} for an example.

\begin{figure}
    \centering
    \includegraphics[width=0.3\textwidth]{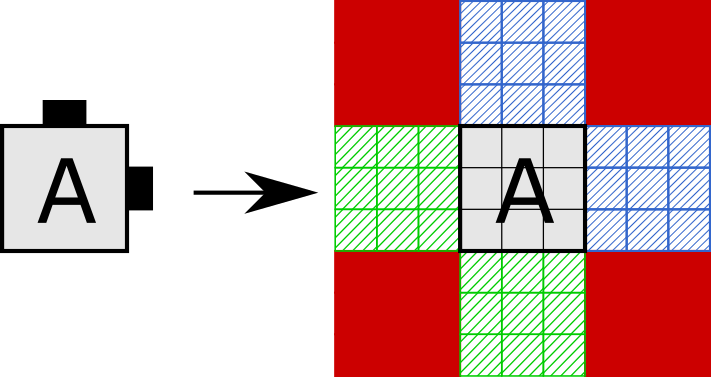}
    \caption{When simulating an aTAM tile `A' with glues on its north and east faces and null glues on south and west faces, although a simulation is allowed to grow fuzz into each of the adjacent north, east, south, and west marcortile locations, we further restrict the locations for ``legal fuzz'' to grow only in the adjacent macrotiles to its north and east, indicated by blue hashed tiles. Locations for placement of ``cheating fuzz'', in adjacent directions in which no glue is presented by A, are indicated by the green hashed tiles.}
    \label{fig:S3-10_cheating_fuzz_locations}
\end{figure}


In the following definitions, let $\mathcal{T} = \left(T,\sigma_T,\tau_T\right)$ be a tile assembly system, let $\mathcal{S} = \left(S,\sigma_S,\tau_S\right)$ be a tile assembly system, and let $R$ be a valid $m$-block representation function $R:B^S_m \rightarrow T$.


\begin{definition}\label{def-shape-sim}
We say that $\calS$ \emph{shape-simulates} $\calT$ if the following conditions hold:
\begin{enumerate}
        \item $\left\{ \dom R^*(\alpha') | \alpha' \in \termasm{\mathcal{S}}\right\} = \termdom{T}$.
        \item For all $\alpha'\in \termasm{\mathcal{S}}$, $\alpha'$ maps cleanly to $R^*(\alpha')$.
\end{enumerate}
\end{definition}

Essentially, for $\mathcal{S}$ to shape-simulate $\calT$, its terminal assemblies must map to the exact same set of domains as those of the terminal assemblies of $\calT$.

\begin{definition}
\label{def-equiv-prod-mod} We say that $\mathcal{S}$ and $\mathcal{T}$ have \emph{equivalent productions modulo $\sigma_\calT$} (under $R$), and we write $\mathcal{S} \Leftrightarrow^{\sigma_\calT}_R \mathcal{T}$ if the following conditions hold:
\begin{enumerate}
        \item $\left\{R^*(\alpha') | \alpha' \in \prodasm{\mathcal{S}}\right\} = \prodasm{\mathcal{T}} \cup \{ \alpha \mid \alpha \sqsubseteq \sigma_\calT\}$.
        \item $\left\{R^*(\alpha') | \alpha' \in \termasm{\mathcal{S}}\right\} = \termasm{\mathcal{T}}$.
        \item For all $\alpha'\in \prodasm{\mathcal{S}}$, $\alpha'$ maps cleanly to $R^*(\alpha')$.
\end{enumerate}
\end{definition}

Note that the definition of equivalent productions modulo the seed assembly differs slightly from the definition of equivalent productions in other papers (e.g. \cite{IUSA,DirectedNotIU,temp1notIU}) since we must allow the producible assemblies of the simulating system to represent subsets of the seed of the simulated system. This is because we are interested in being able to begin with a singly-seeded system that \emph{first} grows a complete representation of the seed of the simulated system, and then continues with its simulation. Alternatively, the next definition can be used to define a notion of simulation in which growth can extend beyond the simulated seed before the full representation of the seed has completed.

\begin{definition}
\label{def-equiv-prod-minus} We say that $\mathcal{S}$ and $\mathcal{T}$ have \emph{equivalent productions minus $\sigma_\calT$} (under $R$), and we write $\mathcal{S} \Leftrightarrow^{-\sigma_\calT}_R \mathcal{T}$ if the following conditions hold:
\begin{enumerate}
        \item For all $\alpha' \in \prodasm{\mathcal{S}}$, all locations of $\alpha'$ that map to a location in $\dom{\sigma_\calT}$ map, under $R$, to the tile of $\calT$ in that location of $\sigma_\calT$.
        \item $\left\{R^*(\alpha') \cup \sigma_\calT | \alpha' \in \prodasm{\mathcal{S}}\right\} = \prodasm{\mathcal{T}}$.
        \item $\left\{R^*(\alpha') | \alpha' \in \termasm{\mathcal{S}}\right\} = \termasm{\mathcal{T}}$.
        \item For all $\alpha'\in \prodasm{\mathcal{S}}$, $\alpha'$ maps cleanly to $R^*(\alpha')$.
\end{enumerate}
\end{definition}

\begin{definition}
\label{def-t-follows-s} We say that $\mathcal{T}$ \emph{follows $\mathcal{S}$ modulo $\sigma_\calT$} (under $R$), and we write $\mathcal{T} \dashv^{\sigma_\calT}_R \mathcal{S}$ if $\alpha' \rightarrow^\mathcal{S} \beta'$, for some $\alpha',\beta' \in \prodasm{\mathcal{S}}$, implies that $R^*(\alpha') \cup \sigma_\calT \to^\mathcal{T} R^*(\beta') \cup \sigma_\calT$.
\end{definition}

The next definition essentially specifies that every time $\mathcal{S}$ simulates an assembly $\alpha \in \prodasm{\mathcal{T}}$, there must be at least one valid growth path in $\mathcal{S}$ for each of the possible next steps that $\mathcal{T}$ could make from $\alpha$ which results in an assembly in $\mathcal{S}$ that maps to that next step.

\begin{definition}
\label{def-s-models-t} We say that $\mathcal{S}$ \emph{models} $\mathcal{T}$ (under $R$), and we write $\mathcal{S} \models_R \mathcal{T}$, if for every $\alpha \in \prodasm{\mathcal{T}}$, there exists $\Pi \subset \prodasm{\mathcal{S}}$ where $\Pi \neq$ \O
and $R^*(\alpha') = \alpha$ for all $\alpha' \in \Pi$, such that, for every $\beta \in \prodasm{\mathcal{T}}$ where $\alpha \rightarrow^\mathcal{T} \beta$, (1) for every $\alpha' \in \Pi$ there exists $\beta' \in \prodasm{\mathcal{S}}$ where $R^*(\beta') = \beta$ and $\alpha' \rightarrow^\mathcal{S} \beta'$, and (2) for every $\alpha'' \in \prodasm{\mathcal{S}}$ where $\alpha'' \rightarrow^\mathcal{S} \beta'$, $\beta' \in \prodasm{\mathcal{S}}$, $R^*(\alpha'') = \alpha$, and $R^*(\beta') = \beta$, there exists $\alpha' \in \Pi$ such that $\alpha' \rightarrow^\mathcal{S} \alpha''$.
\end{definition}

We now present the two new definitions for types of simulations that are relevant when we want to use systems with single-tile seeds to simulate systems with multi-tile seeds. The first, seed-first-simulation, requires that a complete representation of the multi-tile-seed of the simulated system self-assembles before any growth may occur away from the seed, and once the seed is complete, simulation continues in the standard way. The second, seed-growth-simulation, simply requires that the entire seed is eventually grown 
, but doesn't restrict growth away from the seed from beginning before the representation of the multi-tile-seed is complete. However it must still correctly simulate the behavior of the system with a multi-tile seed. 

\begin{definition}
\label{def-s-sf-simulates-t} We say that $\mathcal{S}$ \emph{seed-first-simulates} $\mathcal{T}$ (under $R$) if $\mathcal{S} \Leftrightarrow^{\sigma_\calT}_R \mathcal{T}$ (they have equivalent productions modulo the seed of $\calT$), $\mathcal{T} \dashv^{\sigma_\calT}_R \mathcal{S}$ and $\mathcal{S} \models_R \mathcal{T}$ (they have equivalent dynamics).
\end{definition}

\begin{definition}
\label{def-s-sg-simulates-t} We say that $\mathcal{S}$ \emph{seed-growth-simulates} $\mathcal{T}$ (under $R$) if $\mathcal{S} \Leftrightarrow^{-\sigma_\calT}_R \mathcal{T}$ (they have equivalent productions minus the seed of $\calT$), $\mathcal{T} \dashv^{\sigma_\calT}_R \mathcal{S}$ and $\mathcal{S} \models_R \mathcal{T}$ (they have equivalent dynamics).
\end{definition}

\begin{corollary}\label{cor:seed-growth-implies-shape}
If system $\mathcal{S}$ seed-growth-simulates $\calT$ or seed-first-simulates $\calT$, then $\mathcal{S}$ shape-simulates $\calT$.
\end{corollary}

Corollary \ref{cor:seed-growth-implies-shape} follows immediately from the fact that both seed-growth-simulation and seed-first-simulation require equivalent sets of terminal assemblies between the simulated system $\calT$ and the simulating system $\mathcal{S}$ (under mapping by $R$), which implies shape-simulation.

\begin{corollary}\label{cor:seed-first-implies-seed-growth}
If system $\mathcal{S}$ seed-first-simulates $\calT = (T,\sigma,\tau)$, then $\mathcal{S}$ seed-growth-simulates $\calT$.
\end{corollary}

Corollary \ref{cor:seed-first-implies-seed-growth} follows immediately from the fact that the only difference between the two types of simulation is that seed-first-simulation requires equivalent productions modulo $\sigma$, while seed-growth-simulation requires equivalent productions minus $\sigma$. Seed-first-simulation means that assemblies that map strictly to subassemblies of the seed are first produced, and all of these are valid under equivalent productions minus $\sigma$, and from there all producible assemblies contain the full seed and must map to producible assemblies in $\calT$, and all such assemblies are also valid under equivalent productions minus $\sigma$.

\section{Basic Results and Observations About Seeds}\label{sec:basic}

In this section, we provide a set of relatively basic results and observations that display the amount of information that can be contained in a seed, that it is uncomputable to know whether or not blocking may be caused by the tiles of a seed, the fact that single-tile seeds cannot always be used to replace multi-tile seeds due to the potential for blocking, and that some shapes would require infinite-sized seeds to self-assemble in the aTAM.

\subsection{Encoding information in a seed structure}\label{sec:seed-encoding-GoE}

An aTAM system is defined as a triple. For example, system $\calT = (T,\sigma,\tau)$. $T$ is the tile set and can be encoded using an amount of information proportional to the number of tile types, $t = |T|$, multiplied by the amount of information required to represent each tile type, which given the number of glue types $g$ and the maximum glue strength $\tau$, is $O(\log(g)\log(\tau))$ bits. The seed structure $\sigma$ can be encoded using an amount of information proportional to the number of locations it contains $|\sigma|$, multiplied by the amount of information required to represent which tile type is at each location, which is $O(\log(t))$, for a total of $O(|\sigma|\log(t))$ bits. 
While the number of tile types in $T$ impacts the number of bits required to represent each seed location, a constant number of tile types can clearly be used to define an infinite number of seeds.
Figure \ref{fig:garden-of-eden} gives an example of a tile set which can be used to form an infinite number of seeds that can encode any binary number and are stable at temperature $\tau = 2$. It also gives an example of a seed structure which may encode information in its geometry. 

\begin{figure}
    \centering
    \includegraphics[width=6.0in]{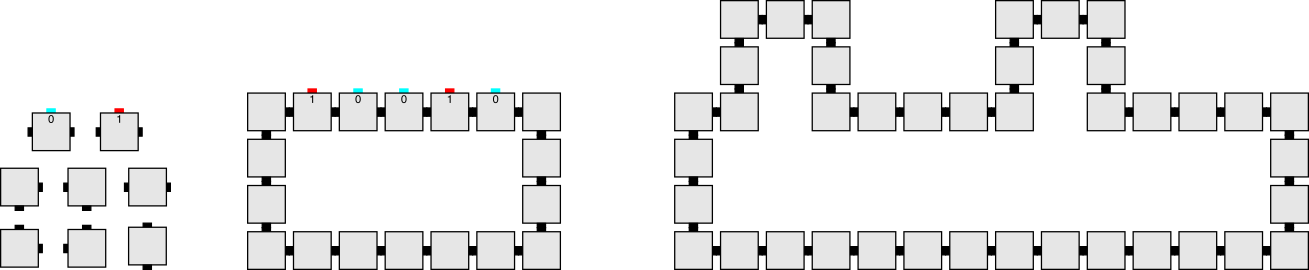}
    \caption{(Left) An example of a constant-sized tile set which can be used to form an infinite number of stable (at temperature 2) assemblies in a ``Garden of Eden'' manner (i.e. they are not stable if even a single tile is missing, and thus there is no valid growth process in which they form). (Middle) An example assembly encoding a binary number via its north-facing glues. (Right) An example assembly potentially encoding information in its geometry (i.e. information that could be ``read'' by paths of tiles that may or may not crash into tiles in specific locations).}
    \label{fig:garden-of-eden}
\end{figure}

\subsection{Uncomputability of blocking}

\begin{observation}\label{obs:uncomp-blocking}
Given an arbitrary aTAM system $\calT = (T,\sigma,\tau)$, it is uncomputable to know if some locations in the seed assembly $\sigma$ can block tiles (in one or more valid assembly sequences).
\end{observation}

\begin{figure}
    \centering
    \includegraphics[width=6.0in]{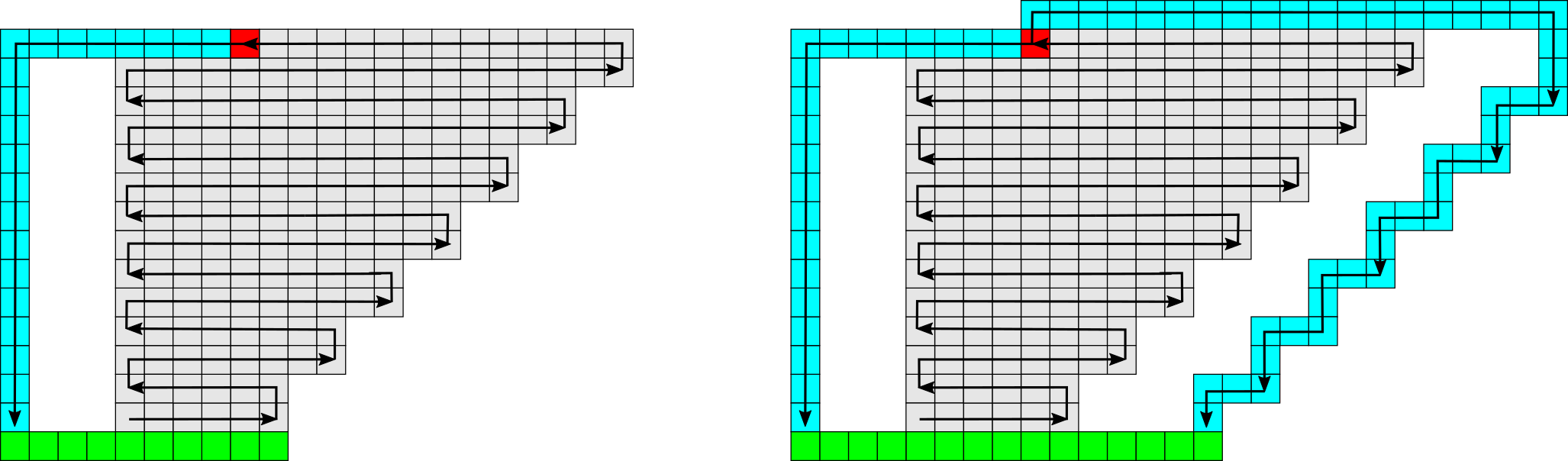}
    \caption{(Left) Schematic example of an assembly simulating a Turing machine and showing the uncomputability of blocking. The seed row (green) encodes the input to the Turing machine (TM) via exposed glues on the north and initiates growth of an assembly where the rows grow in a zig-zag pattern and simulate the TM in a standard way (i.e. each pair of rows represents the configuration of the TM at the time step following that of the pair below it). If and only if the TM halts on the input (the halt state is depicted by the red tile), a row of blue tiles grows to the left and then a column of tiles of the same type grows downward. This column will crash into the seed. Since the crashing column only grows if the TM halts, it is uncomputable to know if growth will crash into the seed.
    (Right) Slight modification to the construction on the left in which, upon halting, the TM initiates a crashing column on the left and another on the right side. This makes it uncomputable to know if the leftmost and/or rightmost seed locations will block and impossible to start from a single seed tile yet guarantee sufficient growth of the seed to provide whatever blocking may be needed.}
    \label{fig:TM-crash}
\end{figure}

\begin{proof}
We prove Observation \ref{obs:uncomp-blocking} by providing a simple reduction to the Halting problem. We'll call the language that consists of binary strings representing aTAM systems that have some seed location that can block the growth of a tile the \texttt{BLOCKING} language. Assume that \texttt{BLOCKING} is computable. Then, there exists some Turing machine, say $M_{block}$, which decides the language. We now construct a new Turing machine, $M_{halt}$ as follows: On input $\langle M,b \rangle$, where $M$ is an encoding of an arbitrary Turing machine and $b$ is an arbitrary binary string, $M_{halt}$ first constructs aTAM system $\calT_M = (T_M, \sigma_b, 2)$ following the design of the system on the left of Figure \ref{fig:TM-crash}. The seed $\sigma_b$ is the green row which encodes $b$ in its northern glues with an additional 4 tiles to the left that have no northern glues. The tile set $T_M$ consists of the seed tile types, plus a set of tiles that simulate a Turing machine in a standard manner (e.g. see \cite{DirectedNotIU} for a definition of zig-zag Turing machines and Figure \ref{fig:TM-crash} for a schematic depiction), plus a small constant-sized set (depicted by the blue tiles) that grows to the left from the tile representing the halting state (if it appears in the assembly, as it does if and only if $M(b)$ halts), past the left edge by four tiles, and then a single tile type that can repeatedly attach to copies of itself in a downward growing column. Note that this column can only form if $M$ halts on $b$, and in this case the downward growing column crashes into the leftmost seed tile, which blocks its growth. With $\calT_M$ thus constructed, $M_{halt}$ gives $\calT_M$ to $M_{block}$ as input. $M_{block}(\calT_M)$ accepts if some location of the seed blocks, which occurs if and only if $M$ halts on $b$, and rejects otherwise. $M_{halt}$ simply outputs the answer given by $M_{block}$. Thus, $M_{halt}$ solves the halting problem, which we know is uncomputable. Therefore, \texttt{BLOCKING} must also be uncomputable.
\end{proof}

\subsection{Impossible simulation at scale factor 1}\label{sec:imp-scale-1}

\begin{figure}
    \centering
    \includegraphics[width=2.5in]{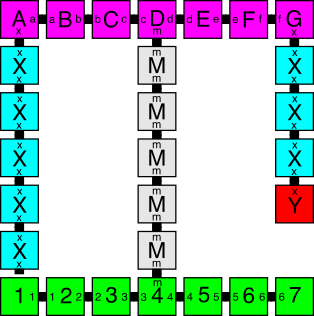}
    \caption{A simple nondeterministic system with a multi-tile seed (green) that can be neither seed-growth-simulated, nor even shape-simulated, by any system with a single-tile seed at scale factor 1. This is because the center column grows upward to a nondeterministic height before possibly growing the pink row that initiates the downward growing blue columns. Each blue column may nondeterministically terminate with a red tile at any height, or continue until crashing into the seed. Any system attempting to simulate it at scale factor 1 must be able to grow a single side of the green row before growing rows that can crash into either side - or continue downward through and beyond the unfilled location of an incomplete side, which would lead to an invalid simulation and/or shape. (Note that the spacing of two empty tile locations between pairs of columns prevents the use of fuzz from enabling a simulation at scale factor 1, and the nondeeterministic ability of blue columns to terminate with red tiles before crashing prevents shape simulation.)}
    \label{fig:simple-scale-1}
\end{figure}

Although it is uncomputable to know if the leftmost seed tile will block growth in the system on the left in Figure \ref{fig:TM-crash}, it is still possible to simulate that system at scale factor 1 with a system utilizing a single-tile seed because the leftmost tile of the multi-tile seed could be used as the single seed tile of the simulating system. The seed could then grow to the right, and if the simulated TM ever halts and a column grows downward, it is guaranteed that the blocking tile of the seed must already be there. (Note that this simulation is only seed-growth simulation if the first row of the zig-zag TM growth starts from the right side and grows right-to-left, since only then can the seed be guaranteed to be fully grown before other growth begins. Otherwise, it would be a seed-growth-simulation.)

However, the system depicted on the right of Figure \ref{fig:TM-crash} is different in this respect. Since there are two locations on opposite ends of the seed which have the potential to block growth, it turns out to be uncomputable to determine whether or not that system can be simulated at scale factor 1 by a system with a single-tile seed. This is because any green location which may be chosen as the location of the single seed tile in the simulator, which we'll call $\mathcal{S}$, has the potential to result in an assembly in which a seed tile that would block growth in the original system, say $\calT$, isn't yet placed when the growth to be blocked in $\mathcal{S}$ arrives. Since the seed of $\calT$ is a single-tile-wide row, the glues must be $\tau$-strength between tiles in the green locations so growth can begin from the seed and directly proceed through the shortest sequence which completes placement of all tiles that serve as input to the TM simulation. At this point, it must be the case that one of the green end tiles (the leftmost or rightmost) has not been placed. From here, a valid assembly sequence is one which places no other tiles in the green locations but instead grows the TM simulations for an arbitrary amount of time. If both TMs halt, it is possible for the assembly sequence to proceed to grow both of the blue paths back down to the seed and for the one on the incomplete side to grow through the seed location.
Thus, it is only possible to always block one side, and it is uncomputable to know if both sides need to be blocked, and therefore uncomputable to know whether or not the simulation will fail.

We now use Figure \ref{fig:simple-scale-1} to present a much simpler system and give intuition as to why even the relaxed notion of shape-simulation (i.e. only needing to generate terminal assemblies with the same shapes as the original system, without needing to follow the same dynamics) using single-tile seeds is impossible for certain systems at scale factor 1. Note that properly accounting for all technicalities, such as growth in fuzz locations, requires more complexities and in-depth analysis such as that which is provided for the proof of Theorem \ref{thm:multi-imposs} (which includes the impossibility of shape-simulation at scale factor 1, among others), and here we will ignore issues that arise with fuzz, etc. to present a high-level sketch that gives the general idea.


We will refer to the system depicted in Figure \ref{fig:simple-scale-1} as $\calT$. In  $\calT$, the green row is the seed assembly. A column grows upward, nondeterministically attaching one of two tile types at each step. The first has the same glue on the north and south and thus allows continued growth, and the second has no glue on the north but has glues on the east and west. Thus, this column can grow arbitrarily high before stopping upward growth. Once a tile of the second type attaches, a row of three tiles grows to each side. Then, a tile of either of two tile types can attach to form a downward growing column on each side.
One of them has the same north and south glue and allows for the column to grow arbitrarily long, and the other (labeled `Y') has no southern glue so stops the growth of the column. The growth of each downward column is nondeterministically either stopped by the attachment of a Y tile, or continues until it is eventually blocked by a seed tile. System $\calT$ makes an infinite number of terminal assembly shapes, which include assemblies with center columns of all heights, and for each of those, assemblies with all combinations of left and right columns of all lengths from 2 to that of the center column.

To show that $\calT$ cannot be shape-simulated at scale factor 1, we'll assume that aTAM system $\mathcal{S}$ is a system with a single-tile seed which attempts to do so.  First, we note that we will make (relatively trivial) use of the Window Movie Lemma from \cite{temp1notIU} (which we include in Section \ref{sec:wml} for completeness), that essentially says that as a constant-width path of tiles becomes arbitrarily long, it must eventually repeat the glues placed along multiple cuts (a.k.a. windows), as well as their ordering (a.k.a. window movies). This means that the segment of the path between those repeated cuts can be ``pumped'' (i.e. there is a valid assembly sequence in which that segment can be repeated an arbitrary number of times). The grey path growing upward and the blue paths which grow down to crash into the seed can be arbitrarily long. Therefore, no matter how many tile types are in $\mathcal{S}$ it is possible for the columns to be tall enough that they must be pumpable.

We now consider all possible locations for the single seed tile of $\mathcal{S}$, which we'll call $\vec{s}$. We'll call the leftmost green location $\vec{l}$, the middle $\vec{m}$, and the rightmost $\vec{r}$. If $\vec{s}$ is located within a green location, it must be possible for the tile that fills location $\vec{m}$ to initiate upward growth of the grey column since that is the only connection between the seed and an infinite number of assemblies whose blue columns stop before crashing into the seed. There must be a valid assembly sequence that grows a row directly from the seed location to place a tile at $\vec{m}$ (noting that if $\vec{m}$ is the seed tile location, no tiles need to be added), at which point an arbitrarily tall grey column must be able to grow upward. It must be possible for this growth, after reaching the distance necessary for pumping, to place the pink tile which initiates the growth to the sides and the downward growing blue columns. Since the blue columns must also be tall enough to pump, it must be possible to pump them until they reach the seed. Since the only portion of the green row to have grown was that directly from the seed location to the grey tile, there must be at least one of the locations $\vec{l}$ or $\vec{r}$ without a tile before the blue column arrives, allowing the blue column to continue pumping past the green row. This creates a shape which is invalid (i.e. it cannot be made in $\calT$), so the simulator fails. Trivial case analysis shows that if the seed location is in any other location, rather than one of the green locations, it is still impossible to ensure that pumpable columns will be blocked, and thus $\mathcal{S}$ must fail.

These scenarios display the difficulty associated with simulating systems with multi-tile seeds by those with single-tile seeds. Given an arbitrary system, it is uncomputable to know which seed locations, if any, can block growth. If multiple seed locations exist which can block growth, it is necessary that there be the ability to grow dependence paths from each to the portion of the assembly that could initiate the crashing growth. Without increasing the scale factor, it is not always possible to do so.


    


\subsection{Shapes requiring infinite-sized seeds}

Trivially, any finite shape (where we define a shape as a connected subset of $\mathbb{Z}^2$) can self-assemble in the aTAM. For instance, given a target shape $S$, for every point in $S$ a unique tile type can be created which, for every adjacent tile when positioned correctly in $S$, has a strength-1 glue that is unique to that pair of adjacent tiles on their abutting sides. Thus, the size of the tile set is the number of points in $S$. The seed can be any single tile from the tile set, and the temperature $\tau=1$. This system will be directed and self-assemble a single terminal assembly of shape $S$.

Therefore, we consider infinite shapes. The aTAM requires that any seed assembly be finite, that is, a seed can contain only a finite number of tiles. Because of this, there are shapes that cannot self-assemble in any system in the aTAM. It has previously been shown that there exists an infinite class of shapes that cannot self-assemble in the aTAM \cite{jSSADST,jSADSSF} (the Sierpinski triangle and many other discrete self-similar fractals), and it would be easy to show that these could self-assemble if allowed infinite seeds (e.g. the seeds could trivially be the entire shapes). However, for completeness in our discussion of how seeds impact self-assembly in the aTAM, here we present a simpler example of a class of shapes which would also require infinite seeds to self-assemble in the aTAM and provide a short proof. An example can be seen in Figure \ref{fig:infinite-seed-shape}.

\begin{figure}
    \centering
    \includegraphics[width=2.0in]{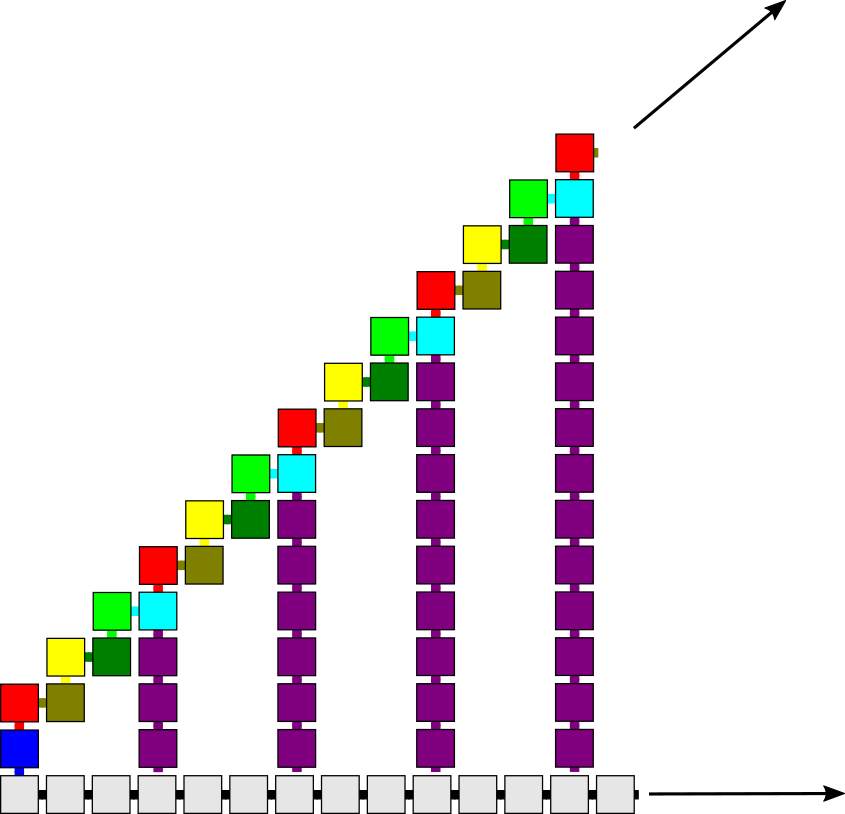}
    \caption{The left side of an infinite shape which could only self-assemble in the aTAM if seeds of infinite size were allowed. The full shape consists of an infinite horizontal line (along the bottom), an infinite diagonal line up and to the right, and vertical lines between those two in every third column.}
    \label{fig:infinite-seed-shape}
\end{figure}

\begin{observation}\label{obs:infinite-seed}
There exist shapes which could only self-assemble in the aTAM if seeds with an infinite number of tiles were allowed.
\end{observation}

\begin{proof}
To prove Observation \ref{obs:infinite-seed}, we refer to the shape depicted in Figure \ref{fig:infinite-seed-shape}, which we'll call $S$. $S$ consists of two infinite paths. One is horizontal and goes infinitely far to the right. The other begins at the left end of the first and goes diagonally upward and to the right infinitely far. At every third $x$-coordinate, a vertical column exists between the two lines. Our proof will be by contradiction, so assume that $\mathcal{S} = (T,\sigma,\tau)$ is a system with a finite seed (i.e. $|\sigma| < \infty$) which self-assembles $S$. Let $x_{max}$ be the $x$-coordinate of the easternmost tile of $\sigma$. Since $S$ is infinite to the right, each of the two main paths extends infinitely far to the right of $x_{max}$, and there are an infinite number of vertical columns to the right of $x_{max}$, each taller than the one to its left.  Clearly, regardless of the size of $\mathcal{S}$'s tile set, each column eventually becomes tall enough to be pumpable. Let that height be denoted by $p$. Let $c_{2p}$ be the first column of height $2p$ which is to the east of any point in $\sigma$. We note that, in order for $\mathcal{S}$ to correctly create an assembly of shape $S$, it must be possible for growth to proceed from $\sigma$ to at least one of the bottom or top tiles of $c_{2p}$, and for that tile to initiate growth of at least half of column $c_{2p}$. (If that is not possible, then either the assembly never gets to column $c_{2p}$, or neither the top nor the bottom can initiate growth that gets to the middle, so column $c_{2p}$ isn't completed and either way $\mathcal{S}$ fails to make shape $S$.) Without loss of generality, assume that $\sigma$ can grow until placing the bottom tile of $c_{2p}$, which is then able to initiate growth upward at least half the height of that column, namely distance $p$. (If instead it is only the top tile that can grow this far, simply flip the directions in the rest of the proof.) 
There must exist some valid, minimal assembly sequence which grows from $\sigma$ to place the bottom tile of $c_{2p}$. Let $c_{2p-1}$ denote the column which is the first to the left of $c_{2p}$. 
Note that it must be possible to grow the path from the bottom of $c_{2p-1}$ to the bottom of $c_{2p}$ without growing the path from the top of $c_{2p-1}$ to the top of $c_{2p}$, since they are disjoint and each one-tile-wide and thus there cannot be a dependency that prevents this.
Now, in $\mathcal{S}$ grow column $c_{2p}$ upward to the mid-point, which it is known to be able to do. However, since the length of this path is the pumping length, it must also be possible to increase the length of this path arbitrarily. Therefore, do so and pump the growth of the column beyond the height $2p$, which must be possible since the top tile of $c_{2p}$ is not in place to block. This places at least one tile outside of shape $S$ and this is a contradiction to the claim that $\mathcal{S}$ self-assembles shape $S$. The only assumption made of $\mathcal{S}$ was that it had a finite seed, thus $S$ cannot be self-assembled within any aTAM system with a finite seed. However, if the seed was allowed to be the infinite bottom row, for example, a system using that seed could easily self-assemble $S$ by growing the diagonal row upward and to the right, and initiating all column growth downward from that. Each column would be guaranteed to crash into the seed row, so $S$ would be correctly self-assembled and Observation \ref{obs:infinite-seed} is proven.
\end{proof}

It is worth noting that this shape actually cannot self-assemble at any scale factor. As the columns get arbitrarily long, any fixed-width provided by a constant scale factor would eventually be unable to prevent pumping, and it must always be the case that either the top or the bottom could grow after the other, and thus the pumping arm could go past that boundary. There exist an infinite set of similar shapes, each differing by just the spacing between the vertical columns and each similarly impossible to self-assemble in the aTAM with a finite seed. Thus, this exhibits an infinite set of shapes, each of which cannot self-assemble in the aTAM at any scale factor.

\section{Limits of Single-Tile Seeds}\label{sec:single-tile-limits}


While we previously provided a relatively simple example to show that some systems with multi-tile seeds cannot be simulated at scale factor 1 by systems with single-tile seeds, in this section we maximize the scale factor for which that is true (since Theorems \ref{thm:scale3-sim} and \ref{thm:scale4-sim} prove that simulation is possible above the bounds shown). Namely, we present a result, Theorem \ref{thm:multi-imposs}, showing that there are systems with multi-tile seeds which require at least scale factor 3 plus the use of cheating fuzz, or scale factor 4 (without cheating fuzz), to shape-simulate using systems with single-tile seeds. Note that since seed-first simulation implies shape-simulation, this also proves that seed-first-simulation requires such scaling for some shapes. This section contains a brief description of the proof, and the full details can be found in Section~\ref{sec:multi-imposs-proof}.

\begin{theorem}\label{thm:multi-imposs}
There exist an infinite number of aTAM systems with multi-tile seeds that cannot be shape-simulated by any aTAM system with a single-tile seed at (1) scale factors 1 or 2, or (2) scale factor 3 without using cheating fuzz.
\end{theorem}


\begin{figure}
    \centering
    \includegraphics[width=3.0in]{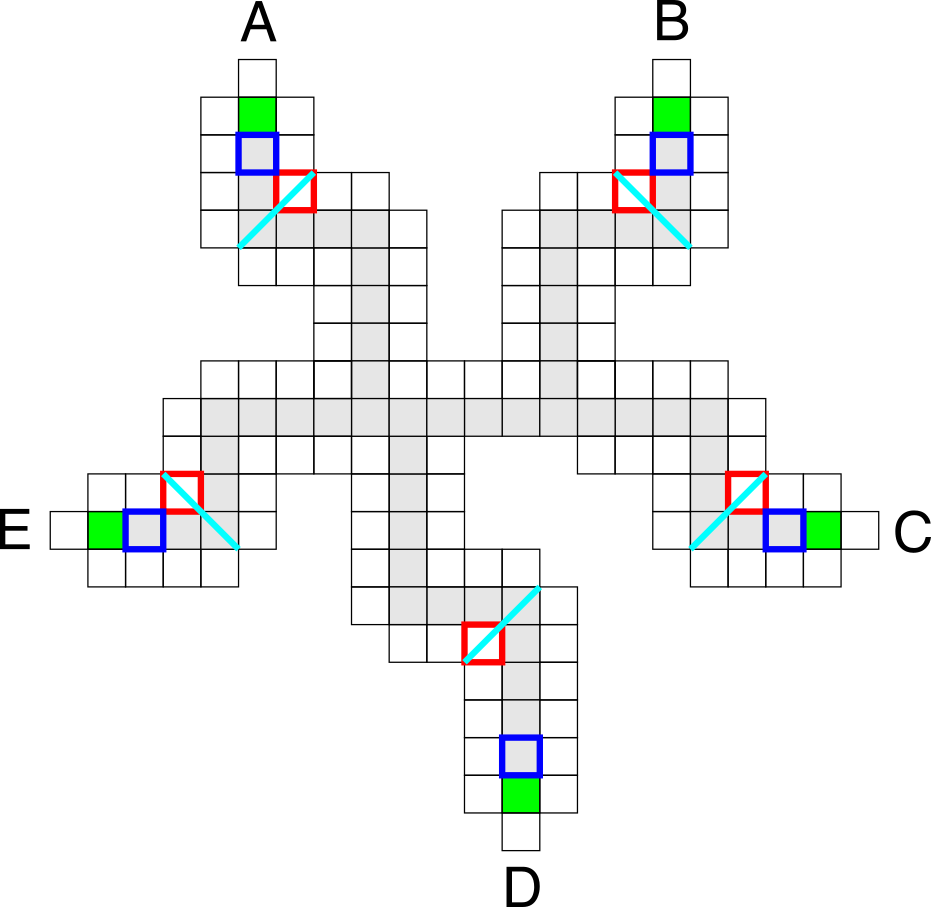}
    \caption{(left) The seed of aTAM system $\calT$ from the proof of Theorem \ref{thm:multi-imposs}. Grey and green locations represent tiles, white positions outlined in black represent potentially valid fuzz locations, tile positions outlined in blue and fuzz positions outlined in red are referenced in the proof. Green locations indicate the only seed tiles which have exterior-facing glues capable of initiating growth. The light blue lines show minimal cuts across each arm and its fuzz (each crossing two locations). The ``arms'' of the seed are labeled A-E for convenience.}
    \label{fig:multi-scale-system}
\end{figure}

\begin{figure}
    \centering
    \includegraphics[width=4.0in]{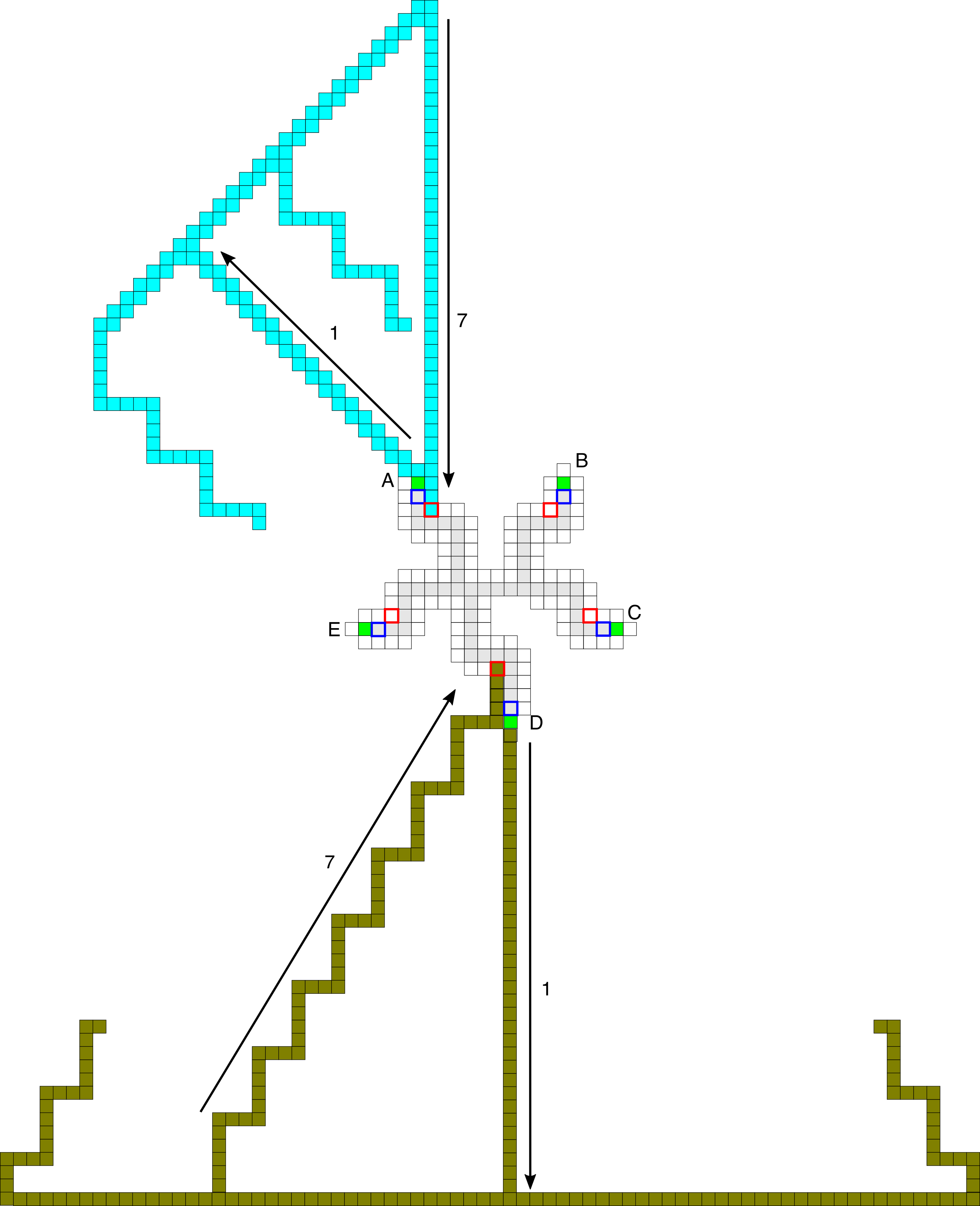}
    \caption{Depiction of an assembly producible in $\calT$ in the proof of Theorem \ref{thm:multi-imposs}. (The seed alone is shown in Figure \ref{fig:multi-scale-system}.) From each arm of the seed, paths can grow away from the seed, and two such (each labeled 1) are depicted here. The paths labeled 1 can grow arbitrarily far away from the seed and initiate the growth of paths back toward the seed. For example, the blue path 7 from arm A, and gold path 7 from arm D, have pumped and crashed into the red locations of those arms.
    }
    \label{fig:empty-crash}
\end{figure}

To prove Theorem \ref{thm:multi-imposs}, we present one system which can't be shape-simulated by any aTAM system with a single-tile seed at (1) scale factors 1 or 2, or (2) scale factor 3 without using cheating fuzz, and note how an infinite set of such systems can easily be derived from it.
Let $\calT = (T,\sigma,1)$, the seed of which is shown in Figure \ref{fig:multi-scale-system} and a producible assembly of which is depicted in Figure \ref{fig:empty-crash}, be such a system.
In $\calT$, there are five locations on the perimeter of the seed in which glues are exposed and to which tiles can attach. They are depicted in green.
Each arm can grow a uniquely colored subassembly, each of which uses tile and glue types unique to that subassembly.
An infinite number of unique terminal assemblies, and assembly shapes, are possible in $\calT$, since, for instance, each path labeled 1 can be of arbitrary length. The intuition behind the proof is that it is possible for arbitrarily long paths that start at the ends of the seed's arms to grow away from the seed assembly, and then to grow paths that come back toward the seed, where they can crash into other arms. In order for a system with a single-tile seed to simulate this system, it would have to ensure that there are tiles in place in each arm to provide blocking before any arm is able to grow crashing paths. However, at the small scale factors allowed there is not enough room for all necessary dependence paths to grow between all arms, so simulation fails.
\section{Single-tile Seeds via Seed-First-Simulation at Small Scales}\label{sec:min-scaled-sims}

Given that we have shown the impossibility of shape-simulation (the broadest form of simulation via singly seeding) at scale factors 1 and 2, a question that naturally follows is: do systems exist which are able to shape-simulate just above the proven lower scale bounds?
The answer is yes - when allowing cheating fuzz, this is possible at scale factor 3 (Section~\ref{sec:scale3-sim-fuzz}), and restricting ourselves to no cheating fuzz allows for simulation at scale factor 4 (Section~\ref{sec:scale4-sim}).
Both constructions utilize seed-first-simulation.

\subsection{Seed-First-Simulation at Scale Factor 4}\label{sec:scale4-sim}
We begin by introducing the construction for seed-first-simulation at scale 4 without cheating fuzz.
Although it is not the smallest possible scale factor for simulation, scale factor 4 has favorable characteristics that allow us to build an overall process for seed-first-simulation in a relatively straightforward manner that can be re-utilized in our scale factor 3 construction. 

\begin{theorem}\label{thm:scale4-sim}
Given an arbitrary aTAM system $\mathcal{T} = (T,\sigma,\tau)$, there exists an aTAM system $\mathcal{T}_4 = (T_4, \sigma_0, \tau_4= \max(2, \tau))$ which seed-first-simulates $\mathcal{T}$ at scale factor 4 and does not use cheating fuzz, where $|\sigma_0| = 1$ and $|T_4| \leq 28s + 16g + 6t$ given that $s = |\sigma|$, $t = |T|$, and $g$ is the total number of unique glue/strength combinations in $T$ and $\sigma$.
\end{theorem}

The full proof of Theorem~\ref{thm:scale4-sim} is located in Section~\ref{sec:scale4-sim-proof}, and here we provide a brief overview.
The construction of a simulating system $\mathcal{T}_4$ consists of two main parts.
First, we create a tileset $T_\sigma$ which assembles a scale-4 version of the seed of $\calT$, $\sigma$.
The generation of $T_\sigma$ comprises the brunt of the construction.
The new assembly that grows to represent the seed $\sigma$ can be logically thought of as having 2 parts: the \emph{core path} and the \emph{perimeter path}. Each scale 4 supertile of $T_\sigma$ is shown to contain an interior scale 2 square which can be easily connected to its neighboring scale 2 squares; these 4 tiles are called \emph{core tiles}.
The core path is a dependence path which represents a Hamiltonian path through the core tiles, allowing for all supertiles in the seed to be connected to the remainder of the seed assembly, enabling seed-first simulation.
By focusing on the connectivity of the core tiles within the supertiles of the seed assembly, the core path is generated utilizing the fact that all shapes with scale factor $2$ contain a Hamiltonian cycle (see proof in \cite{SummersTemp}).
We generate both the core path and the perimeter path from a set of scale factor 4 template tiles which take advantage of the presences of a Hamiltonian cycle in the core tiles.
One of the supertiles represented as a template is the \emph{origin supertile} which contains the single-tile seed $\sigma_0$, and the beginning of the perimeter path.
We connect the end of the core path to the beginning of the perimeter path; this is the key part of the construction which allows for seed-first simulation.
On the edges of tiles of the perimeter path, exterior glues allow for the attachment of tiles.
A high level visualization of the steps of generating the new tileset $\calT_4$ is found in Figure~\ref{fig:S4-4_spanning_tree}.

\begin{figure}
    \centering
    \includegraphics[width=0.6\textwidth]{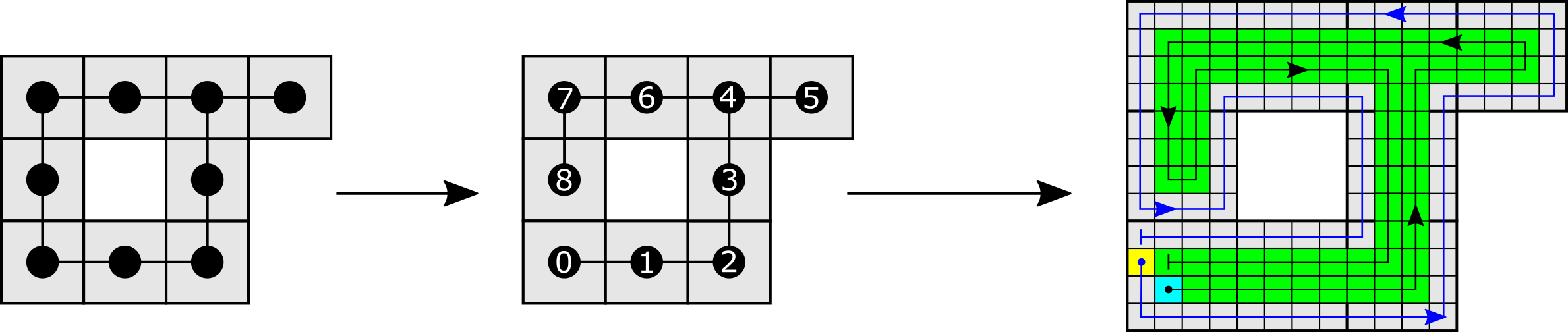}
    \caption{The process of creating the tiles for a simulation at scale 4 with a single-tile seed. Beginning with the binding graph of a seed assembly $\sigma$ (left), we generate a spanning tree from the westernmost tile of the southernmost row (middle). From this spanning tree, we replace each vertex with a tile from the template provided by Figure~\ref{fig:S4-2_seed_path_macrotiles}, leading to the tileset which comprises the scale factor 4 simulation of our initial multi-tile seed (right). The new single-tile seed of $\calT_4$ is shown as a teal tile.}
    \label{fig:S4-4_spanning_tree}
\end{figure}

\begin{figure}
    \centering
    \includegraphics[width=0.25\textwidth]{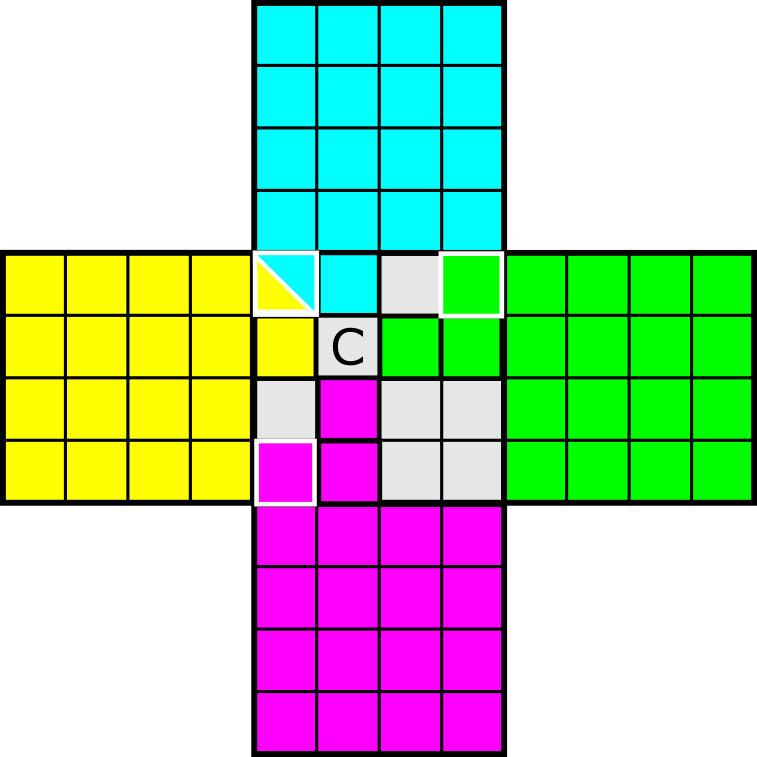}
    \caption{High-level scheme for scale-4 supertiles, showing the point of competition/cooperation. Consider the fully colored $4 \times 4$ supertiles that are filled in to already be representing tiles of the simulated system, and the central $4 \times 4$ supertile to not yet be filled in enough to represent a tile. Locations in the central supertile with the same colors as neighboring supertile locations depict the paths by which tiles from the filled-in supertiles grow into an empty neighboring supertile. The location labeled 'C' is the point of competition/cooperation where tile placement decides the identify of the supertile.}
    \label{fig:S4-8_scale_4_template}
\end{figure}

Second, we create scale 4 versions of the tileset $T$.
Before creating the scaled versions of $T$, we additionally need to prevent the re-growth of portions of the seed which attempt to re-grow the perimeter path. 
Re-growth is prevented by generating $T_{IO}$, an expansion of $T$, which identifies all possible $\tau$ strength combinations which can allow a tile to attach (serving as its ``input'' glues) and then generating new tiles with specific `inward' and `outward' glues. (This is a standard technique in tile assembly constructions, and can be found in the construction of \cite{Versus} and others).
Finally, we generate the scale 4 representations of $T_{IO}$ based upon a single point of cooperation and/or competition between the supertiles.
(See Figure \ref{fig:S4-8_scale_4_template} for a high-level depiction.)
In addition to assigning tiles which grow into legal fuzz regions, we must take into account collisions of an arbitrary supertile with the seed assembly and add glues which allow for the continuation of the dependence path.

\subsection{Seed-First-Simulation at Scale Factor 3}\label{sec:scale3-sim-fuzz}

We demonstrate that an arbitrary aTAM system is able to be seed-first-simulated at scale factor 3 utilizing cheating fuzz.

\begin{theorem}\label{thm:scale3-sim}
Given an arbitrary aTAM system $\mathcal{T} = (T,\sigma,\tau)$, there exists an aTAM system $\mathcal{T}_3 = (T_3, \sigma_0, \tau_3= \max(2, \tau))$ which seed-first-simulates $\mathcal{T}$ at scale factor 3 utilizing cheating fuzz, where $|\sigma_0| = 1$ and $|T_3| \leq 20s + 16g + 6t$ given that $s = |\sigma|$, $t = |T|$, and $g$ is the total number of unique glue/strength combinations in $T$ and $\sigma$.
\end{theorem}

\begin{figure}
    \centering
    \includegraphics[width=0.25\textwidth]{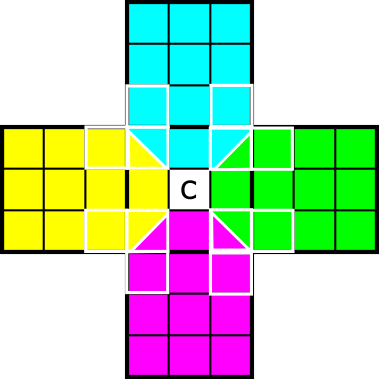}
    \caption{For a scale 3 supertile, the point of cooperativity/competition resides in the center of the $3 \times 3$ square (denoted with `C').}
    \label{fig:S3-7_cooperativity_scale3}
\end{figure}

The full proof of Theorem~\ref{thm:scale3-sim} is found in Section~\ref{sec:scale3-sim-proof}.
This proof is by construction which, while similar to that presented in the construction of Theorem~\ref{thm:scale4-sim}, requires extra care in the routing of the both the dependence path and perimeter path through the seed-representing assembly to prevent diagonal fuzz by the tiles in $T_\sigma$.
In a scale 4 supertile, the perimeter path contains tiles in every supertile, regardless of its location in $\sigma$.
This is not possible in scale 3 seed assemblies, specifically with regards to supertiles which are not on the perimeter of a structure.
This potentially inhibits the placement of perimeter tiles in supertiles adjacent to \emph{cavities} (locations in $\mathbb{Z}^2$ which do not contain a tile but are surrounded by seed tiles) which may be present in a multi-tile seed.
This is necessary for seed-first-simulation.
To ensure the perimeter path may reach these internal cavities, we provide a modification to the connectivity of the seed prior to assigning supertiles from a set of supertile templates  which leverages the changes in connectivity.
For the generation of the remaining tiles of $\calT$, we utilize the the same techniques of Theorem~\ref{thm:scale4-sim} to develop a $T_{IO}$, and the design for the point of competition/cooperation in the $3 \times 3$ supertiles is slightly modified, as shown in Figure \ref{fig:S3-7_cooperativity_scale3}.

\section{Optimal Tile Complexity for Seed-First-Simulation}\label{sec:scaled-sim}

In this section, we give an overview of a universal construction that takes as input a program $P_\calT$, where $P_\calT$ outputs any arbitrary aTAM system, say $\calT$, and the construction outputs an aTAM system $\mathcal{S}$ with a single-tile seed that seed-first-simulates $\calT$. Technical details can be found in Section \ref{sec:scaled-sim-append}. The tile complexity of $\mathcal{S}$ is asymptotically related to the length of the input program. Thus, the tile complexity is $O(\frac{K(\calT)}{\log(K(\calT)})$, where $K(\calT)$ is the Kolmogorov complexity of $\calT$, i.e. the length of the shortest program that outputs $\calT$.

\begin{figure}
    \centering
    \includegraphics[width=4.0in]{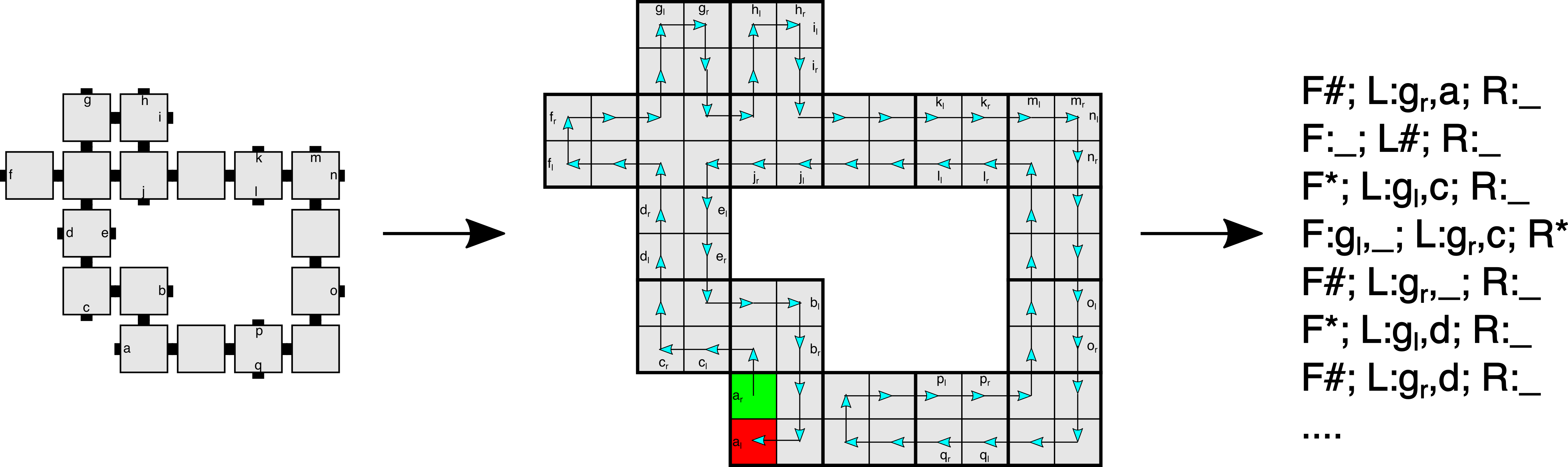}
    \caption{(Left) The seed $\sigma$ of $\calT$, (Middle) $\sigma^2$ with a Hamiltonian cycle through it (beginning in the green location and ending in the red location), (Right) A shorthand representation of the sections representing the first 7 locations along the Hamiltonian cycle in the string encoding it.}
    \label{fig:ham-cycle-encoding}
\end{figure}

\begin{theorem}\label{thm:scaled-simulation}
Given an arbitrary aTAM system $\calT$, there exists $c \in \mathbb{Z}^+$ and aTAM system $\mathcal{S}_\calT = (T_\calT,\sigma_\calT,2)$ such that $|\sigma_\calT| = 1$, $|T_\calT| = O(\frac{K(\calT)}{\log(K(\calT)})$, and $\mathcal{S}_\calT$ seed-first-simulates $\calT$ at scale factor $c$.
\end{theorem}

At a high level, the construction works by taking as input the program $P_\calT$ and creating a hard-coded set of tiles (the first of which will be the single seed tile) that self-assemble a row whose north glues encode a compact version of that program that is first ``unpacked'' by a set of tiles that present the full program $P_\calT$ along the northern glues of a row of tiles. A tile set that simulates a Turing machine then uses that program as input, runs that program to obtain the description of aTAM system $\calT$, then uses that description to run a variety of subprograms. One subprogram executes the algorithm from the intrinsic universality construction of \cite{IUSA} which computes a string encoding information about $\calT$ in a format that can be used by that construction to simulate the tiles of $\calT$ as large supertiles with that information on their exteriors.
Another creates a $2 \times 2$ scaling of the seed assembly of $\calT$ and computes a Hamiltonian cycle through it. (An example can be seen in Figure \ref{fig:ham-cycle-encoding}.) It then creates a string containing entries for every stop along the Hamiltonian cycle that contain the information to be put on each output side. During this computation, a binary counter tile set keeps track of the number of computational steps so that the scale factor is computed and utilized.
Then, the information is moved along the cycle and the information copied to the exterior sides of the assembly representing the scaled version of the seed of $\calT$.
Once that structure is complete, the IU construction of \cite{IUSA} (only minimally changed) takes over and manages the simulation of the rest of steps in the simulation.

\subsection{Simultaneous simulation}\label{sec:simult-sim}

In this section, we describe a relatively simple modification to the construction of Theorem \ref{thm:scaled-simulation} that results in a single aTAM system that simultaneously and in parallel simulates every aTAM system.
At a high level, this construction works by the (standard aTAM) assumption that there are an infinite number of copies of the seed, which in this case is a single tile. From each copy, an arbitrarily long row encoding an arbitrary binary number nondeterministically grows. Each row terminates at a random point when a special tile attaches to ``cap'' the growth of the row. At this point, the row encodes a random binary number $n$ in its northern glues (and every binary number has some chance of being represented). Each such assembly encoding some number $n$ then proceeds to grow into a simulation of the $n$th tile assembly system in the aTAM using the construction of Theorem \ref{thm:scaled-simulation}. Thus, in parallel, the infinite set of seeds grows into an infinite set of assemblies, each simulating one aTAM system from the infinite set, and each aTAM system is simulated (with some non-zero probability)\footnote{An alternative, but also common, interpretation of the aTAM model is that, similar to the behavior of nondeterministic versions of automata such as Turing machines, whenever a nondeterministic choice occurs the system splits into a separate instance to follow each option. In this interpretation, it is considered that there is only a single copy of the seed, and then for each assembly to which more than one tile can attach, a new instance of the assembly is created for each possible attachment. Thus (possibly in the limit) all assemblies which can form from a seed assembly do so. In this interpretation, this construction also validly simulates all systems in parallel.}

In order to discuss such a simulator, we must define a new type of representation function and slightly modify the definition of simulation which allows multiple scale factors to be utilized in parallel.
This is because there are a countably infinite number aTAM systems and therefore for any scale factor $m \in \mathbb{Z}^+$, there must exist an infinite number of aTAM systems, with increasingly large tile sets to represent, that require scale factor greater than $m$ to simulate. 

We will call this new notion of simulation \emph{mixed-scale-simulation}. For use with our result, we will base it off of the definition of seed-first-simulation with the following modification.

    
Rather than the standard definition of a representation function, due to the fact that different scale factors must be allowed for different simulations occurring simultaneously in the same system, we define an \emph{adaptive representation function} as one which is not confined to a grid of fixed-size squares, but which instead is allowed to inspect any locations of an assembly and only be restricted by the requirement that it is able to inspect only a finite portion of an assembly before reading enough information to (1) compute the scale factor of the simulation being carried out, and (2) compute the mapping from supertiles in that assembly to tiles from the simulated system, and then (3) to correctly identify supertiles at that scale.

\begin{theorem}\label{thm:simult-sim}
There exists an aTAM system $U_\infty$ that mixed-scale-simulates all aTAM systems, simultaneously and in parallel.
\end{theorem}

\begin{proof}

We prove Theorem \ref{thm:simult-sim} by construction. Therefore, we present aTAM system $\mathcal{U}_\infty = (U,\sigma,2)$ and discuss how it is constructed and how it behaves.

Let $D$ be the Turing machine defined in the proof of Theorem \ref{thm:scaled-simulation}. We now define a new Turing machine $D'$ which takes as input a binary string $b \in 1(0 \cup 1)^*$ (i.e. any binary string beginning with a 1) that is immediately followed by the letter `x', and does the following:

\begin{enumerate}
    \item $D'$ enumerates over the set of all aTAM systems following some enumeration\footnote{The set of all aTAM systems is countably infinite since it is clearly infinite (e.g. for every $i \in \mathbb{N}$ there exists an aTAM system with $i$ tile types that has a single-tile seed and self-assembles an $i \times 1$ line at temperature 1), but every component of an aTAM system must be finite by definition of the aTAM, and therefore the set of systems must be countably infinite. Since any countably infinite set can be enumerated, there exists some enumeration of all aTAM systems.}, and stops after printing the $b$th aTAM system (where the binary string $b$ is interpreted as a positive integer). Let $\langle \calT \rangle$ be the string representing the $b$th aTAM system.
    
    \item $D'$ creates a Turing machine, which we'll call $P_\calT$, that takes no input and simply prints $\langle \calT \rangle$ and halts (by having the characters of $\langle \calT \rangle$ hard-coded into its transition rules).
    
    \item $D'$ runs $D(P_\calT)$ and outputs what $D(P_\calT)$ outputs.
\end{enumerate}

\begin{figure}
    \centering
    \includegraphics[width=3.0in]{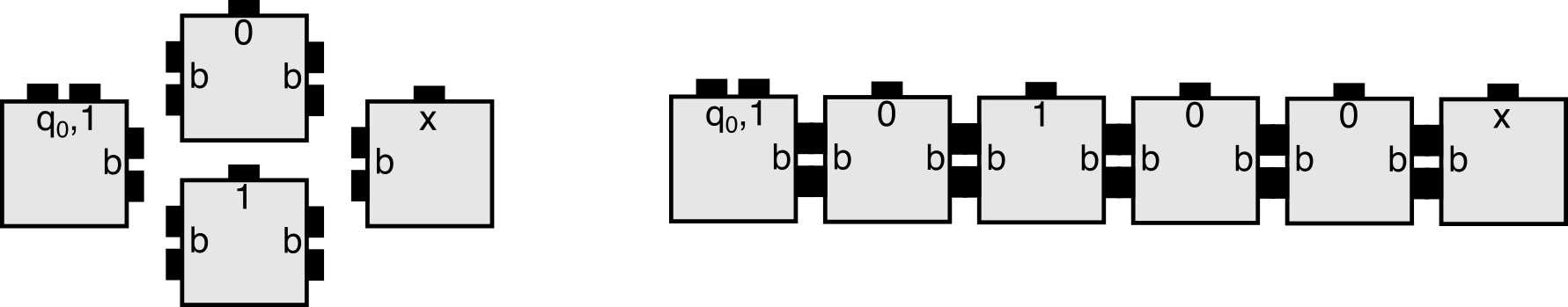}
    \caption{(Left) Assuming that $q_0$ is the start state of Turing machine $D'$, using the tile on the left (which we'll refer to as $b_{seed}$) as the seed tile, this set of tiles can self-assemble to represent any binary string of the form $1(0 \cup 1)^*$ with a terminating $x$ and that assembly can be used as input to $D'$. (Right) An example assembly encoding the binary string ``10100''.}
    \label{fig:binary-string-tiles}
\end{figure}

We now define tile set $U_\infty$ as the tile set $T_\calT$ from the proof of Theorem \ref{thm:scaled-simulation} minus the tile sets $T_\sigma$ and $T_D$. Instead of the tiles of $T_\sigma$ which encode some specific program, we add the set of tiles $T_b$ pictured on the left in Figure \ref{fig:binary-string-tiles}. Then, instead of the tiles $T_D$ we add tiles that simulate $D'$ in the same way that the tiles of $T_D$ simulated $D$ (i.e. zig-zag growth, an embedded counter, etc.).

We can now fully define $\mathcal{U}_\infty = (U_\infty, \sigma, 2)$, where $\sigma$ is one copy of a tile of type $b_{seed}$ located at $(0,0)$.
$\mathcal{U}_\infty$ does the following. For every $i \in \mathbb{Z}^+$, there exists a producible assembly that self-assembles from the tiles of $T_b$ that represents $i$ in binary, followed by the character `x'. Each such assembly initiates the growth of an assembly that simulates $D'$ on that value of $i$. This causes $D'$ to print out the $i$th aTAM system and then input that to $D$, which will then cause the tiles from the construction of Theorem \ref{thm:scaled-simulation} to use that assembly to grow in such a way that it correctly seed-first-simulates the $i$th aTAM system.

The adaptive representation function $R^*$ works as follows. Given an arbitrary assembly, it begins at location $(0,0)$, which is the location of the seed tile, $b_{seed}$. 
It reads the number encoded by the tiles to the right of that tile until encountering a tile encoding an `x' and thus the end of the binary number. If no such tile is found, the assembly maps to empty space as it doesn't (yet) simulate any system. For those that do encode a complete number, $R^*$ runs Turing machine $D'$ with that number as input to get the full specification of the system being simulated by that assembly (which we'll call $\calT$), and from that it is also able to compute the scale factor of the simulation, $m$. $R^*$ is then able to inspect the $m \times m$ squares and match the information encoded in sufficiently completed supertiles to tiles in $\calT$.

Since, for every aTAM system, there exists an assembly which simulates it, and there exists an adaptive representation function capable of identifying the system being simulated by each assembly and correctly mapping the supertiles to tiles of that simulated system, $\mathcal{U}_\infty$ mixed-scale-simulates all aTAM systems, simultaneously and in parallel.

\end{proof}

\vspace{-20pt}
\subsection{Cross-model intrinsic universality}

In the following we claim that, given some model $M$, if for any arbitrary aTAM system $\calT$, there exists a system in $M$ that simulates it, then there is a tile set in model $M$ that is intrinsically universal for the aTAM (using seed-growth-simulation). This means that a single tile set in $M$ is capable of being used to simulate any aTAM system. First, we note that our use of the term ``model'' in this case applies not only to the complete set of systems within a model, but also to what is commonly referred to as a ``class of systems within a model''. For instance, the aTAM is a model and so is the 2-Handed Assembly Model, so referring directly to one of those models would mean the full set of systems within those models. Alternatively, the subset of systems within the aTAM which are directed is often referred to as the class of directed aTAM systems. The following result applies to both models and classes of systems within models.

While it may seem like the existence of a one-to-one relationship between systems in some model $A$ and systems in another model $B$ that can simulate them may imply the existence of a tile set in $B$ which is intrinsically universal for model $A$, this is not the case. A counterexample occurs with the class of systems in the aTAM which are directed. Trivially, using the identity function as the representation function and scale factor 1, for every system in the directed aTAM there exists a system within the directed aTAM which simulates it. However, as proven in \cite{DirectedNotIU}, there exists no tile set which can be used within the directed class of aTAM systems to simulate all directed aTAM systems.

\begin{claim}\label{clm:cross-model-IU}
Let $M$ be a model of tile-based self-assembly such that, for every system $\calT$ in the aTAM, there exists a system $\calT' \in M$ that simulates system $\calT$ of the aTAM. Then, there exists a tile set in $M$ which is intrinsically universal for the aTAM by performing seed-first-simulations.
\end{claim}

To support Claim \ref{clm:cross-model-IU}, we now define aTAM system $\mathcal{U}_b$ as follows. Use the tile set $U$ from $\mathcal{U}_\infty$ of the proof of Theorem \ref{thm:simult-sim}, set the temperature to 2, select an arbitrary binary number $b$ and create a seed assembly, $\sigma_b$, that is pre-built from the tiles of $T_b$ (starting with $b_{seed}$) to represent $b$ with an `x' after the last bit (i.e., rather than letting such assemblies nondeterministically form, the seed is now a single, pre-selected number encoded into a seed assembly).

Let $M$ be a model of tile-based self-assembly such that, for every aTAM system $\calT$, there exists a system $\calT'$ in $M$ that simulates $\calT$. If this is true, then there must be a system in $M$, which we'll call $\mathcal{M}_b$, that simulates $\mathcal{U}_b$. Let $T_M$ be the tile set used by $\mathcal{M}_b$, $c$ be the scale factor of the simulation, and $R$ the representation function. 
We now claim that, since a seed assembly using $T_M$, $c$, and $R$ was made for an arbitrary value of $b$, it should be possible to use the same $T_M$, $c$ and $R$ to make seed assemblies encoding all values $i \in \mathbb{Z}^+$. The intuition behind this claim is that it must be possible to reuse the macrotiles representing tiles of individual bit values, since that would have to be the case if $b$ was large enough, so it must be possible to combine them to form any desired value. Assuming that claim, 
it is possible to generate system $M_i$, using $T_M$, $R$, and $c$ to create an assembly over $T_M$ that maps to $\sigma_i$ so that $M_i$ simulates $\mathcal{U}_i = (U, \sigma_i, 2)$ for arbitrary $i$, and note that not only does $M_i$ simulate $U_i$ in the standard way (i.e. with a pre-built seed structure for a multi-tile seed), using $R$ and $c$, but we can also determine the representation function $R''$ and scale factor $c''$ such that we can interpret the assemblies in $\mathcal{M}_i$ as assemblies in $\calT_i$, which $\mathcal{M}_i$ is also seed-first-simulating via $R''$ at scale factor $c''$.

Let Turing machine $M^*$ be a Turing machine that takes as input an assembly $\sigma_i$ and runs an aTAM simulator on the system $(U,\sigma_i,2)$ until the it completes the simulation of $D'(i)$ (where $D'$ is the TM defined in the proof of Theorem \ref{thm:simult-sim}).  
At that point, $M^*$ will be able to compute the scale factor $c'$ at which $(U,\sigma_i,2)$ simulates aTAM system $\calT_i$, as well as the representation function $R'$ for that simulation.

Define $R''$ as follows: On input $\alpha'$ which is a producible assembly in $\mathcal{M}_i$ (i.e. an assembly over the tiles of $T_M$),
\begin{enumerate}
    \item Run $M^*$ on input assembly $\sigma_i$ to obtain $c'$ and $R'$. (Note that the assembly $\sigma_i$ is hard-coded into $R''$ so that it is specific for the simulation of aTAM system $i$.)
    \item Compute $c'' = c * c'$
    \item Return $R'(R(\alpha')$ (i.e. the result is an assembly in $\calT_i$)
\end{enumerate}

Since we claim it's possible, for an arbitrary $i \in \mathbb{Z}^+$, to create an assembly from the tiles of $T_M$ that represents $\sigma_i$ under $R$ at scale factor $c$, and that such an assembly seed-first-simulates the $i$th aTAM system, it would then be the case that $T_M$ is intrinsically universal for the aTAM by performing seed-first-simulations.


\section{Technical Appendix}\label{sec:append}


In this section we provide technical definitions and  details of proofs of the results.

\subsection{Window Movie Lemma}\label{sec:wml}

Here we include definitions related to the Window Movie Lemma from \cite{temp1notIU}, as well as restate that lemma, for completeness.

A \emph{window} $w$ is a set of edges forming a cut-set in the infinite grid graph over $\mathbb{Z}^2$.
Given a window $w$ and an assembly $\alpha$, a window that {\em intersects} $\alpha$ is a partioning of $\alpha$ into two  configurations (i.e.\ after being split into two parts, each part may or may not be  disconnected).
In this case we say that the window $w$ cuts the assembly $\alpha$ into two configurations $\alpha_L$ and~$\alpha_R$, where $\alpha = \alpha_L \cup \alpha_R$.
Given a window $w$, its translation by a vector $\vec{c}$,  written $ w + \vec{c}$ is simply the translation of each of $w$'s elements (edges) by~$\vec{c}$.

Given an assembly sequence $\vec{\alpha}$ and a window $w$, the associated {\em window movie} is the maximal sequence $M_{\vec{\alpha},w} = (v_{0}, g_{0}) , (v_{1}, g_{1}), (v_{2}, g_{2}), \ldots$ of pairs of grid graph vertices $v_i$ and glues $g_i$, given by the order of the appearance of the glues along window $w$ in the assembly sequence $\vec{\alpha}$.
Furthermore, if $k$ glues appear along $w$ at the same instant (this happens upon placement of a tile which has multiple  sides  touching $w$) then these $k$ glues appear contiguously and are listed in lexicographical order of the unit vectors describing their orientation in $M_{\vec{\alpha},w}$.

\begin{lemma}[Window movie lemma]
\label{lem:windowmovie}
Let $\vec{\alpha} = (\alpha_i \mid 0 \leq i < l)$ and $\vec{\beta} = (\beta_i \mid 0 \leq i < m)$, with
$l,m\in\Z^+ \cup \{\infty\}$,
be assembly sequences in $\mathcal{T}$ with results $\alpha$ and $\beta$, respectively.
Let $w$ be a window that partitions~$\alpha$ into two configurations~$\alpha_L$ and $\alpha_R$, and $w' = w + \vec{c}$ be a translation of $w$ that partitions~$\beta$ into two configurations $\beta_L$ and $\beta_R$.
Furthermore, define $M_{\vec{\alpha},w}$, $M_{\vec{\beta},w'}$ to be the respective window movies for $\vec{\alpha},w$ and $\vec{\beta},w'$, and define $\alpha_L$, $\beta_L$ to be the subconfigurations of $\alpha$ and $\beta$ containing the seed tiles of $\alpha$ and $\beta$, respectively.
Then if $M_{\vec{\alpha},w} = M_{\vec{\beta},w'}$, it is the case that  the following two assemblies are also producible:
(1) the assembly $\alpha_L \beta'_R = \alpha_L \cup \beta'_R$ and
(2) the assembly $\beta'_L \alpha_R = \beta'_L \cup \alpha_R$, where $\beta'_L=\beta_L-\vec{c}$ and $\beta'_R=\beta_R-\vec{c}$.
\end{lemma}

Essentially, the Window Movie Lemma states that if the same window movie occurs in two different assembly sequences of some TAS $\calT$, but in different locations, valid producible assemblies in $\calT$ include (1) an assembly with the ``left'' half created by the first sequence and the ``right'' half created by the second sequence, and (2) an assembly with the ``left'' half created by the second sequence and the ``right'' half created by the first sequence. Typical use of this lemma includes showing that some portion of a sufficiently large growing assembly must ``pump'', i.e. have repetitive structure exhibited by repetition of identical window movies. When this occurs, there must exist valid assembly sequences in $\calT$ in which the portions of the assembly that grow after each occurrence of that window movie can be swapped, and/or the subassembly between the identical window movies can be repeated (a.k.a. \emph{pumped}) an arbitrary number of times.

\subsection{Proof of Lemma \ref{lem:depend-path}}\label{sec:depend-path-proof}

Restatement of Lemma \ref{lem:depend-path} (Dependence paths):
Given a singly-seeded aTAM system $\calT$, producible assembly $\alpha \in \prodasm{\calT}$, and sets of locations $L_1,L_2 \subset \dom{\alpha}$, if $L_2$ strictly depends upon $L_1$, then in each valid assembly sequence of $\calT$ there must be a dependence path from some  $l_1 \in L_1$ to some $l_2 \in L_2$.


\begin{proof}
We prove Lemma \ref{lem:depend-path} by contradiction. Therefore, assume that, given singly-seeded aTAM system $\calT = (T,\sigma, \tau)$, producible assembly $\alpha \in \prodasm{\calT}$, and sets of locations $L_1,L_2 \subset \dom{\alpha}$, that $L_2$ strictly depends upon $L_1$. However, for the sake of contradiction assume that there exists some valid assembly sequence in $\calT$ in which there is not a dependence path from any point $l_1 \in L_1$ to any of the points $l_2 \in L_2$.

By the definition of strict dependence, we know that a tile is placed in $L_1$ before a tile is placed in $L_2$. Therefore, we have two cases to consider: (1) $\dom{\sigma} \in L_1$, i.e. $L_1$ contains the location of the seed of $\calT$, or (2) $L_1$ does not contain the seed location. We'll show that case (1) can't hold by induction. The induction hypothesis is that, given a producible assembly $\alpha$ in which every tile has a dependency path from the seed (which is in $L_1$) to it, then any tile which binds has a dependency path to the seed. The base case is the first tile attachment, which must be directly to the seed and therefore that tile has a dependency path to the seed. We prove the induction hypothesis simply by noting that if every tile of $\alpha$ has a dependency path to the seed and a new tile attaches to $\alpha$, the dependency path of any tile to which it binds is simply extended by 1 to be a dependency path from the location of the seed to the new tile. Therefore, there is a dependency path from the seed to any location in a producible assembly, and thus case (1) cannot hold.

Now, consider case (2). Assume that $L_2$ strictly depends upon $L_1$, but that there is no dependency path from any location in $L_1$ to any of the locations in $L_2$. Let $\vec{\alpha}$ be any assembly sequence which places tiles in $L_1$ and $L_2$. Since $L_2$ strictly depends upon $L_1$, $\vec{\alpha}$ must place a tile in $L_1$ before $L_2$ (by definition of strict dependence). However, we can now use $\vec{\alpha}$ to make a new assembly sequence $\vec{\alpha}'$ as follows. Step through $\vec{\alpha}$ one tile placement at a time. Add all tile placements from $\vec{\alpha}$ to $\vec{\alpha}'$ until any tile placement in $L_1$. Do not add that tile placement to $\vec{\alpha}'$, and from that point add all tile placements from $\vec{\alpha}$ which are not in $L_1$ and the tiles are able to attach to the assembly growing from $\vec{\alpha}'$ to that sequence as well, until and including the first tile placement in $L_2$. By the assumption that there is no dependency path from $L_1$ to $L_2$, it must be possible for $\vec{\alpha}'$ to also place the tile in $L_2$. Then, in $\vec{\alpha}'$, place the same first tile in $L_1$ which $\vec{\alpha}$ places there. This must be possible because no tile has been placed in that location in $\vec{\alpha}'$ and whichever tiles bordered that location in $\vec{\alpha}$ to allow a tile to attach there also border it in $\vec{\alpha}'$. However, this would mean that in $\vec{\alpha}'$, a tile is placed in $L_2$ before $L_1$, which is a contradiction. Therefore, there must be a dependence path from $L_1$ to $L_2$ in any valid assembly sequence in $\calT$ and Lemma \ref{lem:depend-path} is proven.

\end{proof}

\subsection{Proof of Theorem~\ref{thm:multi-imposs}}\label{sec:multi-imposs-proof}


Restatement of Theorem \ref{thm:multi-imposs}:
There exist an infinite number of aTAM systems with multi-tile seeds that cannot be shape-simulated by any aTAM system with a single-tile seed at (1) scale factors 1 or 2, or (2) scale factor 3 without using cheating fuzz.

\begin{proof}

To prove Theorem \ref{thm:multi-imposs}, we present one system which can't be shape-simulated by any aTAM system with a single-tile seed at (1) scale factors 1 or 2, or (2) scale factor 3 without using cheating fuzz, and note how an infinite set of such systems can easily be derived from it. Let $\calT = (T,\sigma,1)$, whose seed is depicted in Figure \ref{fig:multi-scale-system}, be such a system. (An infinite set of similar systems can be derived by increasing the lengths of the arms of the seed to arbitrary values and adjusting the tile set as necessary to accommodate the growth that will be described.) In $\calT$, there are five locations on the perimeter of the seed in which glues are exposed and to which tiles can attach. They are depicted in green. Each arm can grow a uniquely colored subassembly, each of which uses tile and glue types unique to that subassembly. An example assembly is depicted in Figure \ref{fig:pumped-paths-terminal}.
An infinite number of unique terminal assemblies, and assembly shapes, are possible in $\calT$, since, for instance, each path labeled 1 can be of arbitrary length. 

\begin{figure}
    \centering
    \includegraphics[width=5.0in]{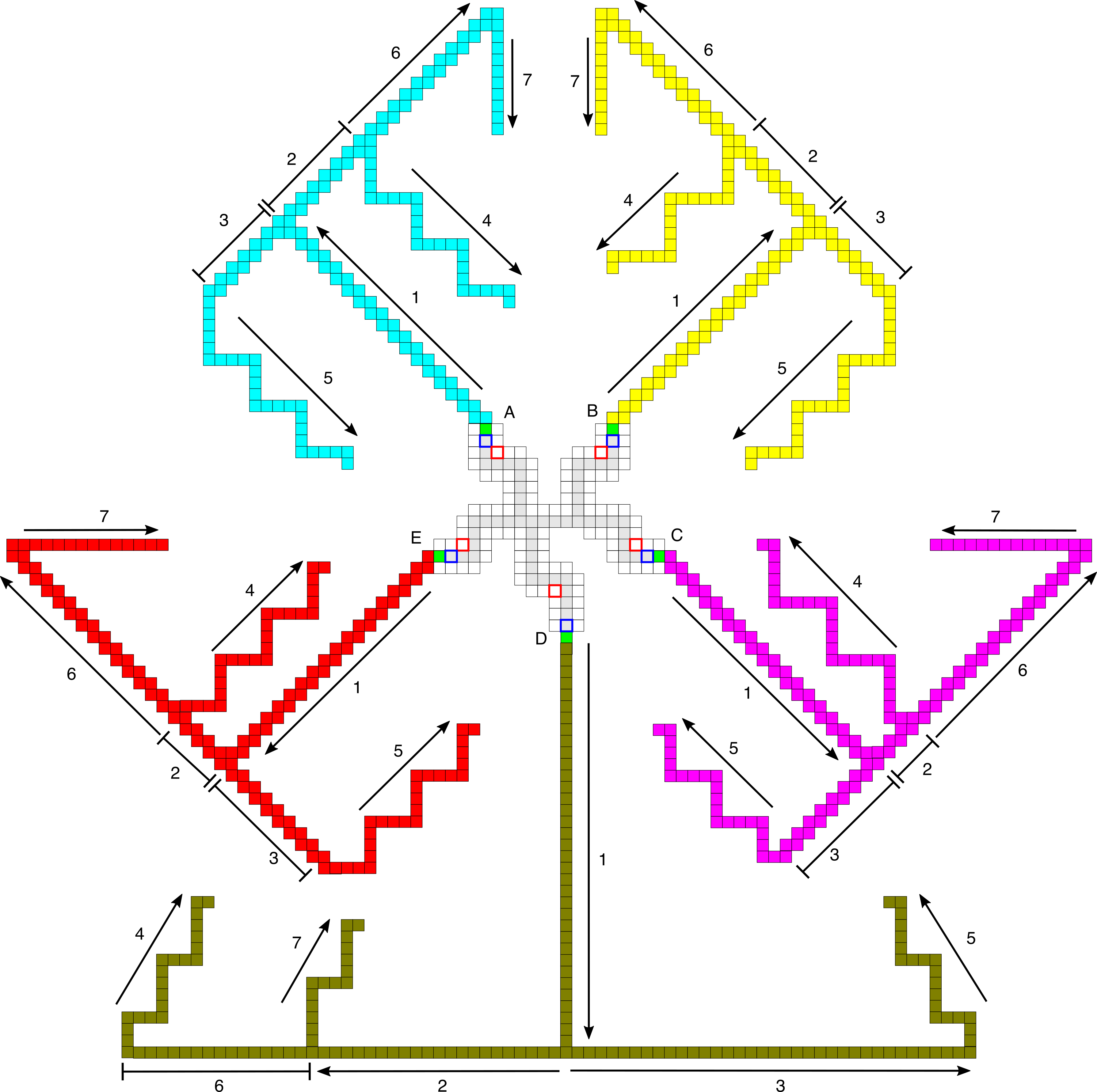}
    \caption{Depiction of an example assembly in $\calT$ (from the proof of Theorem \ref{thm:multi-imposs}) in which paths have grown from each arm of the seed. In order to show all paths, their lengths are greatly reduced, but all path segments marked with an arrow are assumed to be of length greater than the pumping length $p$. The blue, yellow, pink, and red sets of paths, which grow from seed arms A, B, C, and E, respectively, are similar modulo rotation/reflection and the lengths of the repeating periods of some of the paths labeled 4 and 5 (i.e. some of those paths travel distance 5 horizontally and 5 vertically before repeating, and some travel 6 horizontally and 6 vertically before repeating). Path segments marked with capped lines (i.e. those marked 2 and 3 in the blue yellow, pink, and red, and that marked 6 in the gold) are fixed lengths. Note that the slopes of the lines marked 7, 8, and 9 in the gold section are $\pm 6/4$, and those of the pink are red sections are all $\pm 1/1$, so all pumpable paths can be greater than the pumping length without colliding with each other and still be aimed at the appropriate seed locations.}
    \label{fig:pumped-paths-terminal}
\end{figure}

Our proof will be by contradiction, so assume that there exists singly-seeded aTAM system $\mathcal{S} = (S, \sigma_\mathcal{S}, \tau)$ such that $\mathcal{S}$ shape-simulates $\calT$ at either (1) scale factor 1 or 2, or (2) scale factor 3 but does not use cheating fuzz. Let $m$ be the scale factor used by $\mathcal{S}$ to shape-simulate $\calT$, and let $g$ be the number of unique glues types on the tile types of $\mathcal{S}$. We define $p$ to be the ``pumping length'' of a simulated one-tile-wide path in $\mathcal{S}$ and set $p = ((g+1)^{6m}\cdot (6m)!+1)\cdot 6 + 1$. This value is derived such that we can apply the Window Movie Lemma of \cite{temp1notIU}. Let $R: B^T_m \dashrightarrow S$ be the representation function that maps $m$-block supertiles over $S$ to tiles of $T$.

Figure \ref{fig:pumped-paths-terminal} shows an example of a producible, terminal assembly $\alpha \in \termasm{\calT}$ (with the pumpable path segments shortened to fit on the page). From each green location of the seed a path grows outward from the seed (these paths are labeled 1 in the figure). These paths can be arbitrarily long, and we will only consider lengths significantly longer than the pumping length of $\mathcal{S}$, $p$. Each path labeled 1 then splits, ultimately resulting in 3 paths which grow back toward the seed. Since the paths labeled 1 are significantly longer than $p$, so can all of these others be, and we consider an $\alpha$ in which this is the case.
    
We'll now refer to the blue paths originating from the seed point labeled A, and note that the same arguments hold for the yellow, pink, and red sections up to rotation. The gold section has a different geometry, but the same general principles hold.
In this assembly $\alpha$, all paths terminate, but only after growing longer than length $p$. Each has a periodic section which can repeat an arbitrary number of times consecutively, but for the paths labeled 4, 5, and 7, nondeterministically there is also a chance, after each repetition, for a ``capping'' tile to attach causing the path to terminate. However, note that the shape and slope of path 4 would allow it, if it were not capped, to extend via continued repetitions until it crashed into the green location of arm B (i.e. it would be able to place a tile in that location if there wasn't already a tile there). Similarly, the path labeled 5 could crash into the green location of arm E. Each colored section has the potential to grow paths that can crash into the green locations of the two neighboring arms.

    
    
\begin{figure}
    \centering
    \includegraphics[width=4.0in]{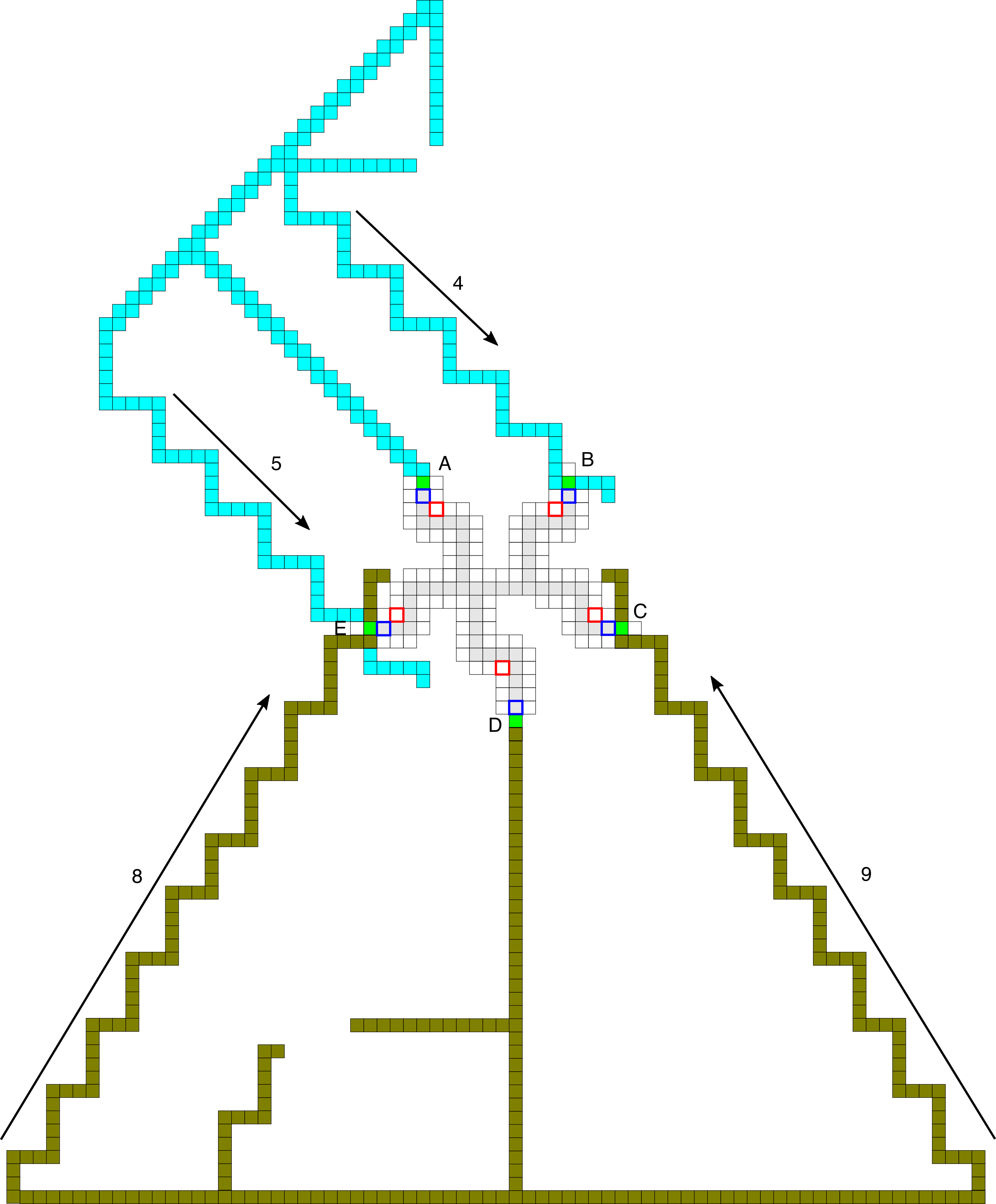}
    \caption{Depiction of an assembly (derived from $\alpha$ in the proof of Theorem \ref{thm:multi-imposs}) in which blue paths from arm A have pumped and crashed into the green locations of arms B and E, and gold paths from arm D have pumped and crashed into the green locations of arms C and E. Note that such an assembly is not actually possible in $\calT$ since all tiles of the seed, including the green tiles, are present before any paths begin growth. Additionally, the figure depicts the blue and gold paths growing through each other in arm E, which is also impossible but depicted to show the trajectory of each path.}
    \label{fig:multi-scale-crashed-paths}
\end{figure}

Since $\mathcal{S}$ is assumed to shape-simulate $\calT$, there must be a terminal assembly $\beta \in \termasm{\mathcal{S}}$ that maps, via $R$, to an assembly with the same shape as $\alpha$. Thus, there must be an assembly sequence in $\mathcal{S}$ that creates $\beta$, which we'll call $\vec{\beta}$. By the Window Movie Lemma, since each of paths 4 and 5 are of lengths greater than $p$, there must also be valid assembly sequences in $\mathcal{S}$ that produce extended versions of each path (by pumping an earlier segment before the terminating ``capping'' tile attaches). We examine a few such possible extensions and potential collision now (and a depiction of an assembly impossible to form in $\calT$, but to which the assembly maps can be seen in Figure \ref{fig:multi-scale-crashed-paths}).
    
Let $l_B$ refer to the location of the blue highlighted seed tile in arm B of $\alpha$, $l_D$ to the blue highlighted location in arm $D$, $l_{p4}$ to a location on path 4 which is greater than distance $p$ from the beginning of path 4, and $l_{p5}$ a location on path 5 which is greater than distance $p$ from the beginning of path 5. Given some location $l \in \mathbb{Z}^2$, let $R^{-1}(l)$ refer to the set of locations in $\mathcal{S}$ which map to location $l$ in $\calT$ under $R$.
    
\begin{lemma}\label{lem:dep}
$R^{-1}(l_{p4})$ strictly depends upon $R^{-1}(l_B)$.
\end{lemma}

\begin{proof}
Lemma \ref{lem:dep} says that all valid assembly sequences in $\mathcal{S}$ place a tile in the macrotile location mapping to $l_B$ before placing any tile in the macrotile location mapping to $l_{p4}$. We prove this by contradiction, so therefore assume that there exists an assembly sequence which places a tile in the location mapping to $l_{p4}$ before placing any tiles in the macrotile location mapping to $l_B$. Because the length of path 4 is greater than $p$, so is the portion of $\beta$ which maps to path 2. This means that, by the Window Movie Lemma, there must be a repeated window movie along cuts of the subassembly mapping to path 2, and therefore valid producible assemblies can also be made by pumping the repeated section of that subpath arbitrarily many times. (A schematic depiction of the pumping of a path can be seen in Figure \ref{fig:pumped-example}. Let $\beta'$ be an assembly, which must therefore be producible in $\mathcal{S}$, in which the subassembly mapping to path 2 is pumped so that it passes through the green location of arm B and places tiles which resolve to at least one more period of the path (i.e. it places the 5 tiles of each of a vertical and horizontal section of the path). This must be possible by the Window Movie Lemma (which allows the pumping) and the fact that no tiles can previously have been in locations occupied by the newly extended path since, if no tile has been placed in $l_B$, then there also cannot be tiles in the fuzz locations to its left and right (since those are not legal fuzz locations before $l_B$ has resolved to a tile in $\calT$, which it cannot since it's empty). Thus, the tiles mapping to the blue path as well as any allowable fuzz positions around it can be freely placed (in exact duplication of the pumpable ordering that previous copies must have been placed in for an earlier segment of the path). Since the macrotiles of the previous section of the path map to tiles of path 4, so must the newly pumped macrotiles since the representation function remains constant.

At this point we can note that, regardless of whatever additional tile additions may be made to $\beta'$, it must map to an assembly whose shape cannot be that of a terminal assembly of $\calT$. This is because there is no possible assembly sequence in $\calT$ in which a path with the repeating period (5 horizontal and 5 vertical positions) can extend from the right side of the green location of arm B. This is because the green tile of arm B can only initiate paths that would not be able to return to the location to its immediate east, and any path that could extend into that location from the pink growth would have a different period (i.e. 6 horizontal and 6 vertical positions). Therefore, $\mathcal{S}$ would fail to shape-simulate $\calT$ in this case. This is a contradiction to the assumption that $R^{-1}(l_{p4})$ does not strictly depend upon $R^{-1}(l_B)$, and thus Lemma \ref{lem:dep} is proven.

\end{proof}

\begin{figure}
    \centering
    \includegraphics[width=6.0in]{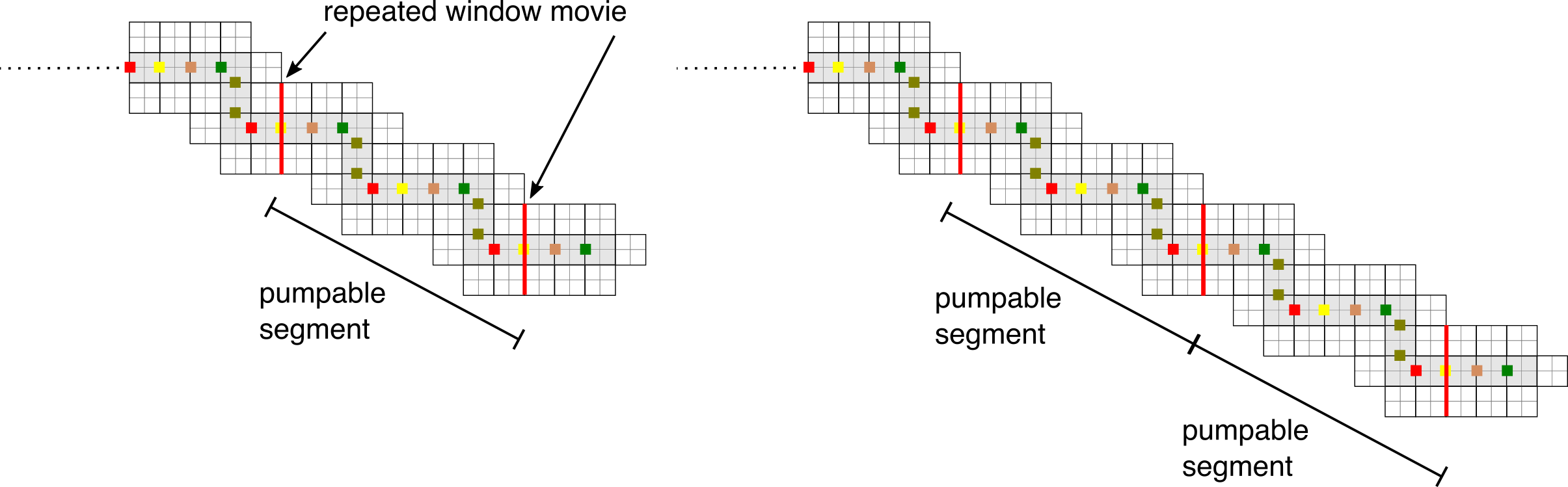}
    \caption{Example of a pumpable segment.}
    \label{fig:pumped-example}
\end{figure}

The same logic of Lemma \ref{lem:dep} can be applied to show that strict dependence occurs between two of the pumpable paths growing from each arm of the seed and the locations marked in blue in Figure \ref{fig:multi-scale-crashed-paths} of each of the two neighboring arms. By Lemma \ref{lem:depend-path}, this means that there is a dependence path going from a point in each blue-marked location, through each of the two neighboring arms, to the locations on those pumpable paths. In Figure \ref{fig:multi-scale-system}, the blue line across each arm shows a minimal cut across each. In $\alpha$ that cut crosses two locations (one tiled and one that could validly include the growth of fuzz). In $\beta$, the scaled version of each cut crosses two macrotile locations (one that must resolve to a tile under $R$, and one that is potentially allowed to contain growth of fuzz and also potentially able to eventually resolve into a tile). The scale factor $m$ of the simulation determines the size of each of these macrotiles.

\begin{figure}
    \centering
    \includegraphics[width=3.5in]{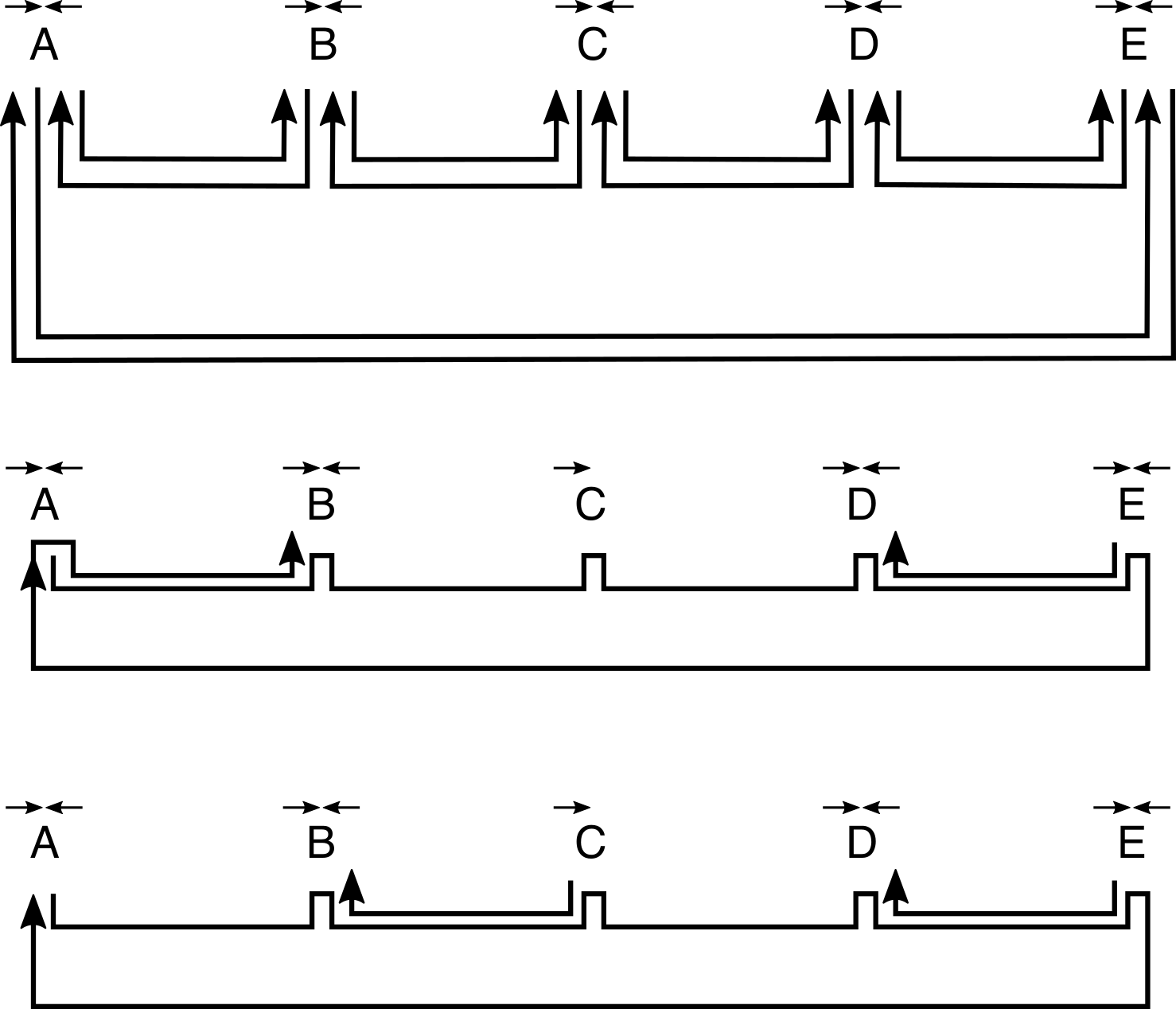}
    \caption{Schematic depiction of the dependence paths between the arms of $\beta$ from the proof of Lemma \ref{lem:4-paths}. Arrows above a label indicate which of the two neighbors has the required dependence path leading to that arm. (Top) Every dependence path is separate (i.e. formed by a disjoint set of tiles). This leads to all arms having minimal cuts across 4 tiles. (Middle) A single path originates in arm A, visits every other arm, then returns to arm B. That path satisfies the dependencies of arms A, B, and E. A separate path grows from E to D, satisfying the dependencies of D. At this point, the only ways to satisfy the missing dependency for arm C (i.e. a path from D), are to (1) continue the path from B (since it has passed through arm D) or to grow a new path directly from arm D. The first results in a minimal cut across 4 tiles in arm B, and the second results in a minimal cut across 4 tiles in arm D. (Bottom) Besides the circular path which begins and ends at arm A (satisfying dependencies for A and E), paths grow from arm C to arm B, and from arm E to arm D (satisfying the dependencies for B and D, respectively). At this point, the most succinct way to satisfy the missing dependency for arm C requires a path from D to C, which causes a minimal cut across 4 tiles in arm D. All other options either require increasing additional cut sizes and/or are identical modulo rotation.}
    \label{fig:4-paths}
\end{figure}

\begin{lemma}\label{lem:4-paths}
In assembly $\beta$ of $\mathcal{S}$ (as depicted in Figure \ref{fig:pumped-paths-terminal}), a minimal cut across the macrotiles of at least one arm (and its associated fuzz) must cross at least 4 tiles.
\end{lemma}

\begin{proof}
We prove Lemma \ref{lem:4-paths} by contradiction, so assume that there is a cut across at least one of the arms of $\beta$ that crosses fewer than 4 tiles. We first note that two consecutive tiles on a dependence path cannot also both be part of a dependence path traveling in the opposite direction. This is by the definition of a dependence path, since the $i$th tile of a dependence path must be placed before the $(i+1)$th tile, which only allows one direction for a dependence path across a pair of tiles.

We now refer to Figure \ref{fig:4-paths}, which gives a schematic depiction of the directed paths that must exist between each neighboring pairs of arms. The topmost figure depicts each path growing separately, and shows that in this case, there would need to be four paths traveling through each arm, resulting in minimal cuts across 4 tiles for every arm, so this cannot be the case. Since all dependencies must be resolved by dependence paths, the only alternative is to concatenate paths (which can allow dependencies to be resolved in a cyclic manner). Concatenating paths only serves to reduce the need for some separate path if a single path is created which visits all arms.  This is depicted in the middle and bottom figures of Figure \ref{fig:4-paths}, which show that it is impossible for all necessary dependence paths to exist without requiring at least 4 paths passing through each arm. Even though, in some cases, at different points the same path may pass into or out of an arm (e.g. when the same path travels through all arms and extends to revisit one or more arms), in each such case an additional path of tiles is required to maintain the dependencies of the full path and/or the correct directionality. Therefore, in all cases, at least one arm must have a minimal cut crossing at least 4 tiles. This is a contradiction to the assumption that some cut crosses less than 4 tiles, and proves Lemma \ref{lem:4-paths}.
\end{proof}

We now know by Lemma \ref{lem:4-paths} that the minimal cut across some arm must cross at least 4 tiles. This immediately means that the simulation scale factor of $\mathcal{S}$, $m$, must be greater than 1 since the shortest cut across each arm at scale factor one (depicted as blue lines in Figure \ref{fig:multi-scale-system}) crosses only one tile and one valid fuzz location. Additionally, for the case of $m = 3$ where cheating fuzz is not allowed, we note that the cuts represented by the blue lines in Figure \ref{fig:multi-scale-system} each cross only one macrotile location which can resolve to a tile, providing a maximum of three tiles to be crossed by the cut, and the fuzz macrotile location (outlined in red) is not allowed to receive any tiles. This is because no glues are exposed on the exterior edges of the seed tiles adjacent to those locations and since $\beta$ is terminal, no paths can crash into those locations. This means that those fuzz locations could never resolve to tiles or be in locations that map to locations adjacent to exposed glues, and thus represent cheating fuzz. Therefore, $\mathcal{S}$ must not be simulating $\calT$ at scale factor 3 without using cheating fuzz, since at least one of the dependence paths would have to place tiles in locations of cheating fuzz.

Thus, the only remaining case to consider is that in which the scale of the simulation is $m = 2$. In this case, the cuts marked in blue cross exactly 4 locations which can receive tiles. This provides enough locations for the dependence paths to grow, but requires that
if there are any unfilled locations within any of the fuzz location outlined in red, they must be contained between a pair of dependence paths and unreachable from the exterior, which is a fact we will utilize shortly.

We now have two cases to consider: (1) there exists even one assembly sequence that results in $\beta$ and which fills one of those locations in a way that the representation function $R$ maps to a tile in $\calT$, or (2) in all assembly sequences, such locations are filled so that $R$ maps them to empty space in $\calT$. In case (1), we follow such an assembly sequence and since $\beta$ is terminal, the domain of the assembly that it maps to in $\calT$ via $R$ is a domain that no terminal assembly in $\calT$ matches (since there is no chance for a tile to grow from the seed at that location and no path crashes into it) and thus $\mathcal{S}$ fails to shape-simulate $\calT$.
Therefore, case (1) cannot hold.
This leaves us with the final case to consider, that in which there is always at least one red outlined fuzz position that maps to empty space. Since any unfilled locations in such macrotiles are sealed off from the exterior, and we've already determined that no tiles can be placed within those cavities, there is no way to add tiles to those macrotile locations and thus they must always map to empty space. We can now simply note that the paths labeled 7 have the potential to grow so that they crash into each of those locations (examples seen as the two filled locations outlined in red in Figure \ref{fig:empty-crash}). This means that there are producible assemblies in $\calT$ that have tiles in those locations, but there are no producible assemblies in $\mathcal{S}$ that can map to assemblies that have tiles in those locations. Therefore, $\mathcal{S}$ also fails in the case, meaning that it fails in all cases. This is a contradiction that $\mathcal{S}$ shape-simulates $\calT$ at either (1) scale factor 1 or 2, or (2) scale factor 3 without using cheating fuzz. Since the only assumption made about $\mathcal{S}$ is that it is singly-seeded, it holds that no singly-seeded aTAM system can shape-simulate $\calT$ at (1) scale factors 1 or 2, or (2) scale factor 3 without using cheating fuzz. As previously mentioned, $\calT$ can be easily modified by changing the lengths of the seed arms (and appropriately modifying the path-growing tiles) through an infinite range, yielding an infinite set of aTAM systems with multi-tile seeds that cannot be shape-simulated by singly-seeded systems at these scale factors, and thus Theorem \ref{thm:multi-imposs} is proven.

\end{proof}

\subsection{Proof of Theorem~\ref{thm:scale4-sim}}\label{sec:scale4-sim-proof}


Restatement of Theorem \ref{thm:scale4-sim}:
Given an arbitrary aTAM system $\mathcal{T} = (T,\sigma,\tau)$, there exists an aTAM system $\mathcal{T}_4 = (T_4, \sigma_0, \tau_4=\max(2, \tau))$ which seed-first-simulates $\mathcal{T}$ at scale factor 4 and does not use cheating fuzz, where $|\sigma_0| = 1$ and $|T_4| \leq 28s + 16g + 6t$ given that $s = |\sigma|$, $t = |T|$, and $g$ is the total number of unique glue/strength combinations in $T$ and $\sigma$.

\begin{proof}
We prove Theorem~\ref{thm:scale4-sim} by construction. Let $\mathcal{T} = (T,\sigma,\tau)$ be an arbitrary aTAM system, and let $s = |\dom \sigma|$ be the number of tiles in the seed $\sigma$.

We define the system $\calT_4$ which seed-first-simulates $\calT$ by $\calT_4 = (T_4, \sigma_0, \tau_4)$. The temperature is defined by the function $\tau_4 = \max(2, \tau)$.
Temperature 1 systems can be trivially simulated by temperature 2 systems, and cooperative growth is utilized in the perimeter path to allow for correct seed-first-simulation. 
We define $T_4$ as the combination of two sets of tiles, $T_4 = T_\sigma \cup T_{IO}$.
$T_\sigma$ are the $16s$ tiles which will self-assemble the scaled version of $\sigma$.
$T_{IO}$ is the scaled expansion of the tileset $T$.


We begin by the generation of $T_\sigma$.

\begin{observation}\label{obs:scale4-contains-scale2}
Every scale 4 supertile contains a scale 2 square.
\end{observation}

Observation~\ref{obs:scale4-contains-scale2} is demonstrated in Figure~\ref{fig:S4-7_scale_4_contains_scale_2}, and from this we define this scale 2 square as the \emph{core tiles} of a scale 4 supertile.
Figure~\ref{fig:S4-7_scale_4_contains_scale_2} also demonstrates that the core tiles of adjacent supertiles can be connected via 2 additional tiles in each supertile, demonstrated by the red lines.
The presence of a Hamiltonian cycle is proven to exist for shapes of scale factor 2, as shown in \cite{SummersTemp}.
This allows us to guarantee that a core path can be created which visits each supertile contained in the seed.
The generation of this Hamiltonian cycle follows the procedure developed in \cite{griffith2004growing} which consists of the following two steps.

\begin{figure}
    \centering
    \includegraphics[width=0.25\textwidth]{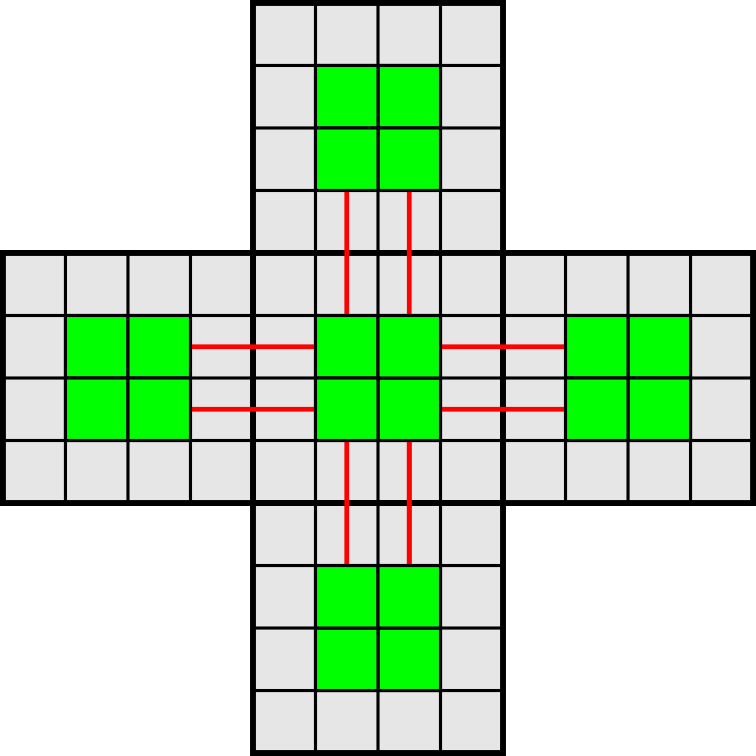}
    \caption{The scale 2 square contained inside the scale 4 supertiles is defined by the green tiles. Note that scale 2 squares of adjacent supertiles can be connected by 2 tiles, indicated by the tiles under red lines.}
    \label{fig:S4-7_scale_4_contains_scale_2}
\end{figure}

First, a spanning tree must be generated; this can be found utilizing breadth-first or depth-first search algorithms.
For convenience, we define the \emph{origin} supertile as the westernmost tile in the southernmost row of the seed.
The root of the spanning tree is set as the origin from which either breadth-first or depth-first search is carried out.

Second, we are able to replace each tile in the spanning tree with an associated supertile based upon its neighbors - we begin with the origin supertile.
If $s=1$, the system $\calT$ already contains a single-tile seed, and we utilize a specific origin tile; this is case 0).
Due to the location of the origin supertile, only 3 additional cases exist for how the origin tile is connected to the remainder of the seed: 1) the origin tile is connected to the remainder of the seed by its north edge only, 2) the origin tile is connected to the remainder of the seed by its east edge only, 3) the origin tile is connected along both its north and east edges.
Cases 0), 1), 2) and 3) are illustrated as $\sigma$, N, E and N+E in Figure~\ref{fig:S4-1_origin_macrotile}, respectively.
The new seed tile of our single-tile seed ($\sigma_0$) is defined by the teal tiles in Figure~\ref{fig:S4-1_origin_macrotile}; these are the tiles which encode the start of core path.
The remaining tiles of the spanning tree are represented by a set of five supertiles (adapted from \cite{griffith2004growing}), shown in Figure~\ref{fig:S4-2_seed_path_macrotiles}.
These five supertiles (and their rotations) allow for the core path to be connected to any supertile in the representation of a seed.
The tiles of the core path are assigned once all supertiles are assigned to the spanning tree.
The first tile of the core path is $\sigma_0$, and to ensure that the core path is a dependence path tiles are added with unique strength $\tau_4$ glues along adjacent edges.
Additionally, each tile in the core path assigns a unique strength $\tau_4 - 1$ glue on the edge facing the exterior of the supertile; this glue is utilized in the creation of the precedence path.


    \begin{figure}
        \centering
        \includegraphics[width=0.65\textwidth]{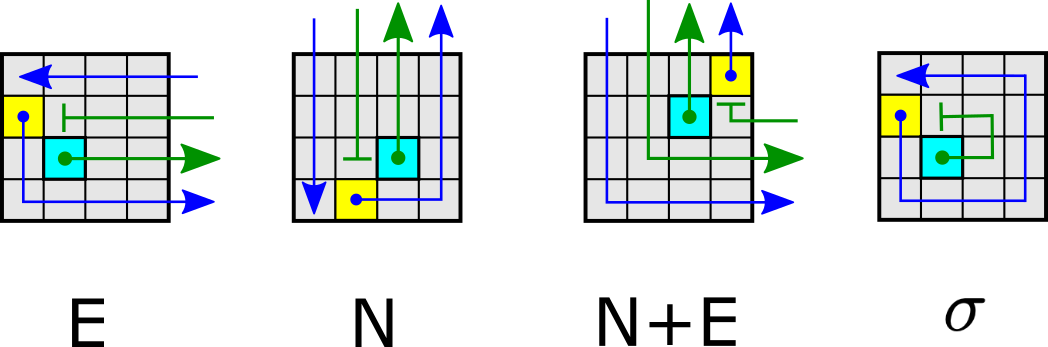}
        \caption{The four possible supertiles which allow for the new single-tile seed system to begin growth. The green arrows indicate the order of tile connections of the core path, and the blue arrows indicate the order of the perimeter path. The teal tile indicates the location of the new seed tile $\sigma_0$, the yellow tile indicates the first tile placed by the perimeter path. The yellow tile is connected by a strength 2 glue to the final tile of the core path. The rightmost origin tile is the case of system being simulated already containing a single-tile seed.}
        \label{fig:S4-1_origin_macrotile}
    \end{figure}
    
    \begin{figure}
        \centering
        \includegraphics[width=0.5\textwidth]{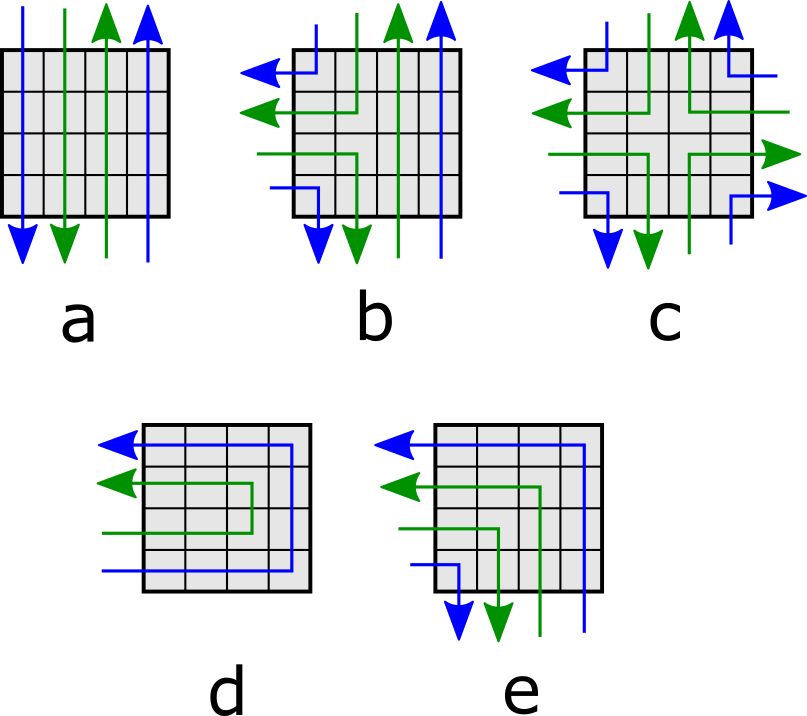}
        \caption{The five possible supertiles which allow for the core and perimeter path to be connected to each tile in the spanning tree of the seed. Green arrows indicate order in which the tiles of the core path are connected to each other. Blue arrows indicate the order in which the tiles of the perimeter path follow the core path. Note that these supertile templates may be rotated to connect adjacent supertiles as necessary.}
        \label{fig:S4-2_seed_path_macrotiles}
    \end{figure} %

The final process in the creation of the seed is creating the tiles of the \emph{perimeter path}.
This is the set of tiles that allow for the encoding of the glues which are present on the exterior of the seed.
The first tile of the perimeter path is the yellow tile in the origin supertiles of Figure~\ref{fig:S4-1_origin_macrotile}, and is connected via a unique strength $\tau$ glue to the last tile of the core path (indicated by the vertical bar).
The remaining tiles of the perimeter path then `follow' the direction of the core path.
For perimeter path tiles which share an edge with tiles of the core path, a unique strength $\tau_4 - 1$ glue is present along that edge for each supertile.
A strength 1 glue of type $p$ is present along the the edge the perimeter tile shared with its predecessor in the perimeter path.
Certain perimeter path tiles may not share an edge with a tile of the core path - this occurs in tiles of types d and e (on the corners) in Figure~\ref{fig:S4-2_seed_path_macrotiles} when the perimeter path takes a 90 degree turn.
In this case, the dependence path tiles which share two edges with other dependence path tiles have one edge with a unique strength $\tau$ glue connecting to the preceding tile in the dependence path.
The second edge contains a strength 1 glue of type $p$, allowing for the next tile in the dependence path to attach to the seed assembly using cooperative growth. 
On the exterior-facing edges of all perimeter path tiles, the glues which represent the glues of the simulation of the tile being simulated by the supertile are added to the tiles of the perimeter path following the encoding and point of competition demonstrated in Figure~\ref{fig:S4-8_scale_4_template}.
The glues exposed by the perimeter path function identically to that of the remaining scale 4 supertile representations, allowing for growth of fuzz into points of cooperation (the process by which will be outlined).

A visualization of the process of seed tile generation is presented in Figure~\ref{fig:S4-4_spanning_tree}, from the initial scale 1 seed to the creation of a perimeter path.

    
To complete the creation of the full tile set, we must create a supertile representation of each tile type in $T$ that can attach outside of the assembly representing the seed $\sigma$.
An interesting item to note here is that, for each tile type $t \in T$ for which one or more tiles of type $t$ appear in the seed assembly $\sigma$, there will be two logically different types of supertiles that represent $t$ in $\calT_4$.
One will be that which is grown via the tiles of $T_\sigma$ (and which contains the core and precedence paths), and one which is a ``standalone'' supertile that represents a copy of a tile of type $t$ outside of the seed.
The result of the transformation to make the $T_{IO}$ (a.k.a. ``inward/outward'') tile set is that it is never possible to regrow any portion of a version of a supertile that is specific to the seed in any location outside of that representing the seed, and to always instead allow supertiles of the standalone version for type $t$ to grow in those locations.


We generate this tile set in the manner similar to that outlined in \cite{Versus} in the sections ``minimal glue sets'' and ``inward-outward glues'' (and note that this technique is now relatively common in tile assembly results).
$T_{IO}$ is an expansion of the tiles of $T$ so that for each $t \in T$, we make a new tile for every subset of glues that share glue labels with the seed whose strengths sum to ``just barely'' $\tau_4$.
Each glue in $T$ is expanded to have N, E, S, W variants corresponding to the the direction they are pointing, also represented by the unit vectors $\{(0,1),\: (1,0),\:(0,-1),\:(-1,0)\}$.
For example, a glue of type $x$ would be represented by types $x_N,\: x_E,\: x_S,\: x_W$.
We say a glue is pointing \emph{inwards} if the the glue type has a direction opposite of the edge which it resides (e.g. $x_N$ on a south edge of a tile or $x_E$ on a west edge).
Alternatively, a glue is pointing \emph{outwards} if the glue type matches the edge which it resides (e.g., $x_N$ on a north edge of a tile or $x_E$ on an east edge).
For each tile $t \in T$, we consider every combination of that tile's glues.
We define a \emph{minimal glue set S} as a set of glues such that, if any glue is removed from S, the set falls below combined strength $\tau_4$.
For each minimal glue set, we generate a tile where the glue labels of the minimal set are such that they are pointing inwards (i.e. they act as ``input'' glues), and the remaining glue labels are pointing outwards (i.e. they act as ``output'' glues).
Finally, the only tiles of $T_\sigma$ which are modified are perimeter path tiles which represent exterior-facing glues. 
These perimeter path tiles are modified such that their glues labels are modified to be outward facing.
This transformation ensures the invariant that any tile which is designed to be part of a supertile that represents a tile in $\sigma$ appears only in the assembly mapping to $\sigma$, and prevents incorrect regrowth of any portion of the seed outside of that, since the seed supertiles are designed to only grow a full copy of the seed, rather than allowing individual supertiles to grow independently.






Next, we create the tiles to form the scale four supertiles representing each tile generated by $T_{IO}$. 
We begin by defining a point of cooperation and/or competition in each supertile which allows for supertiles of $T_{IO}$ to be mapped from assemblies of $\calT_4$ to $\calT$ by $R^*$.
Figure~\ref{fig:S4-8_scale_4_template} shows this cooperation point with the tile labeled `C'.
Given this point of cooperation, for each $t \in T_{IO}$ we provide a single tile with the minimum glue set of inward glue labels at the prescribed strength.
Whichever tile type attaches in location C of a supertile fully determines the tile type in $T$ that the supertile maps to under $R$.
Each neighboring supertile that has grown to represent a tile in $T$ grows a path toward the C locations in the neighboring supertiles that haven't yet filled in.
For those whose edges represent adjacent glues of strength $\tau$, this path, if unblocked, can grow into the neighboring supertile's C location without cooperation to determine the type of tile that that supertile will represent all by itself, which is by design since the $\tau$-strength glue being represented would be able to cause the neighboring tile to bind all by itself.
In this case, the location C acts as a point of ``competition'' where various growth paths may be competing to be the first to arrive and place a tile there.
For those that edges that represent glues of strength less than $\tau$, upon reaching the C location in a neighboring supertile the path stops and exposes a glue of strength less than $\tau_4$, which forces a tile placed in the C location to use cooperation if it is going to bind using that path.
This is analogous to the cooperation that would be required with the glue exposed by the tile which is represented in $\calT$, and in this case location C acts as a point of ``cooperation''.
As soon as the C location of a supertile receives a tile, that supertile represents the corresponding tile in $T$, and the output sides of the tile in location C each have $\tau_4$ strength glues which match with a set of 3 tiles that grow towards the cooperation point of an adjacent supertile.
We note that these paths are legal fuzz, since the supertile with the C location filled in now maps to a tile and those locations are thus (non-diagonally) adjacent to at least one mapping supertile.

\begin{figure}
    \centering
    \includegraphics[width=0.3\textwidth]{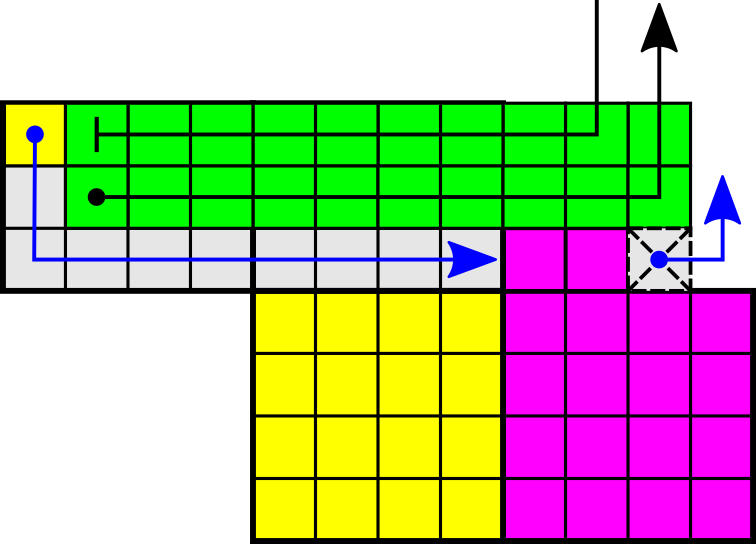}
    \caption{A possible case of the perimeter path growth around the supertiles representing $\sigma$ being blocked by growth of non-seed representing supertiles. The supertile represented by the fuchsia tiles, which itself is allowed to grow from the supertile represented by the yellow tiles, grows north to the point of cooperation. Further growth north of the fuchsia tiles is blocked by the presence of the core path, but the currently placed tiles of the perimeter path can no longer cooperate with the core path to allow for placement of additional perimeter path tiles. The next tile of the perimeter path is indicated by the grey tile with dashed edges - it cooperates with a strength 1 `p' glue provided the fuchsia tiles in the path along with the $\tau_4 - 1$ strength glue provided by the core path tile to continue the growth of the perimeter path.
    }
    \label{fig:S4-6_perim_path_collision}
\end{figure}

At this point, we've described both portions of the tile set of $\calT_4$: $T_\sigma$ which contains the tiles that grow the supertiles representing the seed $\sigma$, and $T_{IO}$ which contains the tiles that form the supertiles that represent all tiles not contained within the seed-representing assembly.
We've shown how the single seed tile in $\calT_4$ initiates growth of first a core path that goes through the center of every supertile of the seed-representing assembly (providing the correctness of the ``seed-first'' portion of the simulation, since these supertiles can then map to the full seed assembly before allowing any growth outside of the seed).
It then initiates the growth of the perimeter path, which circles the entire assembly and effectively ``activates'' each supertile by placing glues on their perimeters that allow them to initiate growth away from the supertile.
The growth away from the supertile is handled by the tile types in $T_{IO}$, which are able to correctly grow the $4 \times 4$ supertiles representing each tile type of $T$ by growing paths representing glues from neighboring supertiles to points of cooperation/competition that correctly handle the selection of valid supertiles to grow into, which in turn then grow the necessary output glue paths to neighbors.
Only in cases where the tile represented by a supertile has a glue in direction $d$ that wasn't used in its initial binding (i.e. an ``output'' glue) does that supertile grow a path into the neighboring supertile.
Since the supertile originating the output path maps to a tile in $T$ at this point, this represents legal fuzz, and since sides which do not have output glues do not grow such fuzz, cheating fuzz is never grown.
Furthermore, the careful design of $T_{IO}$ and the exterior glues of the seed-representing structure using the ``inward/outward'' glue conventions ensures that the tiles of $T_\sigma$ never appear outside of the seed-representing supertiles.

To complete the proof of correctness of the seed-first-simulation by $\calT_4$, the final case left to consider is that where a portion of the assembly representing the seed has completed its perimeter path, allowing growth to proceed away from the seed while the perimeter path hasn't yet completed for another portion of the seed-representing assembly.
This growth away from the seed could potentially eventually crash into a portion of the seed-representing assembly which hasn't yet completed growth of its perimeter path. Potentially this could obstruct the completion of the perimeter path and thus stall the completion of the seed structure.
More specifically, this may occur where the tiles attempting to grow an output path from a completed supertile into the location of a neighboring supertile (which happens to be one in the seed structure that hasn't yet completed its perimeter path), to the cooperation/competition point of that supertile. This may `obstruct' the normal growth of the perimeter path (and is demonstrated in Figure~\ref{fig:S4-6_perim_path_collision}).
To allow for the continued growth of the perimeter path, the tile of the output path growing into the seed supertile's location contains a strength 1 $p$ glue on exposed sides along the direction of the perimeter path. Thus in the case that growth of the perimeter path is `cut off'; east and west growing output tiles have a $p$ glue on their north and south edges, and north and south growing output tiles have a $p$ glue along their east and west edges.
Additionally, depending upon the direction from which the output path grows into the perimeter path, it may block the corners diagonal to the growth path.
The output path growth from each direction will place a tile in the diagonal corners of the supertiles, denoted by the locations highlighted in white in Figure~\ref{fig:S4-8_scale_4_template}.
In this way, it is ensured that a perimeter path can always complete and that growth away from the seed cannot cause incorrect or incomplete growth.
Note that in cases of these collisions, the tiles on both sides of the collision map to the correct tiles in $T$ under $R$, and these cases are representative of situations where there is blocking performed by a seed tile.

We investigate the overall tile complexity of $\calT_4$.
For the tile complexity of $\calT_4$, we define the function $C_4(s,t,g) : \mathbb{N} \rightarrow \mathbb{N}$ where $s = |\sigma|$, $t = |T|$, $g$ is the total number of unique glue/strength combinations in $T$ and $\sigma$.
We consider the worst possible upper bound for each of the items, and thus our tile complexity is a conservative estimate.
This function takes into account the number of tiles utilized by both $T_\sigma$ and $T_{IO}$.

For $T_\sigma$, each tile utilizes all $16s$ tiles to include the perimeter path and core path.
We then consider the fuzz required for each edge of a tile which contains an exterior facing glue.
At maximum, each tile in the seed can have 3 exterior facing glues; we take this as our upper bound.
For each edge with tiles will have up to 2 additional tiles required to grow to the point of cooperation from the perimeter.
This leads to an additional $3 \times s \times 2 = 6s$ tiles.

Let us consider the tiles of $T_{IO}$.
For output glues, fuzz can be shared between different tiles; it is the tile at the point of cooperation/competition which causes the output to be differentiated.
Glue has an input and output variant, and we assume this is present for all 4 sides.
4 tiles of fuzz are required per side to grow to adjacent points of cooperation/competition, leading to $2 \times 4 \times 4 \times g = 16g$.
Finally, we consider the $T_{IO}$ expansion.
From \cite{Versus}, the size of the minimal glue set is at most 6 (4 choose 2) for a tile at maximum, thus we require at most 6 tiles for each tile in $t$ for our expansion.
We additionally include tiles of the seed in this count - leading to an additional $6t + 6s$ tiles.

Thus, the conservative estimate for the tile complexity is $C_4 = 16s + 6s + 16g + 6t +6s = 28s + 16g + 6t$; as such, $T_4 \leq 28s + 16g + 6t$.

Thus, $\calT_4$ correctly seed-first-simulates an arbitrary aTAM system $\calT$ at scale factor 4 without making use of cheating fuzz using $T_4 \leq 28s + 16g + 6t$ tiles, and so Theorem~\ref{thm:scale4-sim} is proven.
\end{proof}

\subsection{Proof of Theorem~\ref{thm:scale3-sim}}\label{sec:scale3-sim-proof}


Restatement of Theorem \ref{thm:scale3-sim}:
Given an arbitrary aTAM system $\mathcal{T} = (T,\sigma,\tau)$, there exists an aTAM system $\mathcal{T}_3 = (T_3, \sigma_0, \tau_3= \max(2, \tau))$ which seed-first-simulates $\mathcal{T}$ at scale factor 3 utilizing cheating fuzz, where $|\sigma_0| = 1$ and $|T_3| \leq 20s + 16g + 6t$ given that $s = |\sigma|$, $t = |T|$, and $g$ is the total number of unique glue/strength combinations in $T$ and $\sigma$.

\begin{proof}
We prove Theorem~\ref{thm:scale3-sim} by construction.
Let $\mathcal{T} = (T,\sigma,\tau)$ be an arbitrary aTAM system.  
We generate the tile system $\mathcal{T}_3 = (T_3, \sigma_0, \tau_3= \max(2, \tau))$ taking as input tileset $T$, seed $\sigma$ and temperature $\tau$.
The first step is to generate the tileset $T_\sigma \subset T_3$ which allows for growth of the seed assembly.
We develop a template for all scale 3 supertiles which allows for the creation of tiles which allow for seed-first simulation, shown in Figure~\ref{fig:S3-14_core_path}.
We show by induction that utilizing this template (without specifying exact tile locations at this point),  

\begin{figure}
    \centering
    \includegraphics[width=0.7\textwidth]{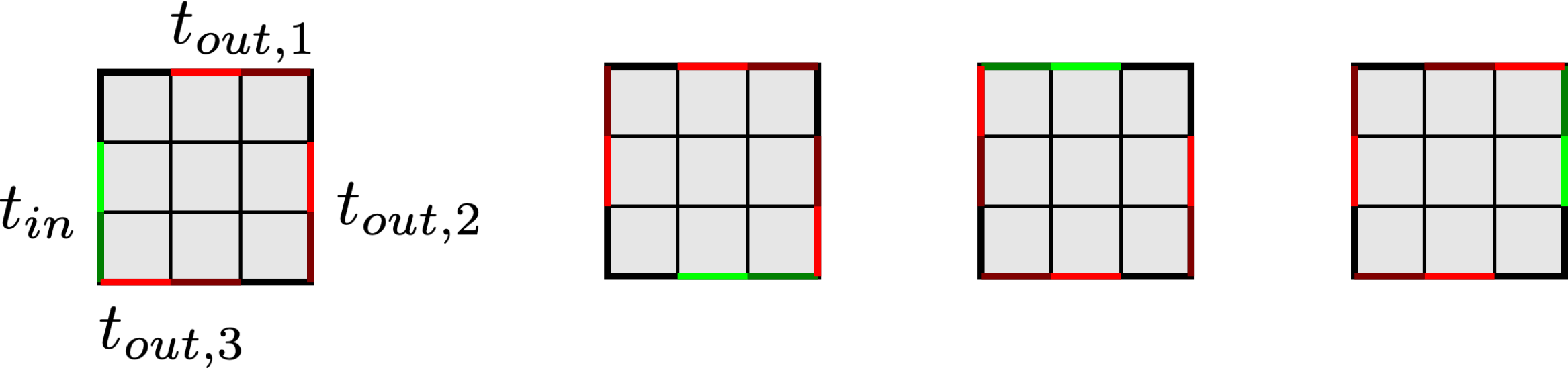}
    \caption{The basic outline of a scale 3 supertile $t$. This general tile makes no assumption on the path which is utilized to represent the tiles of the dependence path within the supertile. The green edge represents the edges shared with the neighbor supertile $t_{in}$ which contains tiles of the dependence path with an index smaller than any tile in $t$ (i.e., assembled earlier than $t$). The red edges represent the edges shared with up to 3 neighbor supertiles $t_{out,1},\:t_{out,2},\:t_{out,3}$ such that there exists tiles of $t$ in the dependence path with indices smaller than any tile in $t_{out,1},\:t_{out,2},\:t_{out,3}$. Two edges are taken, as this allows for the dependence path to both extend and return through the edges shared with $t_{in}$. The three tiles on the right demonstrate the rotations of this tile.}
    \label{fig:S3-14_core_path}
\end{figure}

\begin{lemma}\label{lem:S3-supertile-template-connectivity}
Given a finite scale 1 assembly, $\alpha_1$, we can generate an arbitrary finite scale 3 assembly, $\alpha_3$, such that a dependence path $l$ exists which visits each supertile and returns to the supertile which contains the tile.
\end{lemma}
\begin{proof}
We prove by induction.
Let us consider the binding graph of $\alpha_1$.
The base case is the westernmost tile of the southernmost row; we assign the supertile of Figure~\ref{fig:S3-14_core_path} to this location.
Due to the location of the tile, no other tile will be present along the location $t_{in}$; as such, this contains the minimal value of the dependence path.
We call this tile the `origin' tile.
For this tile, we investigate the tile locations adjacent to the origin tile in $\alpha_1$.
For each location of $t_{out,i}$ which contains a tile where $i \in \{1,2,3\}$, we rotate the supertile of Figure~\ref{fig:S3-14_core_path} as necessary such that $t_{in}$ matches up with the edge of $t_{out,i}$
Once all supertiles adjacent to the origin (which is at most 2) have been added, we select the an arbitrary supertile which has neighbor tiles adjacent to it in $\alpha_1$ which are not part of $\alpha_3$. 
This repeats until all locations of $\alpha_1$ have an associated supertile in $\alpha_3$.
At termination, the resulting $\alpha_3$ then contains a supertile in each location of $\alpha_1$ with a dependence path $l$ which visits each supertile and begins and ends the origin.  
\end{proof}

With Lemma~\ref{lem:S3-supertile-template-connectivity} we have a template which to define our supertiles, which will guarantee a dependence path through an assembly $\alpha_3$.
In addition, we define the single tiles which provide input and output to the system by specific edges of Figure~\ref{fig:S3-14_core_path}.
The light green edge shares a glue from the precedent path of the prior supertile adjacent to the current supertile.
Light red edges are those where the precedent path exits the current tile.
Dark red edges are those where the precedent path re-enters the current tile.
Finally, the dark green edge is that which the precedent path returns to the prior supertile.
As such, we then are restricted for the precedent path traveling between specific edges; tiles adjacent to light green can only have a precedent path which leads to a tile adjacent to dark green or light red.
Tiles adjacent to light red can only have a precedent path which leads to a tile adjacent to dark red or dark green.

Using the prior restrictions, the process of seed tile creation begins by defining the origin supertile.
As in Lemma~\ref{lem:S3-supertile-template-connectivity}, we define the origin supertile as the westernmost tile in the southernmost row of the seed.
If the system $\calT$ already contains a single-tile seed, we utilize a specific origin tile; this is case 0).
Due to the location of the origin supertile, only 3 possible cases exist for how the origin tile is connected to the remainder of the seed: 1) the origin tile is connected to the remainder of the seed by its north edge only, 2) the origin tile is connected to the remainder of the seed by its east edge only, 3) the origin tile is connected along both its north and east edges.
Cases 0), 1), 2) and 3) are illustrated as $\sigma$, N, E and N+E in Figure~\ref{fig:S3-8_origin_macrotile}, respectively.

\begin{figure}
    \centering
    \includegraphics[width=0.4\textwidth]{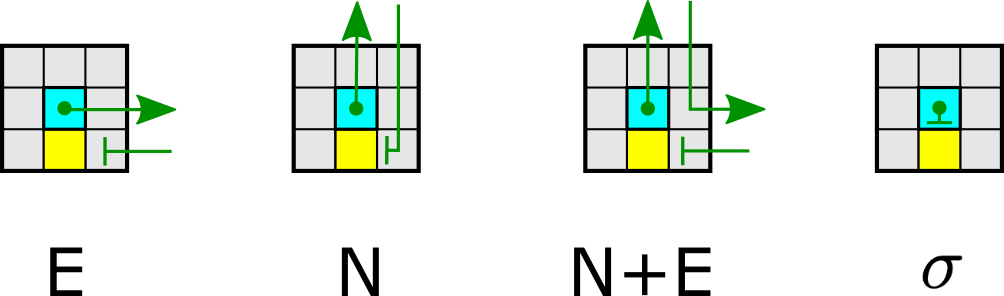}
    \caption{The 4 possible origin supertiles of westernmost tile of the southernmost row of the seed. `N+E' is the case where the origin supertile is connected to two neighbors along both its North and East edges, whereas `N' and `E' are connected to only one neighbor along their North or East edge, respectively. $\sigma$ is utilized in the case that the system being simulated is already singly seeded. The teal tile represents the location of the new single tile seed. Green bars indicate location of the final tile of the precedent path. The yellow tile represents the location of the first tile of the precedent path.}
    \label{fig:S3-8_origin_macrotile}
\end{figure}

As compared to scale 4, the number of general supertiles is increased in scale 3.
Simply applying rotation to a single case of a supertile with 3 neighbors (similar to supertile b in Figure~\ref{fig:S4-2_seed_path_macrotiles}) would result in the dependence path unable to be connected between neighbor tiles.
An additional difference is that instead of the dependence path taking the the clockwise most neighbor, it utilizes the counter-clockwise most neighbor.
This results in a total of 8 supertiles, as shown in Figure~\ref{fig:S3-2_seed_path_macrotiles}.
We note that these 
These tiles have an additional invariant - any possible perimeter path which utilizes fuzz outside the borders of the scale 3 supertile does not utilize diagonal fuzz.
While each of the tile types may need cheating fuzz, we must guarantee that fuzz is not diagonal.
We note that diagonal fuzz is only reachable from a subset of the supertiles in Figure~\ref{fig:S3-2_seed_path_macrotiles} which may be present on convex corners of $\sigma$: these are a, c and d.
The perimeter path for this scale 3 construction also follow the dependence path, similar to scale 4.
By inspection and comparing the locations of diagonal fuzz shown as red locations in Figure~\ref{fig:S3-10_cheating_fuzz_locations}, no locations which contain cheating fuzz overlap with diagonal fuzz.
Any fuzz which may be required by the perimeter path to be in a diagonal location is present on corners where two tiles are adjacent - as such, this is a location of acceptable cheating fuzz.

\begin{lemma}
The supertiles of Figure~\ref{fig:S3-2_seed_path_macrotiles} generate a dependence path which visits each supertile and returns to the supertile which contains the origin. 
\end{lemma}
\begin{proof}
The tiles of Figure~\ref{fig:S3-2_seed_path_macrotiles} obey the edge restrictions of the template scale 3 tile; as such, Lemma~\ref{lem:S3-supertile-template-connectivity} follows for these tiles.
\end{proof}
 

\begin{figure}
    \centering
    \includegraphics[width=0.85\textwidth]{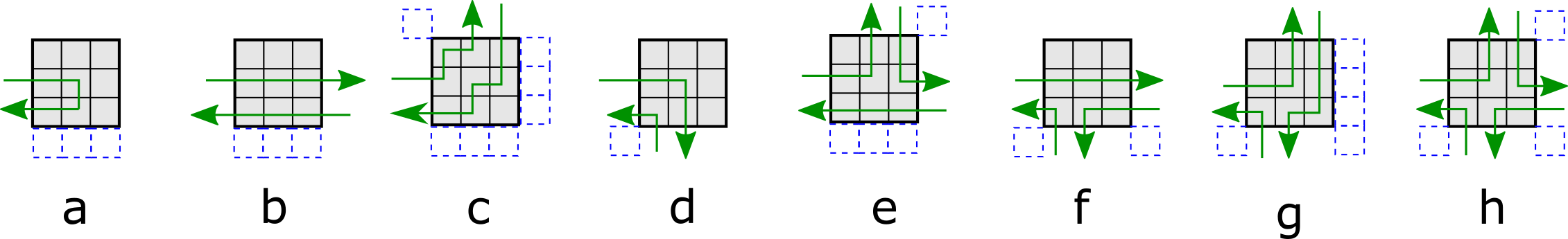}
    \caption{The supertile types which allow for the dependence path to follow the depth-first spanning tree depending upon the side of the origin tile which the tile can be connected through. Note, we can rotate these to allow for all possible path through the seed assembly. Blue dashed boxes indicate the locations of cheating fuzz which may be needed for the perimeter path.}
    \label{fig:S3-2_seed_path_macrotiles}
\end{figure}

The application of general set of scale 3 supertiles to the spanning tree of an arbitrary seed has the potential to isolate cavities from the perimeter path.
For full seed-first-simulation (and thus, also shape-simulation) it may be required for tiles to attach to glues inside of cavities.
Figure~\ref{fig:S3-6_cavity_enclosure} demonstrates a case where a seed which is assigned supertiles from Figure~\ref{fig:S3-2_seed_path_macrotiles} and a cavity is prevented from being accessed by the perimeter path.
To resolve this issue, we require two solutions: 1) developing a set of tiles which we guarantee allow for the perimeter path to access a cavity regardless of its location, and 2) correctly placing the tiles in the seed.
At a high level, we choose the northernmost corner of the westernmost edge of the cavity and draw a line in the $+y$ direction until it reaches either another cavity or the perimeter of the seed.
All the edges which this line passes through we remove.
We demonstrate his can be done safely in the case of \emph{cyclic cavities} - cavities for which there exists a cycle between the vertices of the binding graph of the seed which contain only the vertices within a Chebyshev distance of 1 from any location within the cavity.
In the case of \emph{non-cyclic cavities}, cavities for which there does not exist a cycle between the vertices of the binding graph of the seed which contain only the vertices within a Chebyshev distance of 1 from any location within the cavity, we demonstrate that they can be recursively combined until either they form a cyclic cavity or they are connected to the perimeter path.

\begin{figure}
    \centering
    \includegraphics[width=0.7\textwidth]{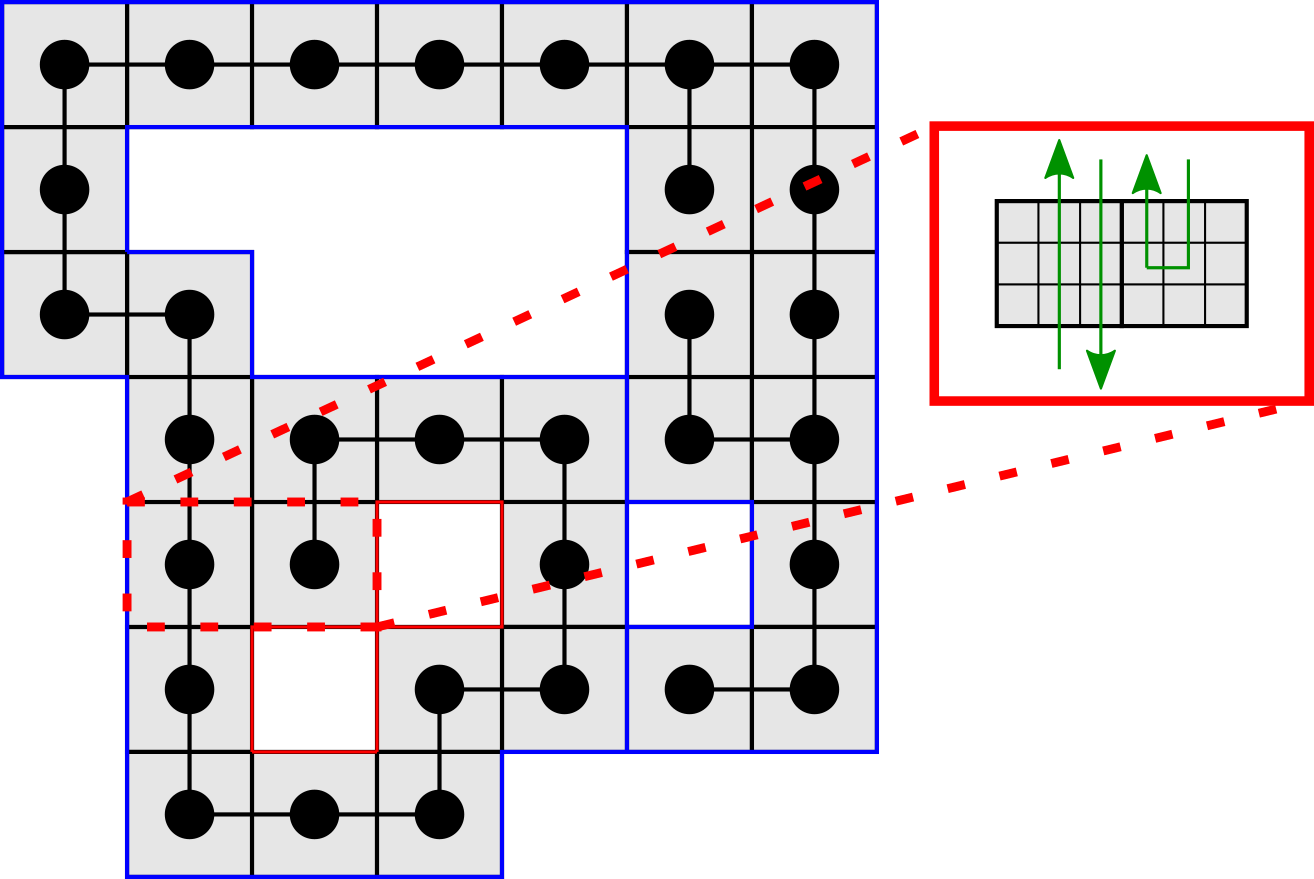}
    \caption{A valid assignment of supertiles which leads to cavities which are not reachable by the perimeter path. The red line indicate the cavities which are connected together but unable to be accessed from the perimeter path on the exterior of the seed. Due to the orientations of tiles, the dependence paths are adjacent to one another and block any possible perimeter path from reaching the internal cavities.}
    \label{fig:S3-6_cavity_enclosure}
\end{figure}

We first provide a description of cyclic and non-cyclic cavities, with the smallest possible cyclic cavity (a single tile missing from the center of a $3 \times 3$ square) demonstrated in Figure~\ref{fig:S3-4_cavity_types}.
We define \emph{regions} as (potentially infinite) connected subsets of $\mathbb{Z}^2$.
The regions split into 3 categories - cavities, the assembly, and \emph{free space} (the remaining plane).
When taking the complement of the domain of an assembly, there exists the infinite connected subset of $\mathbb{Z}^2$ which is the free space.
Additionally, for shapes which contain cavities (sometimes called holes in the literature), we will find finite connected subsets - these are the cavities within an assembly.
Given that a cycle is formed around cyclic cavities, we can safely remove a single edge from the cycle without causing a tile in that cycle to become disconnected.
Non-cyclic corners either separate two cavities in scale 1, or separate a cavity from the exterior of the shape.
Figure~\ref{fig:S3-5_non_cyclic_corners} demonstrates that in the worst case, two tiles which are at the non-cyclic corners of a scale 3 supertile will have a 1 tile wide path.
This path allows for tiles to be generated which connect the two regions - as such, we then can assign a new cavity from the union of the two connected subsets of $\mathbb{Z}^2$.
Either the two regions are a both non-cyclic cavities and can be combined, or the non-cyclic cavity is connected to the perimeter (and thus, free space).
In the latter case that the non-cyclic cavity is connected to the perimeter, we are done; no further modifications need to be made.
In the former case, the two non-cyclic cavities are joined and can the resulting cavity may remain a non-cyclic cavity.
We continue the process of joining the non-cyclic cavity together with its neighbor region.
Otherwise, the two non-cyclic regions may become a cyclic region themselves, at which point we halt.
Once all non-cyclic cavities have been combined and are either part of cyclic cavities or connected to the perimeter, the binding graph can then be modified.

\begin{figure}
    \centering
    \includegraphics[width=0.4\textwidth]{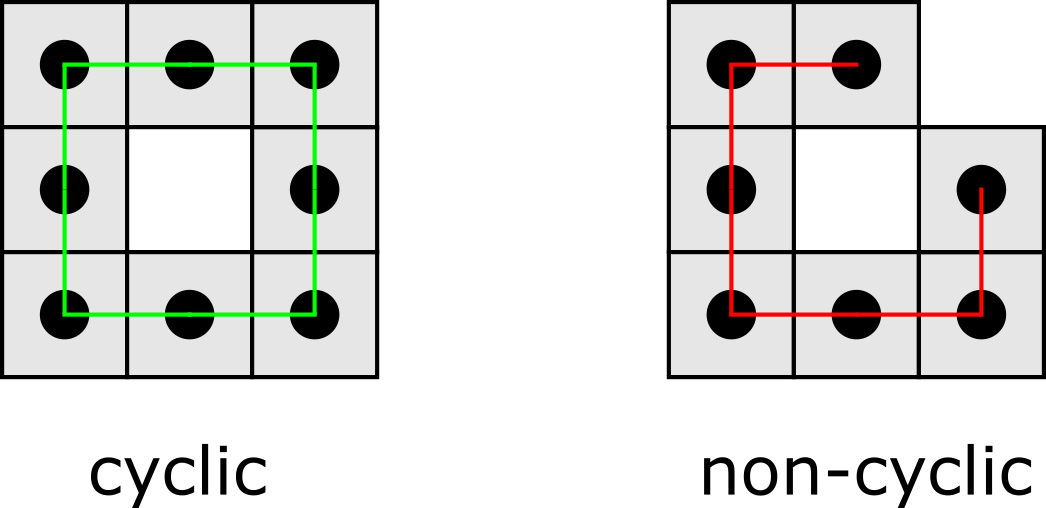}
    \caption{The two types of cavities possible in assemblies.}
    \label{fig:S3-4_cavity_types}
\end{figure}

\begin{figure}
    \centering
    \includegraphics[width=0.4\textwidth]{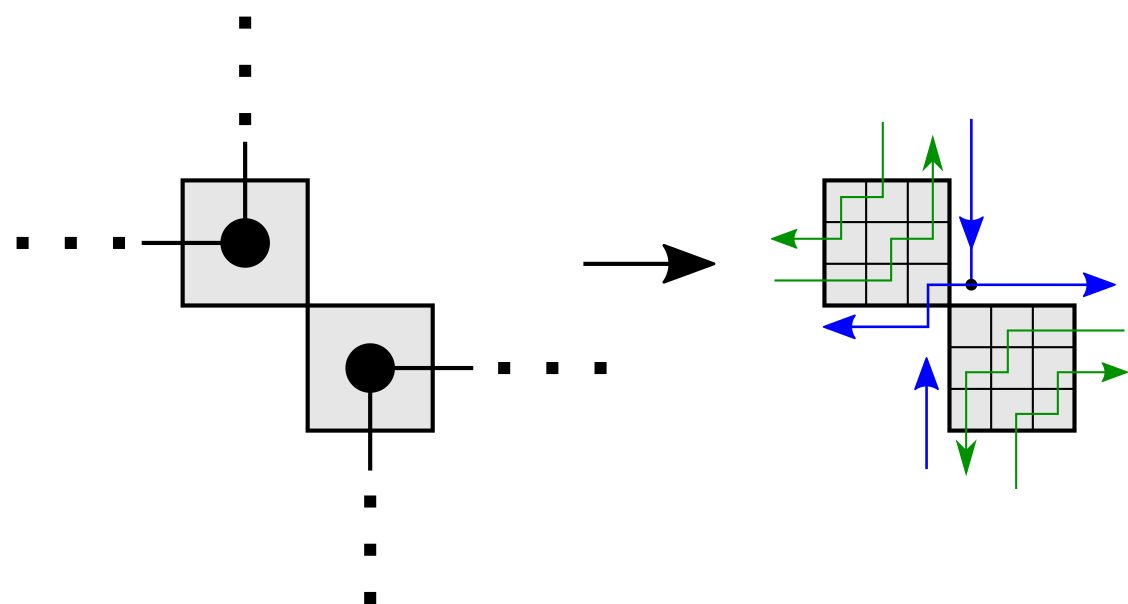}
    \caption{For a non-cyclic cavity, we demonstrate how even in the worst case of corner tiles connected to two neighbors themselves in a non-cyclic corner that a 1 tile wide path exists between the non-cyclic cavity and the adjacent space. This applies to all possible rotations of non-cyclic cavities.}
    \label{fig:S3-5_non_cyclic_corners}
\end{figure}

After the prior steps have finished, all cavities are now cyclic cavities.
This provides us with the property that a single edge can be removed from all cavities, and the vertices which are contained in the cycle remain connected.
For all remaining cyclic cavities, we remove the edge between the vertices surrounding the cyclic cavity adjacent to the northern most vertex of the western most column of the cavity.
From this vertex, we can then define the two vertices which will have their edge removed; for a northwest vertex location of $(i,j)$ we add the vectors $(-1,1)$ and $(0,1)$ to identify the two vertices to separate.
This provides us with the tuple $((x_0, y_0), (x_1, y_1))$ which corresponds to the coordinates of the two vertices.
Of the two vertices, we take $x_0$ as the $x$ coordinate of the two vertices with the smallest value (i.e., the western-most vertex).
We then increment $y_0$ and $y_1$ by 1 successively, removing all edges with the coordinates until either vertex is in the location of another cavity or free space.
When edges are removed from the binding graph such that two cavities are connected, similar to the combination of non-cyclic cavities we combine the two cavities.
We track the values of the vertices defined by $(x_0,y_0)$ and store these values for modifying the supertiles in these locations, we additionally include the vertex at $(-1+i, j)$; this guarantees the 1-tile wide path will reach into the cavity.
We call this tracked set of vertices $M$.
This process is repeated for all cavities until the perimeter is connected to all cavities.
Upon connecting all cavities to the perimeter, we then assign supertiles of Figure~\ref{fig:S3-2_seed_path_macrotiles} to the modified binding graph demonstrated in Figure~\ref{fig:S3-11_combining_cavities}.

\begin{figure}
    \centering
    \includegraphics[width=.8\textwidth]{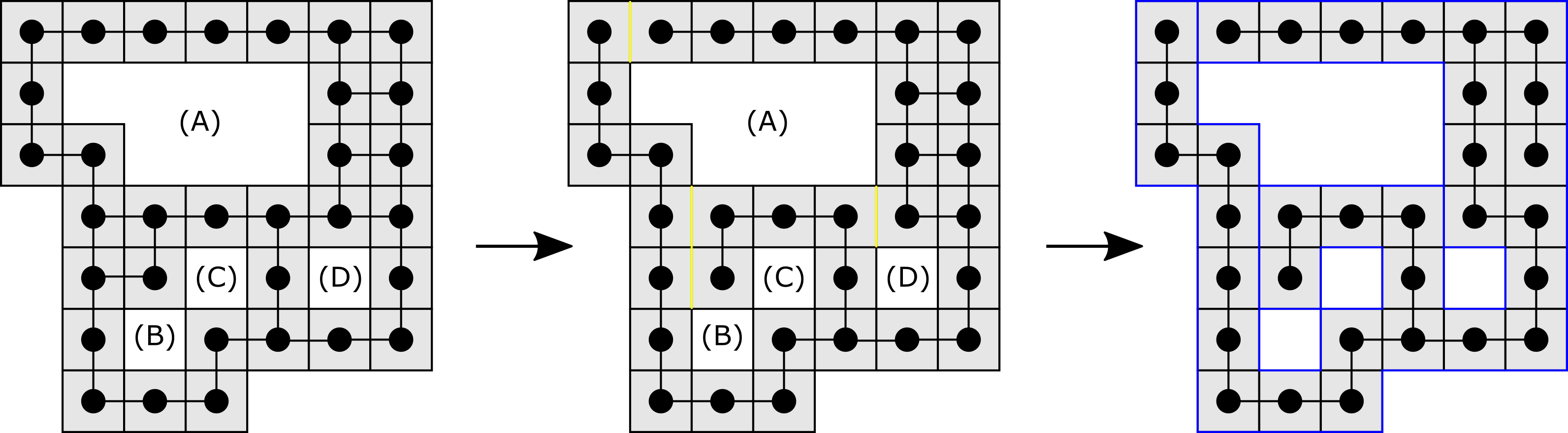}
    \caption{(left) The initial binding graph of a seed which contain cavities. We note that cavities (B) and (C) are, by themselves, non-cyclic cavities. These can be combined to form a larger cyclic cavity. (middle) After being combined, from the northernmost corner of the westernmost edge of each cavity we remove edges of the connection graph along the yellow line, guaranteeing access of the perimeter path to each cavity. We note that both cyclic cavities (B)(C) and (D) are connected to the cavity of (A), which is then connected to the perimeter. (right) A spanning tree is formed from the assignment of the scale 3 supertiles in Figure~\ref{fig:S3-2_seed_path_macrotiles}.}
    \label{fig:S3-11_combining_cavities}    
\end{figure}

\begin{figure}
    \centering
    \includegraphics[width=0.4\textwidth]{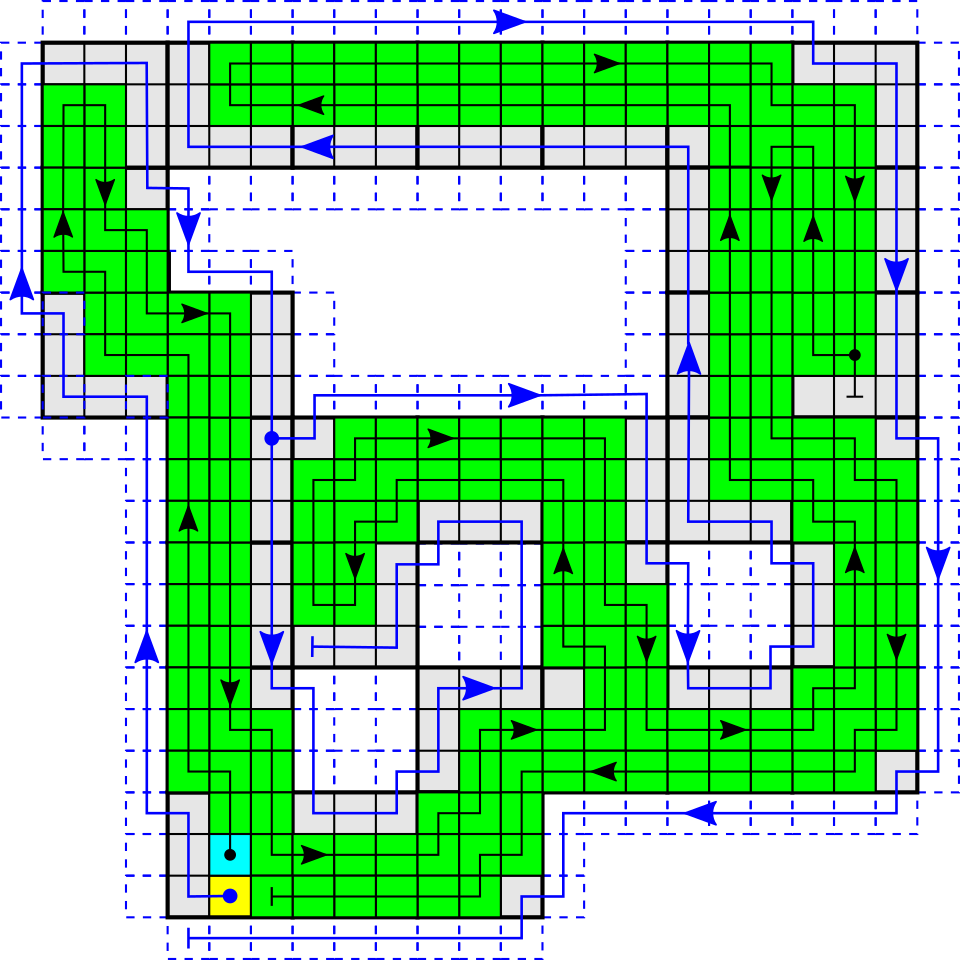}
    \caption{The seed resulting from the addition of supertile templates assigned to each vertex of the modified binding graph from Figure~\ref{fig:S3-11_combining_cavities}. The teal square represents the new single-tile seed for $\calT_3$. Green squares represent tiles of the dependence path, with the black arrows demonstrating the order of tile addition of the dependence path. Blue dashed squares represent locations of cheating fuzz available for the perimeter path. The blue arrow indicates the growth of the tiles of the perimeter path, barring a collision occurring.}
    \label{fig:S3-12_scale_3_example}
\end{figure}

With the modifications of the binding graph, we additionally must assign new supertiles which allow for the perimeter path to take advantage of the new binding graph structure.
The deletion of the edges from each cyclic cavity provides opportunity for a 1 tile wide path to exist in two manners.
Either the dependence path exists in a north-south direction, following the removed edges, or the dependence path enters from the west and terminates at a vertex which was adjacent to the deleted edges.
In the latter case, tile a from Figure~\ref{fig:S3-2_seed_path_macrotiles} is utilized in the presented rotation.
This supertile in its presented state allows for a 1 tile wide path.
If the dependence path follows north-south growth, additional tiles must be utilized which allow for the 1 tile wide path.
We note that for proper dependence path alignment, five new variants of tiles are defined.
The three left supertiles in Figure~\ref{fig:S3-13_n_to_e_conversion_tiles} allow for dependence path utilized by the supertiles in Figure~\ref{fig:S3-2_seed_path_macrotiles} to be shifted; in turn, this facilitates the usage of the right two supertiles which shift the dependence path to the left 2/3 tiles and leaving the necessary 1 tile wide gap for the perimeter path.
Tiles from each set of vertices which describe the 1 tile wide path, $M$, are replaced by the supertiles in Figure~\ref{fig:S3-13_n_to_e_conversion_tiles}.

\begin{figure}
    \centering
    \includegraphics[width=0.5\textwidth]{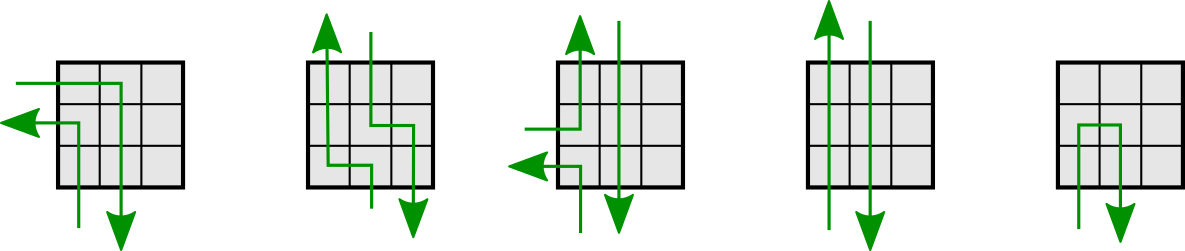}
    \caption{The five cavity connector supertiles which allow for the perimeter path to reach all cavities. The left three tile types allow for the dependence path of the cavity connector tiles to be joined to the dependence path of the remaining seed. The right two tiles shift the dependence path to the west side of the supertiles, guaranteeing that a 1 tile wide path exists to the cavity in question. We note that the tile `a' from Figure~\ref{fig:S3-2_seed_path_macrotiles} may also be utilized in the cavity connector as shown; it also provides a 1 tile wide path for the perimeter path to utilize.}
    \label{fig:S3-13_n_to_e_conversion_tiles}
\end{figure}

The final remaining task for the creation of the seed tiles is assigning the tiles of the dependence path and perimeter path.
As with the scale 4 construction, the dependence path tiles are assigned according to the supertile templates and are assigned unique strength $\tau_3$ glues shared between adjacent tiles of the dependence path.
When the dependence path terminates at the origin supertile, another unique strength $\tau_3$ glue is assigned to the edge between the final tile of the dependence path and the first tile of the perimeter path (yellow tile, Figure~\ref{fig:S3-8_origin_macrotile})
The perimeter path tiles begin by ``following'' the outermost tiles of the dependence path in a clockwise fashion.
The core tiles contain a unique strength $\tau_3 - 1$ glue which is shared with the perimeter path tile adjacent to it, and tiles in the perimeter path exposes a strength 1 glue of type $p$ along the the edge the perimeter tile shared with its predecessor in the perimeter path.
As in the scale 4 simulation, strength $\tau_3$ glues are required for attachment of perimeter path tiles which carry out 90 degree turns.
A key difference between the perimeter paths of scale 3 versus scale 4 is that scale 3 perimeter paths may be required to branch; in particular, this occurs when the perimeter path may need to access cavities which have only a width-1.
This requires that perimeter path tiles may require having two outwards facing $p$ glues; these are denoted by the circle with the surrounding.
Additionally, not all tiles of of the supertile may be immediately defined by either the perimeter path or the dependence path.
This may require small `offshoots' from the dependence path to provide a path to follow; we can see an example of this in Figure~\ref{fig:S3-12_scale_3_example} in the tile type a which is rotated 90 degrees clockwise.
We note that in cases like these, they manifest when type a supertiles from Figure~\ref{fig:S3-2_seed_path_macrotiles} are leafs of a spanning tree adjacent to the perimeter path.


For the remaining simulations of $3 \times 3$ tiles, we define the point of cooperativity/competition as the central tile in a $3 \times 3$ square, as shown in Figure~\ref{fig:S3-7_cooperativity_scale3}.
The process of developing $T_{IO}$ and the modification of $T_\sigma$ to prevent regrowth of the dependence path follows directly from the scale 4 example.
In the case of tiles crashing into the perimeter path, the fuzz growing from the point of cooperation/competition must be able take additional locations, as the perimeter path is guaranteed to be located in any particular location due to the use of cheating fuzz.
These locations have been highlighted in white - tiles which grow into these locations have $p$ glues facing towards the interior of the supertile which the fuzz is growing into.

We investigate the overall tile complexity of $\calT_3$.
For the tile complexity of $\calT_3$, we define the function $C_3(s,t,g) : \mathbb{N} \rightarrow \mathbb{N}$ where $s = |\sigma|$, $t = |T|$ and $g$ is the total number of unique glue/strength combinations in $T$ and $\sigma$.
We consider the worst possible upper bound for each of the items, and thus our tile complexity is a conservative estimate.
This function takes into account the number of tiles utilized by both $T_\sigma$ and $T_{IO}$.

For $T_\sigma$, there exists a variable number of tiles, depending upon the perimeter path taken.
Figure~\ref{fig:S3-2_seed_path_macrotiles} c requires at most 14 tiles for both perimeter path and dependence path - we take that as our high value; $14s$.
We then consider the fuzz required for each edge of a tile which contains an exterior facing glue.
Based upon our assumption high value, the fuzz is already present and reaches the cooperation/competition point.
Comparing with the all other tiles, the most possible fuzz required (Figure~\ref{fig:S3-2_seed_path_macrotiles} a) also requires 14 tiles.
Thus, we have $14s$ tiles required at most for the seed.

Let us consider the tiles of $T_{IO}$.
For output glues, fuzz can be shared between different tiles; it is the tile at the point of cooperation/competition which causes the output to be differentiated.
Glue has an input and output variant, and we assume this is present for all 4 sides.
4 tiles of fuzz are required per side to grow to adjacent points of cooperation/competition, leading to $2 \times 4 \times 4 \times g = 16g$.
Finally, we consider the $T_{IO}$ expansion.
From \cite{Versus}, the size of the minimal glue set is at most 6 (4 choose 2) for a tile at maximum, thus we require at most 6 tiles for each tile in $t$ for our expansion.
We additionally include tiles of the seed in this count - leading to an additional $6t + 6s$ tiles.

The conservative estimate for the tile complexity is $C_3 = 14s + 16g + 6t + 6s = 20s + 16g + 6t$, and it follows that $|T_3| \leq 20s + 16g + 6t$.

Thus, $\calT_3$ correctly seed-first-simulates an arbitrary aTAM system $\calT$ at scale factor 3 requiring the use of cheating fuzz using at most $|T_3| \leq 20s + 16g + 6t$ tile types, and so Theorem~\ref{thm:scale3-sim} is proven.
\end{proof}



\subsection{Technical details of the optimal tile complexity for seed-first-simulation}\label{sec:scaled-sim-append}

In this section we provide the proof and technical details for Theorem \ref{thm:scaled-simulation}.

First, we define a few structures and auxiliary tile sets that will be useful in our construction.

In \cite{IUSA}, a tile set $U$ was given that is intrinsically universal for the aTAM. This means that, given an arbitrary aTAM system $\calT = (T,\sigma,\tau)$, there is a function that specifies how the tiles of $U$ can be arranged to form a seed structure, say $\sigma_\calT$, and another function that returns a representation function $R$ mapping macrotiles of over $U$ to tiles of $T$, so that the resulting system correctly simulates $\calT$ (at temperature $2$, regardless of the value of $\tau$ in $\calT$) using $R$. In that construction, the macrotiles (a.k.a. supertiles) use a well-defined format for their sides, called \emph{supersides}. The format used in \cite{IUSA} is shown on the top in Figure \ref{fig:supersides}. For our construction, we will slightly modify the encoding of supersides by adding ``spacer'' sections, which are potentially very large sections that are simply passed over by the gadgets of the construction and used merely to allow for correct alignment and a larger scale factor that is driven by the space required to simulate a relatively complex Turing machine. In our construction, we will utilize the superside encoding on the bottom, and it will frequently be split in half into portions we'll refer to as $g_l$ and $g_r$. For a given $\calT$, all segments will have constant encodings across all macrotiles except for the ``glue'' sections which will consist of the binary encoding of the specific glue that each superside is simulating. To correctly utilize the superside encoding to perform simulations following the procedure of \cite{IUSA}, our construction will use a slightly modified version of their tile set $U$, which we will call $U'$. The only differences with $U'$ will be its ability to skip over the (arbitrarily long) spacer sections (effectively just absorbing the larger scale factor imparted by them, and a slight change to the tiles that form the left end of a so-called ``frame'' row (which will be discussed when the tiles for growing the ``activation row'' are explained.

\begin{proof}
We prove Theorem~\ref{thm:scaled-simulation} by construction. Therefore, let $P_\calT$ be a program that outputs $\calT = (T,\sigma,\tau)$, an arbitrary aTAM system.
We will define $\mathcal{S}_\calT = (T_\calT,\sigma_\calT,2)$ such that $\sigma_\calT$ consists of a single tile and $\mathcal{S}_\calT$ seed-first-simulates $\calT$ at some scale factor $c \in \mathbb{Z}^+$ while using $O(\frac{|P_\calT|}{\log{|P_\calT|}})$ tile types.

First, we will use the technique of \cite{AdChGoHu01} to compress $P_\calT$, allowing for the use of an optimal number of tile types.
We note that the binary string $P_\calT$ can be encoded in a higher base, which we'll refer to as $b$.
This allows for the length of the compressed encoding to be reduced to $\frac{|P_\calT|}{\log(b)}$. Tiles can be designed so that from the seed tile a row of $\frac{|P_\calT|}{\log(b)}$ tiles assemble with that compressed encoding on their north sides (requiring $\frac{|P_\calT|}{\log(b)}$ unique tile types).
Following the construction of \cite{AdChGoHu01}, tiles of $O(b)$ types can then assemble to the north of that row to ``decompress'' the string, resulting in the encoding of $P_\calT$ on the northern glues of the tiles of their northern row.
Note that, in order to correctly integrate with the tiles to be described later, these decompression tiles also perform binary counting to count the number of rows that they assemble during the decompression. (This only increases the tile complexity by a constant multiplicative factor, thus the decompression still only requires $O(b)$ tile types.)
The subest of tile types described so far will be referred to as $T_\sigma$ later in this proof.
At this point, the tile complexity is $\frac{|P_\calT|}{\log(b)} + O(b)$.  As discussed in \cite{AdChGoHu01}, the base $b$ can be selected such that this quantity is equal to $O(\frac{|P_\calT|}{\log(|P_\calT))})$. Finally, we note that the length of the shortest program $P_\calT$ which outputs $\calT$ is the Kolmogorov complexity of $\calT$, denoted $K(\calT)$. Thus, this portion of the construction requires $O(\frac{K(\calT)}{\log(K(\calT)})$ tile types. As will be shown, all other tile types used are from a constant tile set, independent of the system $\calT$, and thus the overall tile complexity is $O(\frac{K(\calT)}{\log(K(\calT)})$. This is the information theoretic lower bound for an aTAM tile set representing $\calT$.

We now define the rest of the construction.
Let $\sigma^2$ be a mapping of each point of $\sigma$ to a $2 \times 2$ square (i.e. the points of $\sigma$ at scale factor 2).
Let $H'$ be a Hamiltonian cycle through $\sigma^2$ and consider it to begin and end in the leftmost of the bottom-most $2 \times 2$ square of $\sigma^2$, starting in its top-left and ending in its bottom-left. (See Figure \ref{fig:ham-cycle-encoding} for an example.)

\begin{figure}
    \centering
    \includegraphics[width=1.0\textwidth]{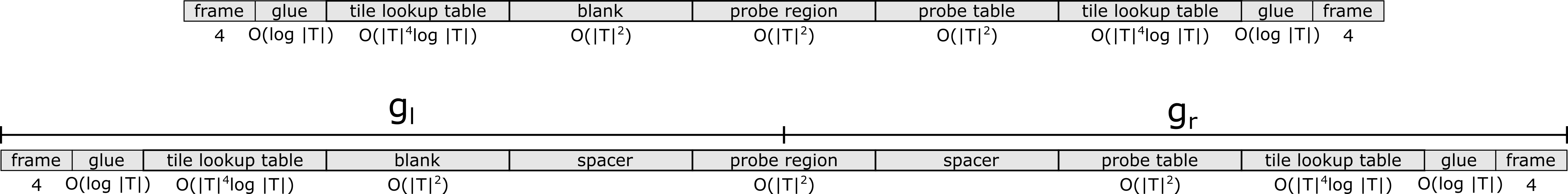}
    \caption{(Top) Encoding of the information presented on a superside in the intrinsically universal construction of \cite{IUSA}. The length of each section is listed below it, with respect to the tile set $T$ that is being simulated, (Bottom) Modified superside encoding for the construction in Theorem \ref{thm:scaled-simulation}. Note the additional ``spacer'' sections (whose sizes will be related to the running time of Turing machine $D$), and the division into halves called ``$g_l$'' and ``$g_r$''.}
    \label{fig:supersides}
\end{figure}

Define Turing machine $D$ that uses a one-way-infinite-to-the-right tape, takes as input program $P_\calT$, and does the following:
\begin{enumerate}
    \item $D$ moves its head across the input from left-to-right then back from right-to-left. (This is merely a technicality to guarantee that the runtime is longer than the output, a fact which will be utilized later.)
    
    \item $D$ runs $P$ to obtain the string encoding the definition of $\calT$, which we will refer to as $\langle \calT \rangle$.
    
    \item $D$ uses $\langle \calT \rangle$ to run the algorithm of \cite{IUSA} and compute the modified encoding of a superside for a macrotile simulating $\calT$ (as shown in Figure \ref{fig:supersides}). In the base encoding, the spacer sections will be a single tile wide and the glue encoding will be of the \emph{null} glue. This encoding will be split into a left half and a right half, which we will call $g_l$ and $g_r$, respectively.

    \item $D$ uses the definition of the seed $\sigma$ from $\langle \calT \rangle$ to compute the set of points contained within $\sigma^2$, which we will refer to as $\langle \sigma^2 \rangle$. (See Figure \ref{fig:ham-cycle-encoding} for an example.)
    
    \item $D$ uses $\langle \sigma^2 \rangle$ to compute the Hamiltonian cycle $H'$. (See Figure \ref{fig:ham-cycle-encoding} for an example.)
    
    \item $D$ creates a string, which we'll call $H$, that contains a section encoding each location on the cycle $H'$. The full string $H$ is a concatenation of one section for each point on the path. The encoding for each point is the following:
    
        \begin{enumerate}
            \item Each contains 3 portions, labeled $F$, $L$, and $R$ for ``forward'', ``left'', and ``right'', respectively.
            
            \item For all sections other than the one for the final point, exactly one of the portions will be marked with either a ``$^*$'' or a ``\#'', indicating that that is the direction of travel to the next point. (For example ``$F^*$'' means the path continues forward in the same direction as from the previous point, ``$L^*$'' means it turns left, and ``$R^*$'' means it turns right.) The difference between the ``$^*$'' or a ``\#'' markings are that ``$^*$'' indicates that the next location is within the same $2 \times 2$ macrotile as the current location, and ``\#'' indicates that the location is in a different, neighboring $2 \times 2$ macrotile.
            
            \item Each of the two remaining un-marked portions will either have (i) the symbol ``\_'' to indicate that its side is an interior side of $\sigma^2$ (i.e. it is adjacent to another block in $\sigma^2$) and thus no information needs to be output to that side, or (ii) if the side is an exterior side of $\sigma^2$ (i.e. it is on the perimeter of $\sigma^2$, whether that is the external perimeter or the perimeter of a cavity enclosed by $\sigma$) one of the strings $g_l$ or $g_r$, with either the \emph{null} glue or the encoding of a glue type in the portion of $g_l$ or $g_r$ labeled ``glue'' in Figure \ref{fig:supersides}. (Figure \ref{fig:ham-cycle-encoding} shows a shorthand encoding of the first 7 locations of the path. Rather than giving a full encoding of each of $g_l$ or $g_r$, it simply notes which is used for that location and gives the name of the glue that is encoded within it, or ``\_'' if it is the \emph{null} glue.) Traveling clockwise around the exterior of a $2 \times 2$ square of $\sigma^2$, the counterclockwise most of the two edges making up each side receives the $g_l$ string, and the clockwise most receives the $g_r$ string. Additionally, there are special markings that may be added to the end of the $g_r$ encoding as follows:
                \begin{enumerate}
                    \item If the right side of the $g_r$ encoding will be at the end of a convex corner, the symbol `c' is added to the rightmost glue of $g_r$.
                    
                    \item If the right side of the $g_r$ encoding will be at the boundary between two different $2 \times 2$ blocks of macrotiles but not at a concave corner, the symbol '4' is added to the rightmost glue of $g_r$.
                    
                    \item If the right side of the $g_r$ encoding will be at the boundary between two different $2 \times 2$ blocks of macrotiles at a concave corner, the symbol `7' is added to the rightmost glue of $g_r$.
                \end{enumerate}
        \end{enumerate}
        
        Additionally, another type of special marker is added to $H$ as follows:
        
        \begin{enumerate}
            \item In the encoding of the last location in $H$, for the portion labeled $L$, the leftmost symbol of the encoding of $g_r$ is given the additional special marker `A'. This will be used to eventually trigger the initiation of the pre-activation row which will begin its growth around the exterior of the seed-representing assembly after this final macrotile of $H$ has grown.
            
            \item If the seed $\sigma$ contains any internal cavities that are completely enclosed by $\sigma$, special handling for those cavities is required. For each such cavity, $D$ determines which side of which macrotile will be the final to be grown along it. Let $d \in \{F,L,R\}$ be the direction of that side relative to the growth of that macrotile. Let $g_d \in \{g_l,g_r\}$ be the string that is to be placed on the $d$ side of the macrotile. If $g_d = g_l$, then the rightmost symbol is given the additional special marker `A'. Else, if $g_d = g_r$, then the leftmost symbol is given the additional special marker `A'. This will be used to eventually trigger the initiation of the pre-activation row which will begin its growth around the the interior cavity to which it is adjacent after this final macrotile of on the perimeter of that cavity has grown.
        \end{enumerate}
        
    \item $D$ outputs $H$ and halts.
\end{enumerate}

\begin{figure}
    \centering
    \includegraphics[width=4.0in]{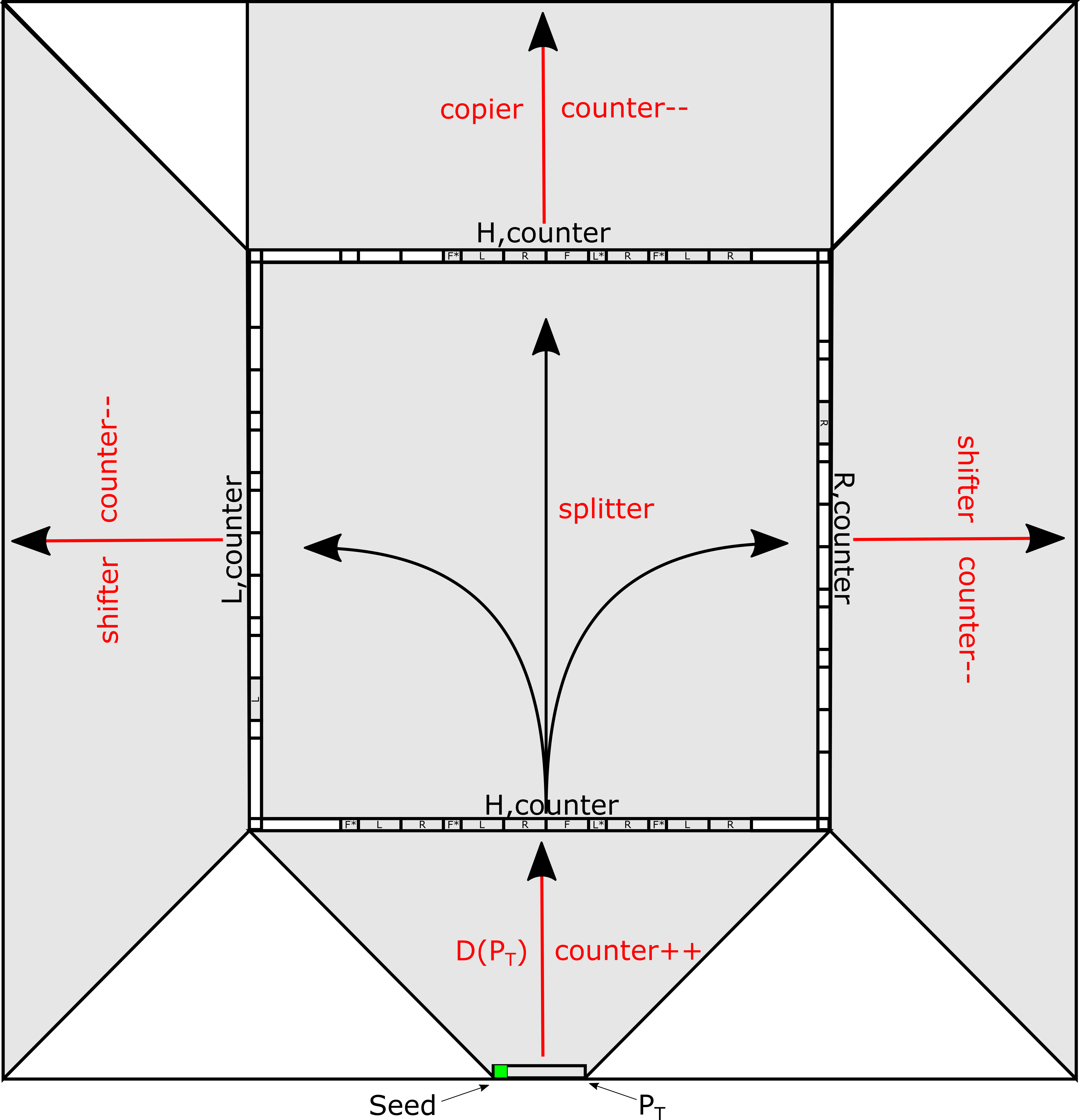}
    \caption{Growth of the first seed macrotile, representing one location of a $2 \times 2$ block in $\sigma^2$. Note that this is a generic example that doesn't match with the first macrotile of the example from Figure \ref{fig:ham-cycle-encoding}. For an example which does match, and depicts the first two macrotiles, see Figure \ref{fig:two-optimal-macrotiles}.}
    \label{fig:optimal-macrotile-growth}
\end{figure}

\begin{figure}
    \centering
    \includegraphics[width=4.0in]{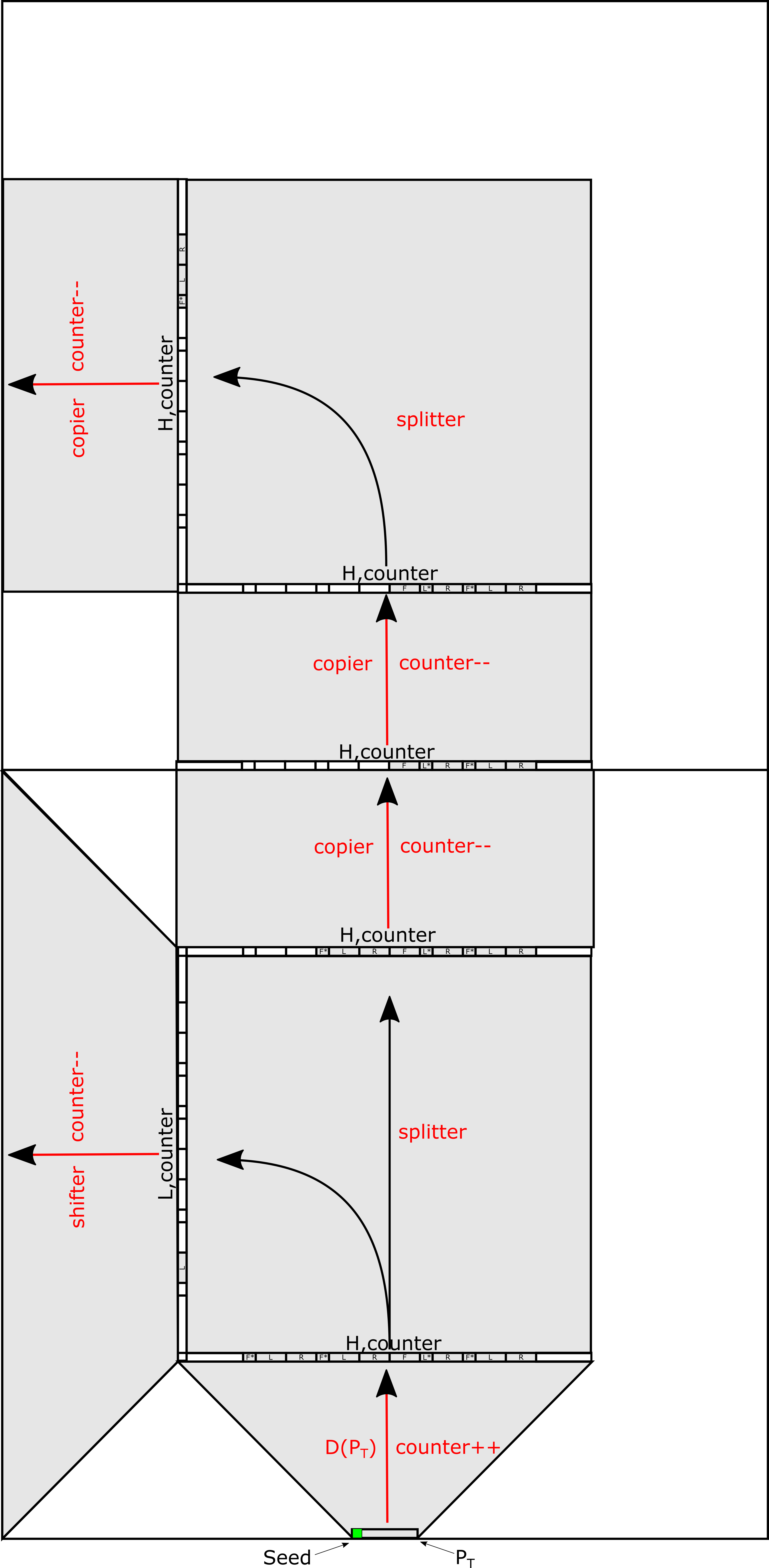}
    \caption{Example of the first two macrotiles from the example in Figure \ref{fig:ham-cycle-encoding}.}
    \label{fig:two-optimal-macrotiles}
\end{figure}

The tile set $T_\calT$ can now be constructed as follows. $T_\calT$ is the union of the following tile sets:
\begin{enumerate}
    \item The previously defined tile set $U'$, which is a modification of the tile set $U$ from \cite{IUSA}.
    
    \item The tile set $T_\sigma$, which was previously discussed, is hard-coded to grow from the seed tile a row whose northern glues encode the compressed version of $P_\calT$, and then decompress the bits of $P_\calT$ (the program which outputs $\calT$) in binary. It will also display a binary number encoding the number of rows which have grown upward to that point. The rightmost tile of the initial row will be the seed tile, $\sigma_\calT$ of $\mathcal{S}_\calT$. The glues between the tiles of this initial row will be strength-2, allowing the full row to grow directly from the seed tile. Additionally, the northern glue of the leftmost tile of the northern row will not only encode the first bit of $P_\calT$, but it will also encode the start state of Turing machine $D$ and be a strength-2 glue.
    
    \item The tile set $T_D$, which is a tile set that simulates Turing machine $D$ in a standard zig-zag manner (see \cite{DirectedNotIU}, among many others, for examples of zig-zag Turing machine simulations). Additionally, (1) each row of the simulation of $D$ grows one additional tile to each of the left and right sides, and (2) the bits of a binary counter are overlaid on the tiles of the simulation, with the least significant bit being on the leftmost tape cell (although each row extends by one additional tile to both the left and right ends, the Turing machine $D$'s tape has a fixed left end which never moves) and the bits extending to the right. This counter starts from the value encoded by the binary counter of $T_\sigma$ and is incremented during the growth of every row so that it effectively keeps a count of the runtime of $D$.
    
    \item The tile set $T_{splitter}$. (There are four rotated versions of these tiles, so we will explain the version that grows from north to south. The other three versions simply consist of rotated versions of the described tile types.) These tiles grow on top of row which encodes (from left to right) an arbitrary number of ``blank encoding'' tiles, followed by an encoding of $H$ overlaid with the bits of a binary counter, followed by an arbitrary number of ``blank encoding'' tiles. The grow in a zig-zag manner and cause portions of the information in the input row to be rotated to each of the left and right, as well as propagated northward. These tiles only copy a subset of the input row to be propagated in each direction. (See Figure \ref{fig:optimal-macrotile-growth} for an example of the ``splitter'' in the middle square.) The subset for each direction is determined by the leftmost section of the encoding of $H$, which contains a portion for each direction ($F$, $L$, and $R$).
    
        \begin{enumerate}
            \item For the direction marked with either a ``$^*$'' or a ``\#'', the full encoding of $H$ is copied to the side minus the leftmost section (i.e. the leftmost $F$, $L$, and $R$ portions are not copied, but instead replaced by ``blank'' values). This serves to erase that portion from $H$ as it is passed to the next block, so that the leftmost remaining portion will encode the correct instructions for that block. If the erased direction was marked with ``\#'', that symbol is still included on the left side of $H$ to signal to the next component (the copier) that the information is being copied across the boundary of one $2 \times 2$ macrotile into another. Additionally, the overlaid bits of the binary counter are copied in this direction.
            
            \item For any direction which consist of either $g_l$ or $g_r$ (with the appropriate glue encodings), that string is rotated to that side, along with the overlaid bits of the binary counter. The rest of the positions encode ``blank'' values.
            
            \item For any direction consisting of just the blank, ``\_'', nothing is rotated to that side, and that side of the macrotile is not grown (since it is interior to the $2 \times 2$ block of $\sigma^2$ and thus can't be used to initiate any additional growth).
        \end{enumerate}
        
\begin{figure}
    \centering
    \includegraphics[width=5.0in]{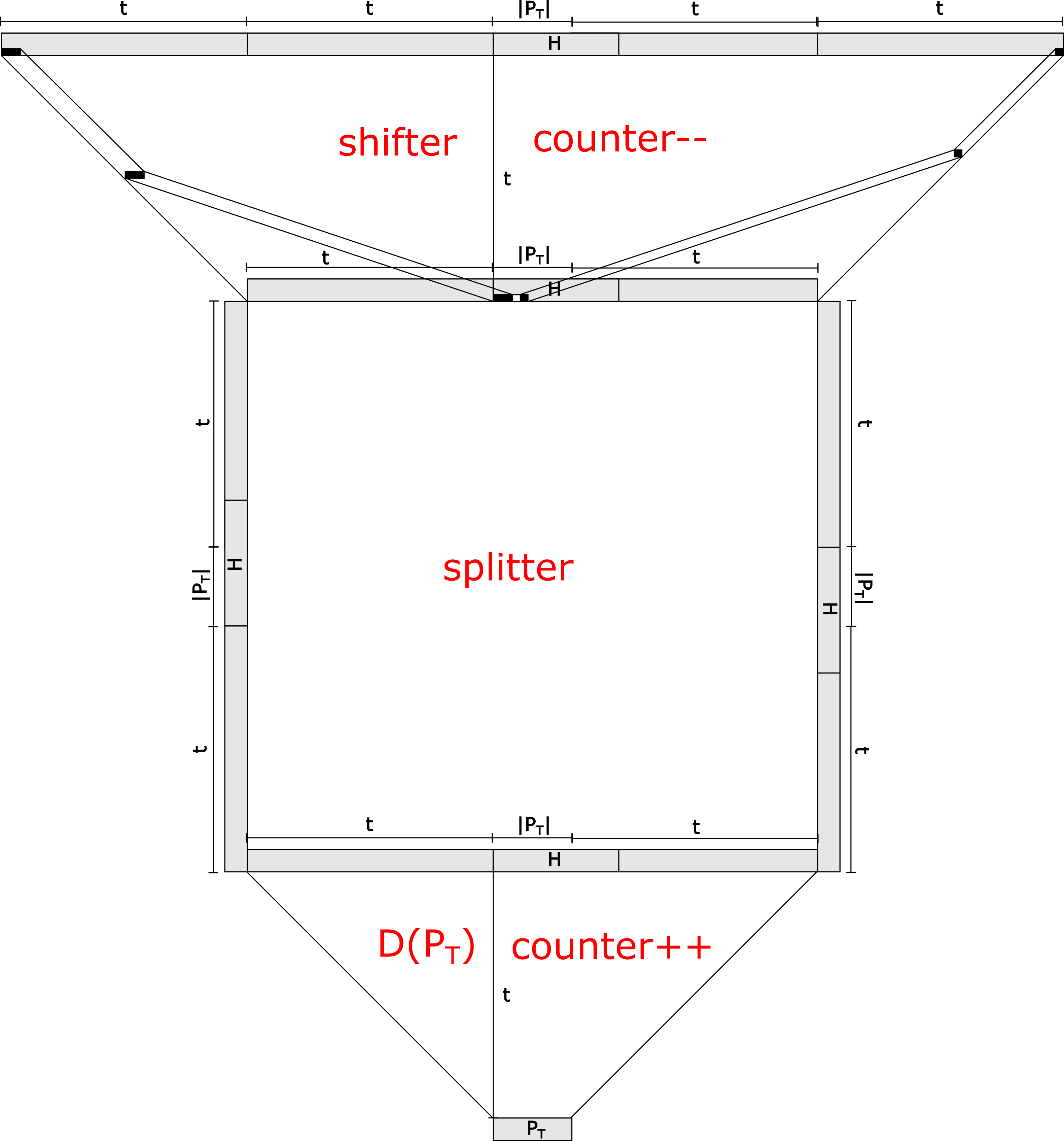}
    \caption{A schematic depiction of dimensions of portions of the macrotiles in $\mathcal{S}$ and how the shifting of the information from the encoding of either $g_l$ or $g_r$ occurs.}
    \label{fig:optimal-macrotile-dimensions}
\end{figure}

    \item The tile set $T_{shifter}$. (There are four rotated versions of these tiles, so we will explain the version that grows from north to south. The other three versions simply consist of rotated versions of the described tile types.) These tiles grow on top of row which encodes (from left to right) an arbitrary number of ``blank encoding'' tiles, followed by an encoding of either $g_l$ or $g_r$, followed by an arbitrary number of ``blank encoding'' tiles, with a subset of the tiles overlaid with the bits of a binary counter. These tiles cause zig-zag rows to grow which decrement a copy of the binary counter value until it reaches the value 0 
    , extend each row by an additional tile on each of the left and right sides, and copy the encoding of $g_l$ or $g_r$ upward while shifting its two halves to either end of the final row. As mentioned previously, and shown in the bottom of Figure \ref{fig:supersides}, each of $g_l$ and $g_r$ have a ``spacer'' section which is initially a side tile wide. As these zig-zag rows grow upward, they shift the information on the left of the spacer to the left, and the information to the right of the spacer to the right. Since rows alternate direction of growth, they alternate the direction in which they can shift information. Initially, the information which is being shifted in the same direction as the growth of a row is shifted by 6 positions. This is possible by having glues that encode up to the last 6 ``remembered'' values at a time. Once the information being shifted reaches the end of the row, the shifting switches to 2 positions per shift so that the information remains at the end of the row until arriving at the final row. (See Figure \ref{fig:optimal-macrotile-dimensions} for a schematic depiction of the dimensions of a macortile and how the encodings of $g_l$ and $g_r$ are shifted.) The reason for initially shifting 6 spaces at a time is that, since shifting occurs every second row for each direction, this provides a slope for the information being shifted of $2/6$, i.e. $1/3$. The dimensions of the macrotiles are largely dictated by dictated by the runtime of $D(\calT)$, which is denoted by the value $t$ in Figure \ref{fig:optimal-macrotile-dimensions}. Since $D$ is designed so that its runtime must always be greater than the length of its output, the length of $H$ must always be less than $t$ (by - significantly - more than 2). Therefore, the maximum amount that any information could need to be shifted is bounded above by $3t$ (see Figure \ref{fig:optimal-macrotile-dimensions}). Since the shifter grows upward a distance of $t-2$ rows, that means a slope of $1/3$ must always suffice to shift the information to the correct information before reaching the top row. Once reaching the edge of the row, the information can then be shifted at a slope of $1/1$ for the remaining rows to stay at the end of the rows. The tiles that perform the shifting can easily accommodate this by having a special marking on the tiles at the end of the rows to detect when the shifting information overlaps that mark.
    
    \item The tile set $T_{copier}$. (There are four rotated versions of these tiles, so we will explain the version that grows from north to south. The other three versions simply consist of rotated versions of the described tile types.)  These tiles grow on top of row which encodes (from left to right) an arbitrary number of ``blank encoding'' tiles, followed by an encoding of $H$ overlaid with the bits of a binary counter, followed by an arbitrary number of ``blank encoding'' tiles. There may also be a ``\#'' symbol on the left end of the encoding of $H$. These tiles cause zig-zag rows to grow which simply decrement a copy of the binary counter value until it reaches 0, while also copying the original value of the binary counter as well as the encoding of $H$ upward. If a ``\#'' symbol is on the left end of $H$, an additional 4 rows are grown upward, copying all of the same information but erasing that ``\#'' symbol. (This accounts for growth from one $2 \times 2$ macrotile into another and the spacing required to account for the pre-activation and activation rows that are between each pair of macrotiles where at least one is a non-seed macrotile. These will be described later.) Since this copying behavior will always need to be performed twice in a row (except for the possible spacing of the extra 4 rows), a special marking is also copied upward through the first occurrence of the copier so that a second occurrence of the copier is initiated. (An example can be seen at the top of the first macrotile and bottom of the second macrotile in Figure \ref{fig:two-optimal-macrotiles}.)
    
\begin{figure}
    \centering
    \includegraphics[width=0.8\textwidth]{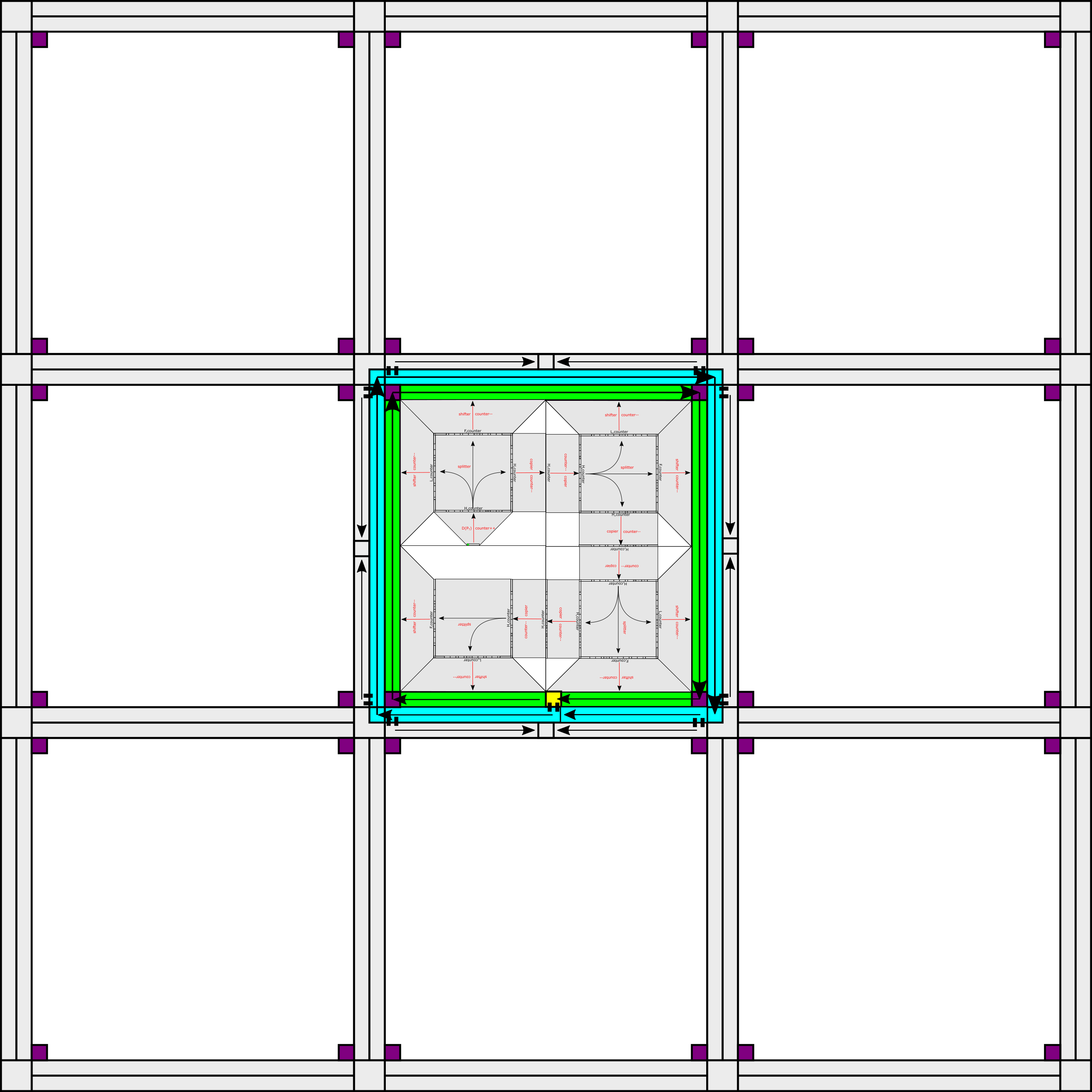}
    \caption{Depiction of the growth of the pre-activation (green) and activation (blue) rows around a simple assembly of seed macrotiles representing a single-tile seed $\sigma$, and the locations of possible future surrounding macrotiles outside of the representation of $\sigma$. Note that the rows showing the locations of pre-activation and activation rows, and between surrounding macrotile locations, are not remotely to scale but shown much larger for clarity, as the scale factor is very large and the size of rows is significantly smaller in the depicted macrotiles in the center.}
    \label{fig:scaled-simulation-seed-macrotile-frame}
\end{figure}

    \item The tile set $T_{act}$. This tile set creates two rows that grow around a perimeter of the assembly consisting of the seed-representing macrotiles. This includes the external perimeter around the entire seed-representing assembly, as well as the perimeters of any internal cavities (if there are any). The first row to grow is called the \emph{pre-activation} row, and begins growth when the special marker `A' makes it to the exterior row of a macrotile tile.
    This special marker causes the tile that is placed there 
    to include an exterior facing strength 2 glue that initiates the growth of the pre-activation row that grows in a clockwise direction around the perimeter of which that side is a part. (This is shown as the green row in Figure \ref{fig:scaled-simulation-seed-macrotile-frame}.)  The tiles of this row (and the activation row which will piggyback on it) grow via cooperation between a glue on the side of the tile most recently placed in the path facing the direction of growth, and a glue on the exterior of the macrotile over which it is growing. However, there are a few special cases in which such cooperation is not possible during the growth of the pre-activation row. Each of these is detected by the growing row due to a special glue marking on the rightmost glue of a $g_r$ encoding. They are as follows:
    
        \begin{enumerate}
            \item Convex corner:  Since such cooperation isn't possible when the row needs to grow around a convex corner, whenever the row reaches the end of the encoding of a copy of $g_r$ and the rightmost glue is marked with a `c', this signifies that it is at a convex corner, so a tile is placed with a strength 2 glue allowing for the placement of a tile that facilitates turning the corner. (Examples can be seen as purple tiles in the pre-activation row in Figure \ref{fig:scaled-simulation-seed-macrotile-frame}.)
            
            \item Straight gap between macrotiles of different $2 \times 2$ blocks: Since there will be an additional spacing of 4 tiles between such macrotiles, the rightmost glue of $g_r$ will have been marked with a '4', and this will cause a pre-activation row tile to attach that is able to initiate the growth of a set of 4 tiles that can grow across the gap. They can be generic tiles, unspecific to the encodings of $g_l$ or $g_r$ because they will be in locations don't need to propagate any such information since they will be interfacing with the tiles from the IU construction of \cite{IUSA} and in the corresponding locations there are generic `frame' tiles.
            
            \item Concave corner between macrotiles of different $2 \times 2$ blocks: Since there will be a spacing requiring 7 non-cooperative tiles between such macrotiles to go around the corner crossing the two 4-tile-wide gaps, the rightmost glue of $g_r$ will have been marked with a '7', and this will cause a pre-activation row tile to attach that is able to initiate the growth of a set of 7 tiles that can grow around that corner, across the gaps. As before, these can be generic tiles.
        \end{enumerate}
    
    This pre-activation row is guaranteed to be able to successfully complete for each perimeter on which it grows, even if it has to temporarily pause progress as some output sides generated by shifter modules may need extra time to complete. However, all macrotiles are guaranteed to already be in place and to have initiated their outputs to the point of supporting growth of the pre-activation row, but none of them can initiate any growth beyond that point, so there is no growth that can interfere with the completion of the pre-activation row. The tiles of the pre-activation row copy the information from the exterior of the macrotiles to their exterior sides, resulting in the perimeter of each pair of adjacent macrotiles in the same $2 \times 2$ square containing a complete superside representation in the form shown in Figure ~\ref{fig:supersides}, which is equivalent to that of the IU construction in \cite{IUSA} with the addition of the two extra ``spacer'' segments. This is because the two adjacent macrotiles will expose the encodings of $g_l$ and $g_r$ so that they are concatenated. In this way, each $2 \times 2$ block of seed-representing macrotiles in $\mathcal{S}_\calT$ will represent a single tile of the seed in $\calT$.
    
    When a pre-activation row completes by placing the final tile, which is adjacent to the first tile, 
    (the final tile of the external perimeter is shown as yellow in Figure \ref{fig:scaled-simulation-seed-macrotile-frame}) that final tile presents a strength 2 glue outward to initiate growth of the \emph{activation} row. The activation row also grows in a clockwise manner, on the outside of the pre-activation row. The activation row is the final row to grow around the edges of the $2 \times 2$ blocks of macrotiles that each represent single tiles of the seed of $\calT$. This row is always able to cooperate using a tile from the pre-activation row except for the case when it goes around a convex corner. In this case, the glue of the tile in the pre-activation row at that corner signals the upcoming corner and a tile attaches with a strength 2 glue, allowing the activation row to grow around the corner.
    
    As each segment completes, it exposes the glues of a superside that, modulo the spacing regions that are ignored, are identical to those used in the construction of \cite{IUSA}. What this means is that, at that point, each superside can behave identically to the simulated side of a tile in the intrinsically universal simulator. All of the information that needs to be propagated to newly growing adjacent supertiles is presented, and from the exterior (other than the larger scale factor accommodated by the spacer sections), the supertiles representing the seed structure are identical to those used by that simulator. Because of this, we call this the activation row since as soon as it grows across a supertile side, if that supertile simulates a tile with a glue exposed on the corresponding side, growth can now proceed outward, out of the representation of the seed into adjacent supertile locations. Thus, the supertile has been ``activated'' to begin simulation of growth beyond the seed assembly.
    
    Since no growth outside of the macrotiles representing the seed can occur before the activation row grows over them, and the activation row cannot even begin until the entire pre-activation row is complete, which can only occur once every macrotile representing the seed that is on the exterior of the seed assembly has completed, nothing can interfere with the growth of the assembly representing the seed until the activation row is growing. Therefore, we only need to consider interactions between growth outside of the seed and macrotiles that have already completed pre-activation rows but don't necessarily have completed activation rows. The only potential conflict then occurs if a supertile outside of the representation of the seed, but adjacent to the seed, attempts growth into the location that would have been occupied by the activation row, potentially blocking its further growth. To handle this situation, since we are leveraging the construction of \cite{IUSA} to handle all growth beyond the seed, we make use of the growth pattern used by that construction. The supertiles of that construction have empty $2 \times 2$ squares in the corners between them (which can be seen in Figures \ref{fig:scaled-simulation-seed-macrotile-frame} and \ref{fig:scaled-simulation-seed-macrotile-frame-blocking}). The row in which the activation row would have been attempting to grow but was blocked would instead be filled by a row of the `frame' of the IU construction. A simple modification to the the tiles of such a row from that tile set (specifically the tiles which are the leftmost tiles - from the perspective of the supertile growing that row - of those rows) adds a strength 2 glue that allows a newly designed tile to attach to the left of that tile into the otherwise empty $2 \times 2$ square. In most circumstances, this tile will be completely unused but out of the way. However, in circumstances in which the row to which it is attached blocked formation of the activation row, then this tile will be able to cooperate with a tile from the next segment of the pre-activation row to place a tile of the activation row and begin a new segment of its growth. This can been seen occurring three times in Figure \ref{fig:scaled-simulation-seed-macrotile-frame-blocking}, with correct growth of the activation row being restarted each time. Note that if such an encroaching supertile (as depicted in red in Figure \ref{fig:scaled-simulation-seed-macrotile-frame-blocking}) were to beat the activation row to the right side of the macrotile edge instead of the left side, then all that does is pause growth until that supertile's growth also occupies the left side (which it must eventually do, by design of the frame tiles in \cite{IUSA} which attempt to occupy both ends of such a row and then grow toward each other in the middle, since growth would pause in the middle until that supertile's growth occupies the left side). Therefore, growth of the activation row can be temporarily paused but never stopped.
    
    As each macrotile completes its activation row, any simulation growth which would be valid can commence. Collisions with any inactive portions of the seed do not cause errors, as don't previously mentioned collisions that pause the growth of the activation row, since the seed representation of the seed tile is sufficiently completed by that point to represent the appropriate tile, and there is no need for it to grow output from the colliding side since there is clearly already a supertile in that location.
\end{enumerate}

Although the necessary glue markings aren't always explicitly described above, each module places specially augmented glues in one or more positions (as needed) of its output row that allow the correct next module to begin growth from that output. This is possible because whenever a module completes growth, it is known which module will use that output to begin the next phase of growth.


\begin{figure}
    \centering
    \includegraphics[width=0.8\textwidth]{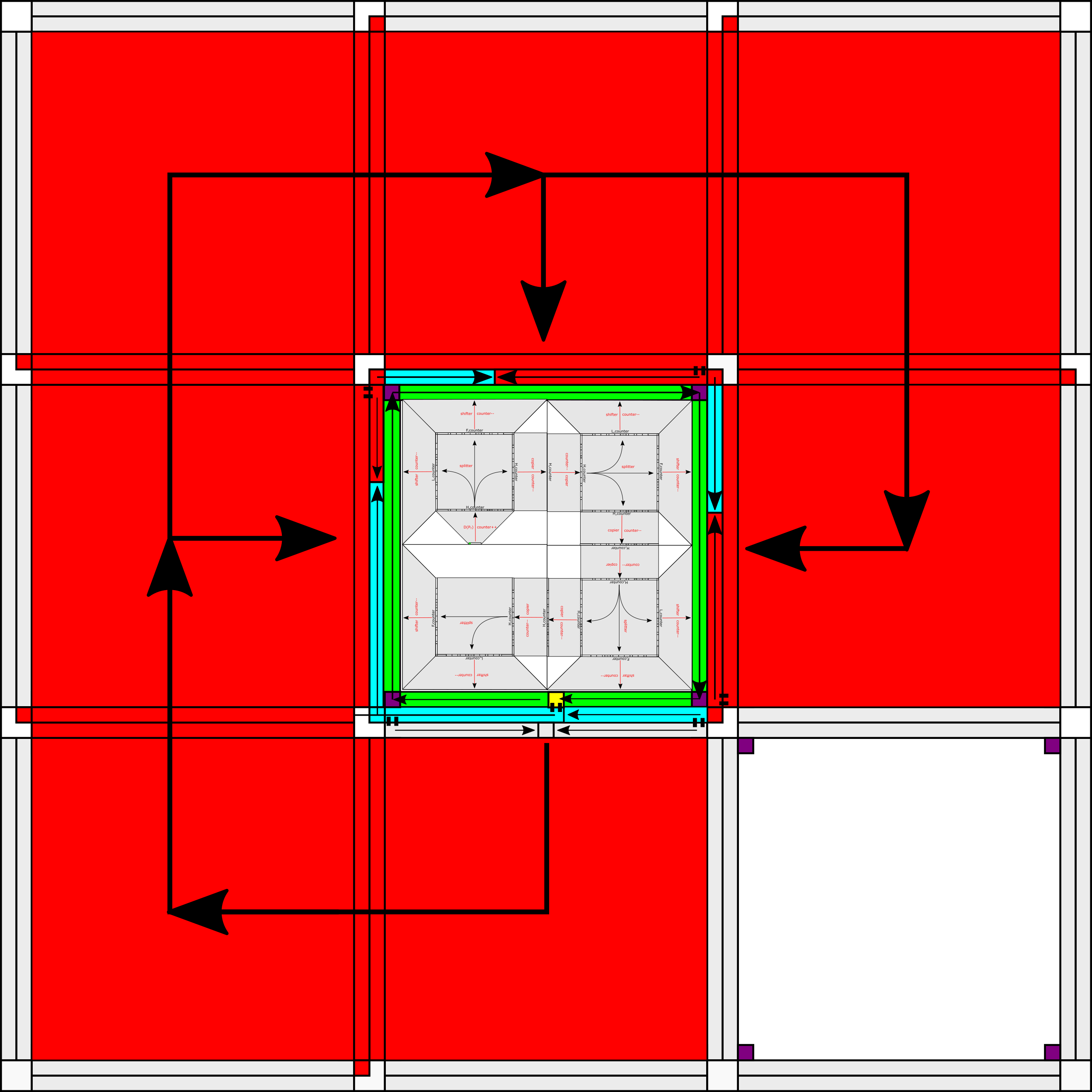}
    \caption{Extending the example of Figure~\ref{fig:scaled-simulation-seed-macrotile-frame}, a depiction of the growth protocol which ensures that, even if the growth of surrounding macrotiles completes and attempt to initiate growth into seed macrotiles locations, the activation of the seed remains correct and no erroneous growth can occur.}
    \label{fig:scaled-simulation-seed-macrotile-frame-blocking}
\end{figure}


With the seed $\sigma_\calT$ and tile set $T_\calT$ defined as above, and $\tau = 2$, the definition of $\mathcal{S}$ is complete. We now provide a sketch of the growth and correctness of the construction.

Starting from a single seed tile, a hard-coded row grows which encodes a program that outputs $\calT$ (potentially the shortest such program). A Turing machine is then simulated by the growth of a set of tiles designed to both simulate that particular Turing machine and also to maintain a counter which keetps track of the runtime of that Turing machine. The Turing machine computes a Hamiltonian cycle through a $2 \times 2$ scaled version of the seed, and then computes the necessary encodings to navigate the growth of the assembly along the cycle and also the strings to be encoded along the exposed sides of those microtiles. During the computation, the increase of width, on each side, of the subassembly performing the computation, is equal to the height it grows (which is equivalent to the runtime). Utilizing these dimensions, a few relatively simple sets of modular construction components are able to assemble the macrotiles along the Hamiltonian cycle and place the necessary encodings along the outer perimeter of the seed-representing assembly so that each $2 \times 2$ block of macrotiles is able to map to a single tile in the seed of $\calT$, with the appropriate information encoded on their perimeters for the assembly to eventually initiate growth that appears from the outside as though it is coming from supertiles of the original IS construction \cite{IUSA} (apart for the extra spacing fields), and thus growth outside of the simulated seed structure follows the design of the IU construction, which has been proven to correctly simulate the system whose information is encoded into the supersides.

The representation function $R$ is able to map a $2 \times 2$ square of macrotiles to a tile of $\calT$ once the fourth macrotile begins forming. At that point, $R$ is able to inspect the encodings of $g_l$ and $g_r$ designated for every side, and within them read the glues encoded for those sides, and then map the $2 \times 2$ block of macrotiles to a single tile in the seed of $\calT$. This ensures that all supertiles simulating tiles in the seed of $\calT$ resolve to those tiles before any growth can occur away from the seed (since all tiles will then be represented before the pre-activation and activation rows grow, which are required for growth away for the seed to begin). Therefore, $\mathcal{S}_\calT$ completely grows the seed first, following the requirement for seed-first-simulation. By connecting with the tiles of the IU construction \cite{IUSA}, we can utilize the full construction and proof of it to assert that as growth now proceeds, it correctly simulates $\calT$. Therefore, given an optimally short program that outputs an arbitrary aTAM system $\calT$, $\mathcal{S}_\calT$ seed-first-simulates $\calT$ using an optimal $O(\frac{K(\calT)}{\log(K(\calT))}$ tile types and a single-tile-seed, so Theorem \ref{thm:scaled-simulation} is proven.

\end{proof}

\section*{Acknowledgements}
The authors would like to thank Trent Rogers for helpful comments and creative brainstorming, and credit him for the idea of Claim \ref{clm:cross-model-IU}. Additionally, we thank the anonymous reviewers who gave extremely helpful comments and suggestions that have greatly improved this version of the paper.


\bibliographystyle{plainurl}
\bibliography{tam,extras,experimental_refs}

\end{document}